\documentclass[twoside,11pt]{article}

\usepackage{graphicx,psfrag,epsf}
\usepackage{enumerate}
\usepackage{amsmath}
\usepackage{bbm}
\usepackage{subcaption}
\usepackage{booktabs}
\usepackage{tikz}
\usepackage{lastpage}
\usepackage[bottom]{footmisc}
\usepackage[abbrvbib,preprint]{jmlr2e}

\newcommand{\X}{\mathbf{X}}
\newcommand{\x}{\mathbf{x}}
\newcommand{\z}{\mathbf{z}}
\newcommand{\s}{\mathbf{s}}
\newcommand{\y}{\mathbf{y}}
\newcommand{\R}{\mathbb{R}}
\renewcommand{\P}{\mathbb{P}}
\newcommand{\E}{\mathbb{E}}
\renewcommand{\L}{\mathcal{L}}
\newcommand{\B}{\mathbb{B}}
\newcommand{\I}{\mathcal{I}}

\newcommand{\1}[1]{\mathbbm{1}\{#1\}}

\DeclareMathOperator*{\argmin}{arg\,min}
\newcounter{assumption}
\newenvironment{assumption}
  {\par\noindent
   \refstepcounter{assumption}%
   \textbf{A\theassumption.}~\itshape\ignorespaces}
  {\par\ignorespacesafterend}

\newcounter{assumptionB}
\newenvironment{assumptionB}
  {\par\noindent
   \refstepcounter{assumptionB}%
   \textbf{A\theassumptionB'.}~\itshape\ignorespaces}
  {\par\ignorespacesafterend}

\jmlrheading{21}{2020}{1-\pageref{LastPage}}{10/19; Revised 7/20}{8/20}{19-833}{Hyebin Song, Ran Dai, Garvesh Raskutti, and Rina Foygel Barber}

\ShortHeadings{Convex and non-convex approaches for inference with class-conditional noisy labels}{Song, Dai, Raskutti, and Barber}
\firstpageno{1}

\begin{document}

\title{Convex and Non-Convex Approaches for Statistical Inference with Class-Conditional Noisy Labels}

\author{\name Hyebin Song \email hps5320@psu.edu \\
       \addr Department of Statistics\\
       The Pennsylvania State University\\
       State College, PA 16801, USA
       \AND
       \name Ran Dai \email  ran.dai@unmc.edu \\
       \addr Department of Biostatistics\\
       University of Nebraska Medical Center\\
       Omaha, NE 68198, USA
       \AND
       \name Garvesh Raskutti \email raskutti@stat.wisc.edu \\
       \addr Department of Statistics\\
       University of Wisconsin-Madison\\
       Madison, WI 53706, USA
       \AND
       \name Rina Foygel Barber \email rina@uchicago.edu\\
       \addr Department of Statistics\\
       University of Chicago\\
       Chicago, IL 60637, USA
       }

\editor{Daniela Witten}

\maketitle

\begin{abstract}%   <- trailing '%' for backward compatibility of .sty file
We study the problem of estimation and testing in logistic regression with class-conditional noise in the observed labels, which has an important implication in the Positive-Unlabeled (PU) learning setting. With the key observation that the label noise problem belongs to a special sub-class of generalized linear models (GLM), we discuss convex and non-convex approaches that address this problem. A non-convex approach based on the maximum likelihood estimation produces an estimator with several optimal properties, but a convex approach has an obvious advantage in optimization. We demonstrate that in the low-dimensional setting, both estimators are consistent and asymptotically normal, where the asymptotic variance of the non-convex estimator is smaller than the convex counterpart. We also quantify the efficiency gap which provides insight into when the two methods are comparable. In the high-dimensional setting, we show that both estimation procedures achieve $\ell_2$-consistency at the minimax optimal $\sqrt{s\log p/n}$ rates under mild conditions. Finally, we propose an inference procedure using a de-biasing approach. We validate our theoretical findings through simulations and a real-data example.
\end{abstract}


\section{Introduction}
\label{sec:intro}
Label noise is a common phenomenon in a number of classification applications. For example, label noise occurs when humans are involved in labeling due to inattention or subjectivity \citep{Ipeirotis2010-te,  Smyth1995-at}. Label noise can also come from bad data entry  \citep{Sculley2008-cx} or is sometimes intentionally introduced to protect the privacy of a respondent \citep{Van_den_Hout2002-fi}. Consequently, it is important to investigate how to carry out valid statistical inference in the presence of label noise. 

An important example of the label noise problem includes the Positive-Unlabeled (PU) learning problem, where labeled samples are known to be positive, but unlabeled samples may be either positive or negative. Positive-Unlabeled learning arises in many applications, where obtaining negative responses is more expensive or intractable. One concrete example arises from deep mutational scanning (DMS) data sets in biochemistry \citep{Fowler2014-li}, where a data set consists of functional (positive) variants of a protein, together with unknown functionality (unlabeled) variants from an initial library. Numerous other applications of PU-learning arise \citep[see e.g.,][]{Liu2003-hv, Yang2014-dd, Elkan2008-aq}.

This article addresses the estimation and testing problems of a binary logistic regression model where noise is present in responses. We assume a standard logistic model between the true binary responses and the known features, and a contamination process of the true labels. In particular, we assume that labels are corrupted with asymmetric probabilities based on their values, but those probabilities are not affected by the features. The goal is to estimate, and perform inferences on, the parameter in the logistic model, which parametrizes the relationship between the features and the true labels.

\subsection{Related Work}
There is a substantial literature on the subject of learning with label noise data. Since the random classification noise model was first proposed in \citet{Angluin1988-da}, extensive studies have been conducted to develop algorithms for building a classifier that effectively separates true positive and negative samples from data with label noise, and to establish theoretical guarantees for the proposed classifiers (\citealp{Natarajan2018-su,Li2018-pf}; see also \citealp{Frenay2014-jp} for a comprehensive survey). 

Parameter estimation problems using probabilistic approaches have also been thoroughly studied by a number of authors, where the likelihood-based method with a latent variable model has been the primary approach. Both settings where the noise rates are known \citep{Magder1997-oq,Hausman1998-yd} and unknown \citep{Pepe1992-rm, Bollinger1997-il, Lyles2011-ss} have been considered. In the latter case, an additional validation set, consisting of both true and noisy labels, is assumed to be available in addition to the main data set \citep{Bollinger1997-il,Lyles2011-ss,Pepe1992-rm}, or a semi-parametric approach was used to model the function containing noise rates nonparametrically \citep{Hausman1998-yd}. 
Additionally, a general treatment of the noisy response problem as a variable with measurement errors is found in Chapter 13 of \citet{Carroll2006-si}. In particular \citet{Carroll2006-si} uses a quasi-likelihood method  which in our case is equivalent to the likelihood approach we propose later. 

On the other hand, \citet{Ward2009-nn} considered modeling of presence and absence of species, assuming that the prevalence of the positive examples in the unlabeled set was known a priori. Treating an indicator of true positive and absence of species as a latent variable, the latent variable model was fitted via the EM algorithm. The non-identification problem of the prevalence of the positives in the unlabeled set without any parametric model assumption has also been pointed out in the same paper \citep{Ward2009-nn}, and in the follow-up paper about the estimation of this prevalence parameter \citep{Hastie2013-wq}. 

In all these aforementioned works in the parametric framework, either convergence to a local optimum was established without theoretical guarantees for the obtained estimators being provided, or the maximum likelihood estimator was considered in the theoretical analysis without discussion of the feasibility of obtaining such global optimum in a non-convex problem, all in low-dimensional settings. In contrast, one of the contributions of our paper is that we demonstrate achieving the global maximum is possible with high probability.

In high dimensions, \citet{Song2018-we} studied the estimation problem in PU-learning and a case-control scheme where the prevalence of positives in an unlabeled set was assumed to be known a priori. They proposed an estimation based on the $\ell_1$ penalized likelihood and devised an algorithm for which the estimator converges to a stationary point of the objective function, where the feasible space was constrained to be an intersection of $\ell_1$ and $\ell_2$ balls. They provided a theoretical mean-squared error guarantee for any stationary point of the objective. In this work, we consider a more general non-convex noisy labels problem which includes the PU problem in \citet{Song2018-we} as a special case. Compared to the results in \citet{Song2018-we}, where the mean-squared error guarantees can only apply to the PU problem within a case-control scheme, our results can be applied to provide mean-squared error guarantees for noisy labels models in both prospective and case-control schemes, while removing the $\ell_1$ constraint in their optimization problem. We also study in this paper an estimator based on a convex objective that can serve as a good starting point for the initialization of the non-convex method. Also, compared to \citet{Song2018-we}, where only an estimation problem was studied, both estimation and testing problems have been addressed in this paper in both low and high dimensions. 

\subsection{Our Contributions}

In this paper, we study the parametric estimation and testing problem given observations where labels are observed with noise. One of the consequences of the label noise is that the maximum likelihood objective yields a non-convex minimization problem \citep{Magder1997-oq, Bootkrajang2012-tc, Song2018-we}. On the other hand, the surrogate loss based on an unbiased estimate of the original loss function leads to a convex minimization problem \citep{Chaganty2014-ep, Natarajan2018-su, Du_Plessis2015-qx}. We propose and compare these two approaches in the classical regime and the high-dimensional regime, where the number of features $p$ is fixed or grows with $n$, potentially at a faster speed. 
In this paper, we make the following contributions:
\begin{itemize}
    \item Theoretical guarantees for parameter estimates for both non-convex likelihood-based and convex surrogate approaches in the classical regime (Proposition \ref{prop:4.1} and \ref{prop:4.2}). Our guarantee is for \textit{any} local minimum, by establishing that the likelihood function has actually at most one stationary point with high probability (Proposition \ref{prop:4.3}). In contrast, prior work either proves convergence to a local minimum or proves theory for the global minimizer without any guarantee of finding this point.
    \item Quantification of the efficiency gap of the two estimators based on the conditions of the design matrix $\X$, which provides an insight into the performance of the convex versus non-convex estimators (Corollary \ref{cor:1}).
    \item Mean-squared error guarantees and valid testing procedures in high dimensions, for the two estimators based on non-convex and convex approaches (Theorem \ref{thm:5.1} and \ref{thm:5.2}). The error bounds match with the optimal $s\log p/n$ rates known as minimax optimal in the sparse regression literature \citep{Raskutti2011-ug}. The testing procedure in high dimensions is based on de-biasing a penalized estimator and to the best of our knowledge, the first such theoretical analysis of testing procedures.
    \item A simulation study and a real data analysis to empirically support our theoretical findings. Our simulation study also indicates a potential advantage of using the convex surrogate of the likelihood in a very sparse regime, in contrast with the classical regime where the likelihood-based approach is provably optimal.
\end{itemize}

Now we outline the remainder of the paper. We begin by discussing the set-up of the work in Section \ref{sec:problem_setup}. In Section \ref{sec:appr}, we discuss how our noisy logistic regression model can be represented as a generalized linear model, and introduce the convex and non-convex approaches for parameter estimation. We establish point estimation guarantees and hypothesis testing in both low dimensions (Section \ref{sec:ld}) and high dimensions (Section \ref{sec:hd}). In Section \ref{sec:emp} and \ref{sec:real_data}, we apply convex and non-convex methods to synthetic and real data and compare the performance of the two estimators. Finally, we conclude the paper with remarks in Section \ref{sec:concl}.

\section{Problem Setup}
\label{sec:problem_setup}

First we define the problem and introduce the major notation. We assume access to samples $(\x_i,\z_i)_{i=1}^n$ where $(\z_i)_{i=1}^n$ are observed labels and $\x_i\in \R^p$ is a $p$-dimensional feature vector such that $\x_i = [1,\x_{2i},\dots,\x_{pi}]$. Each observed label $\z_i$ is a corrupted version of a latent binary outcome $\y_i$, where $\y_i\sim p_{\beta_0}(y_i|\x_i)$ is a true response with p.d.f is given by a logistic model, 
\begin{align}\label{eq:py|x}
p_\beta(y|\x) = \exp (y \x^\top \beta - \log(1+\exp(\x^\top \beta))),
\end{align}
and $\z_i$ is generated by flipping the value of $\y_i$ randomly based on known noise rates $\rho_0$ and $\rho_1$, with
\begin{align*}
	\rho_0:=\mathbb{P}(\z=1|\y=0) \quad \mbox{and} \quad \rho_1:=\mathbb{P}(\z=0|\y=1),
\end{align*}
for $\rho_0+\rho_1<1$. We assume that $\z_i$ and $\x_i$ are conditionally independent given the true response $\y_i$. The goal is to estimate and perform inference on $\beta_0$. 

We note that in the case of the data sets from deep mutational scanning (DMS) experiments, which we discuss as a concrete example of the noisy labels problem in Section \ref{sec:real_data}, the conditional independence assumption between $\z$ and $\x$ given $\y$ is satisfied and the noise rates are known from the previous experiments. In other applications, however, the conditional independence or the known noise rates assumptions can be limiting. The noise rates may depend on both the true class labels $\y_i$ and the features $\x_i$ (violation of conditional independence assumption), or the noise rates may not be known a priori. 

Unfortunately, in the case where the conditional independence assumption is violated, not much can be said about the results of the estimation unless the functional dependence of the noise rates on $\x$ is known or can be estimated from external sources, e.g., using a validation set containing both $(\x_i,\y_i,\z_i)_{i=1}^{n_V}$ \citep{Neuhaus1999-vv, Lyles2011-ss}. Otherwise, we cannot identify the true mean function because different pairs $(\E[\y|\x=x], \rho_y(x):=\P(\z\neq y |\y=y,\x=x))$ can lead to the same observed mean function $\E[\z|\x=x]$, resulting in a lack of identifiability of the true mean function $\E[\y|\x=x]$ and contamination mechanism $\rho_y(x)$ (Table \ref{tab:simple_illust}). 
\begin{table}[t]
\centering
\begin{tikzpicture}
\node (a) at (0,1){
(a) \begin{tabular}{ | c | c | c|}
\hline
& $x=0$ & $x=1$ \\ \hline
$\P(\y=0|\x=x)$ & 0.5 &  0.25\\ \hline
$\P(\y=1|\x=x)$ & 0.5 &  0.75\\ \hline
\end{tabular}
        
};
\node (c) at (0,-1){
(b) \begin{tabular}{ | c | c | c|}
\hline
& $x=0$ & $x=1$ \\ \hline
$\P(\y=0|\x=x)$ & 0.5 &  0.75\\ \hline
$\P(\y=1|\x=x)$ & 0.5 &  0.25\\ \hline
\end{tabular}
};
\node (b) at (7,0)
{
(c) \begin{tabular}{ | c | c | c|}
\hline
& $x=0$ & $x=1$ \\ \hline
$\P(\z=0|\x=x)$ & 0.5 &  0.5\\ \hline
$\P(\z=1|\x=x)$ & 0.5 &  0.5\\ \hline
\end{tabular}
};
\draw[->,ultra thick](a)--(b);
\draw[->,ultra thick](c)--(b);
\end{tikzpicture}
\caption{A simple illustration that two different $(\x,\y)$ distributions (table a,b) can result in the same observed $\x$ and $\z$ distribution (table c) if the noise rates are allowed to be dependent on both $\x$ and $\y$. In this example, the observed label is contaminated only when $\x=1$ ($\rho_y(0)=0, \forall y$). The noise rates are $(\rho_0(1)=0,\rho_0(1)=\frac{1}{3})$ for (a) and $(\rho_0(1)=\frac{1}{3}, \rho_1(1) =0)$ for (b).}
\label{tab:simple_illust}
\end{table}

Similarly, in the case where the noise rates are unknown but the noise is conditionally independent of the features given the class $y$, we may consider jointly estimating $(\beta,\rho_0,\rho_1)$ where $(\rho_0,\rho_1)$ are unknown nuisance parameters. However, this parametric approach can be problematic unless there are additional sources of information available for the estimation of $(\rho_0,\rho_1)$. Although $(\rho_0,\rho_1)$ are identifiable under the logistic model assumption \eqref{eq:py|x}, $(\rho_0,\rho_1)$ are not identifiable without \eqref{eq:py|x} \citep{Hausman1998-yd,Magder1997-oq,Ward2009-nn}. The identifiability of the noise parameters depends entirely on the assumed parametric form, and slight deviations from the assumed parametric model can produce very different estimation results for $(\rho_0,\rho_1)$ \citep[e.g.,][]{Hastie2013-wq}. Hence we focus on the setting where the conditional independence assumption is satisfied and the noise rates are known.

The relationship between the conditional mean of $\y$ and the conditional mean of $\z$ can be obtained under the conditional independence assumption.  By the factorization theorem, we have
\begin{align}\label{eq:factorize}
\E[\z|\x] &= \P(\z=1|\y=1)\E[\y|\x] + \P(\z=1|\y=0)(1-\E[\y|\x])\nonumber\\
& = (1-\rho_1)\E[\y|\x] +\rho_0 (1-\E[\y|\x]) \nonumber \\
& = (1-\rho_1 -\rho_0) \E[\y|\x] +\rho_0.
\end{align}

For the remainder of the paper, we let $\P_{\beta}$ be the distribution of the data when $\y|\x \sim p_{\beta} (\cdot|\x)$ in \eqref{eq:py|x}. We will sometimes write $\P (\cdot)$ for the probability distribution evaluated at the true parameter $\beta_0$, i.e., $\P (\cdot) = \P_{\beta_0}(\cdot)$, where $\beta_0$ is the unique minimizer of the population loss function, i.e., $\beta_0:=\argmin_{\beta \in \R^p} \E[-\log p_\beta(\y|\x)]$.

\subsection*{Connection to Positive-Unlabeled Learning}
The set-up in the previous section has an important implication in Positive-Unlabeled (PU) learning. In PU learning, we learn a model with two sets of samples, where the first set consists of \textit{labeled and positive} subjects and the second set consists of \textit{unlabeled} subjects whose associated responses are unknown. 

Two schemes are considered for PU-learning: the first scheme is a single training set scheme \citep{Elkan2008-aq} whose complete observations $(\x_i,\y_i,\z_i)_{i=1}^n$ are from a single distribution and only $(\x_i,\z_i)_{i=1}^n$ are recorded. The second scheme is where observations in the positive and unlabeled set are drawn separately, with the unlabeled set drawn from the general population \citep{Ward2009-nn, Song2018-we}. A subtle but important difference between the two schemes is that a sample from the first scheme has the same distribution as the joint distribution of the population but a sample from the second scheme does not. In the second scheme, positive subjects are over-represented in the data set, since the distribution of the unlabeled sample is the same as the population distribution and the labeled set consists of only positive subjects.  Therefore, a case-control sampling model \citep{McCullagh1989-zp} is necessary in the second scheme, where different inclusion probabilities are allowed based on the value of the true responses.

We demonstrate how both PU schemes fit into the set-up of our label noise problem and also show how the error rates $\rho_1$ and $\rho_0$ are related with the number of labeled ($n_\ell$) and unlabeled samples ($n_u$), and the proportion of positives in the unlabeled set $\pi := \P(\y=1|\z=0)$. We assume a parametric logistic model between $(\x,\y)$ as in \eqref{eq:py|x}. In both schemes, flipping probabilities from $\y$ to $\z$ do not depend on $\x$. Also, $\rho_0 = \P(\z=1|\y=0) = 0$ since all labeled elements ($\z=1$) are positive ($\y=1$) by the set-up of the PU problem. On the other hand, by Bayes' theorem we have
\begin{align*}
	\rho_1 &= \frac{\P(\y=1|\z=0)\P(\z=0)}{\P(\y=1|\z=1)\P(\z=1)+\P(\y=1|\z=0)\P(\z=0)}\\
	& = \frac{\pi\P(\z=0)}{\P(\z=1)+\pi\P(\z=0)}\\
	& \approx \frac{\pi n_u}{n_\ell + \pi n_u},
	\end{align*}
where we use the definition $\pi := \P(\y=1|\z=0)$ and $\P(\z=0)/\P(\z=1) \approx n_u / n_\ell$. Thus, the knowledge of $\pi$ practically amounts to knowing error rates ($\rho_0,\rho_1$) in PU-learning.

In the case-control sampling model, only \textit{selected} subjects $(\x_i,\z_i, \s_i = 1)_{i=1}^n$ are available in the data set where  $\s_i \in \{0,1\}$ represents whether the $i$th subject is selected or not. 
It is a well-known result \citep[e.g.,][]{McCullagh1989-zp} that case-control probabilities $\P(\y=1|\x,\s=1)$ differ from $\P(\y=1|\x)$ by the intercept, whose adjustment term is given by the log ratio of the different selection probabilities. More concretely, $p_\beta(y|\x,\s=1)$ can be written as 
\begin{align*}
p_\beta(y|\x,\s=1) &= \exp \{ y (\x^\top \beta^\gamma ) - \log(1+\exp(\x^\top\beta^\gamma))\},
\end{align*}
for $\beta^\gamma \in \R^p$ such that $\beta^\gamma_1 = \beta_1+ \gamma$ and $\beta^\gamma_j = \beta_j$, $\forall j\geq 2$, and where $\gamma:= \log(\P(\s=1|\y=1)/ \P(\s=1|\y=0))$ is the log ratio of the different selection probabilities. The log ratio $\gamma$ can also be expressed as functions of $n_\ell$, $n_u$, and $\pi$. Specifically,  $\gamma = \log (1+n_\ell / \pi n_u)$ was derived in \citet{Ward2009-nn}. 

We note that in both PU schemes the conditional distribution of $\y$ follows a logistic model, with the parameter $\beta_0$ in the first scheme and $\beta_0^\gamma$ in the second scheme. Since our target of interest is the coefficients of the model and $\beta_{0j}=\beta_{0j}^\gamma, \forall j\geq 2$, from this point on we will treat both sampling models the same. Specifically, we will omit conditioning on $\s$ and dependence of $\gamma$ in $\beta_0^\gamma$, and we assume $y|\x \sim p_{\beta_0}(y|\x) = \exp \{ y (\x^\top \beta_0 ) - \log(1+\exp(\x^\top\beta_0))\}$ in both PU schemes. 

\section{Convex and Non-Convex Approaches for Inference}
\label{sec:appr}
In this section, we briefly review generalized linear models \citep{McCullagh1989-zp} and discuss how all models discussed above can be fitted into the generalized linear model (GLM) framework. 
Then we introduce two approaches to estimate the true parameter $\beta_0$, i.e., the parameter from which the data is generated. The first approach is to use a negative log-likelihood loss function, which is a non-convex function of $\beta$. In the second approach, we discuss how we can construct a convex surrogate function.

\subsection{Generalized Linear Models (GLMs)}
Generalized linear models \citep{McCullagh1989-zp} are an extension of linear models, where a response $\z \in \mathcal{Z}$ has a p.d.f of the form
\begin{align}\label{def:exp_fam}
	p_\theta(z) = c(z)\exp (z \theta - A(\theta)),
\end{align}
for $\theta\in \R$ which can depend on $\x$, and $A(\theta) = \log \int_\mathcal{Z} c(z) \exp(z\theta) dz$.
The mean and variance of $\z$ can be derived from \eqref{def:exp_fam}:
\begin{align}\label{eq:exp_fam_mean_var}
	\E_\theta(\z|\x) = \mu = A'(\theta) \quad \mbox{and}\quad \textnormal{Var}_\theta(\z|\x) = A''(\theta) = \mathcal{V}(\mu),
\end{align}
where the variance function $\mathcal{V}$ is defined as $\mathcal{V}:= A'' \circ (A')^{-1}$ so that it is a function of $\mu$. 

Another important component of the GLM is the link function $g$, which relates the linear predictor $\x^\top \beta$ to the mean of the response $\mu$ by $g(\mu) = \x^\top \beta$. 
By definition of $g$ and $\mu = A'(\theta)$, $\theta = (g \circ A')^{-1} (\x^\top \beta)$, and we can rewrite \eqref{def:exp_fam} in terms of the linear predictor $\x^\top \beta$ and the link function $g$.
Therefore, the assumed distribution and the link function are two defining components of the GLM. We define the following:
\begin{definition}[{\bf GLM}]
	We say a sample $(\x_i,\z_i)_{i=1}^n $ is from a {\bf (GLM)} with parameters $(A,g)$ if the p.d.f of $\z\in \mathcal{Z}$ has the form
	\begin{align}\label{eq:pz|x-2}
		p_\beta(z|\x) &= c(z)\exp (z h(\x^\top \beta) - A(h(\x^\top \beta))),
	\end{align}
	for some $c$ only depending on $z$ and $h := (A')^{-1} \circ g^{-1}$.
\end{definition}
We require $g$ to be strictly increasing so that responses are positively related with linear predictors. 
A GLM is called canonical if $g=(A')^{-1}$ which implies $h(\cdot) = I(\cdot)$, an identity function. Suppose a random variable $\y$ is from a canonical GLM $(A,(A')^{-1})$. Then we have
\begin{align*}
	p_\beta(y|\x) = c(y)\exp (y\x^\top \beta - A(\x^\top \beta)).
\end{align*}
For example, the logistic model \eqref{eq:py|x} is an example of a canonical GLM. 

As we will discuss shortly in more detail, the statistical models for noisy labels belong to a special class of non-canonical GLMs whose mean is linearly related to the mean $A'$ of a canonical GLM. In this type of case, the link function $g$ is determined by such linear relationship since the link function is the inverse of the mean, i.e., $\E_\beta[\z|\x] = g^{-1}(\x^\top \beta)$. More concretely, suppose we have the following linear relationship
\begin{align}\label{eq:glm-l-mean}
	\E_\beta[\z|\x] = a A'(\x^\top \beta ) +b,
\end{align}
for some $a>0$, and $b\geq 0$. Then $g$ has to satisfy the equation $g(a A'(\x^\top \beta ) +b) = \x^\top \beta$, i.e., 
\begin{align}\label{def:g_glm_l}
	g(t) = (A')^{-1} \left(\frac{t-b}{a}\right).
\end{align}
Conversely, if $g$ is taken to be as in \eqref{def:g_glm_l}, the linear relationship \eqref{eq:glm-l-mean} is satisfied.
We refer to this sub-class of GLM, where the link function $g$ follows the form in \eqref{def:g_glm_l}, as {\bf (GLM-L)} with parameters $(A,a,b)$.

\subsection{Statistical Models for Noisy Labels and GLMs}\label{sec:noiselabel}

Now we relate the statistical models for noisy labels with the GLM framework. Since we have $\z \in \{0,1\}$,
\begin{align}
	p_\beta(z|\x) &= (\E_\beta(\z|\x))^{z}(1-\E_\beta(\z|\x))^{1-z} \nonumber \\
	& = \exp \left(z \theta- \log ( 1+e^{\theta})\right)\nonumber
\end{align}
for $\theta = \log\left(\frac{\E_\beta(\z|\x)}{1-\E_\beta(\z|\x)}\right)$, and thus $p_\beta(z|\x)$ belongs to a GLM with $A(t) =\log (1+e^t )$. 
Also by \eqref{eq:py|x} and \eqref{eq:factorize},  we have the representation
\begin{align}\label{eq:ez|x}
	\E_\beta[\z|\x] &= (1-\rho_1-\rho_0) \E_\beta[\y|\x]+ \rho_0 \nonumber\\
	&= (1-\rho_1-\rho_0) \frac{e^{\x^\top \beta}}{1+e^{\x^\top \beta}}+ \rho_0.
\end{align}
From \eqref{eq:ez|x}, we obtain the link function $g_{LN}$ (label-noise) by solving $g_{LN}(\E_\beta[\z|\x]) = \x^\top \beta$ for $\x^\top \beta$:
\begin{align}\label{def:g_nl}
	g_{LN}(t) = \mbox{logit} \left( \frac{t-\rho_0}{1-\rho_1-\rho_0}\right).
\end{align}
Therefore $(\x_i,\z_i)_{i=1}^n$ belongs to {\bf (GLM-L)} with $(\log(1+\exp(\cdot)),(1-\rho_1-\rho_0),\rho_0)$. 

In the subsequent analysis the variances of a clean label $\y$ and a noisy response $\z$ will play an important role. First we define mean functions $\mu$ and $\mu_z$ as $\mu(t) := A'(t)$ and $\mu_z(t) := A'(h_{LN}(t))$, for $A(\cdot) = \log (1+ \exp(\cdot))$ and $h_{LN}:= (A')^{-1} \circ g_{LN}^{-1}$. In particular, we have $\E_\beta[\z|\x] = \mu_z(\x^\top \beta)$ and $\E_\beta[\y|\x] = \mu(\x^\top \beta)$. By the definition of $\mathcal{V}$ in \eqref{eq:exp_fam_mean_var}, we have
\begin{align*}
	\textnormal{Var}_\beta(\z|\x) &= A''(h_{LN}(\x^\top \beta)) \,\, =  \mathcal{V}(\mu_z(\x^\top \beta))\\
	\textnormal{Var}_\beta(\y|\x) &= A''(\x^\top \beta) =  \mathcal{V}(\mu(\x^\top \beta))
\end{align*}
where the last equality uses the fact that $(A')^{-1} \circ \mu = I$.

\subsection{Non-Convex Approach Using a Negative Log-likelihood Loss}
Given a sample $(\x_i,\z_i)_{i=1}^n$ from {\bf (GLM)} with $(A,g)$, a natural approach for the estimation of $\beta_0$ is to take a likelihood-based approach due to the several optimality properties of a likelihood function. A negative log-likelihood loss can be obtained directly from \eqref{eq:pz|x-2} as
\begin{align}\label{def:Ll}
\L_n^\ell(\beta) &:= \frac{1}{n}\sum_{i=1}^n A(h(\x_i^\top \beta)) -\z_i h(\x_i^\top \beta)=\frac{1}{n}\sum_{i=1}^n \ell(\x_i^\top \beta, \z_i),
\end{align}
where we define $\ell(\x^\top \beta,\z) := A(h(\x^\top \beta)) - \z h(\x^\top \beta)$. In general, the likelihood becomes a non-convex function of $\beta$ unless $g = (A')^{-1}$ i.e., $g$ is canonical and $h$ is an identity function. 

The first and second derivatives of the likelihood function are 
\begin{align*}
\triangledown \L_n^\ell(\beta)= \frac{1}{n}\sum_{i=1}^n \ell'(\x_i^\top \beta,\z_i)\x_i, \quad \triangledown^2 \L_n^\ell(\beta)= \frac{1}{n}\sum_{i=1}^n \ell''(\x_i^\top \beta,\z_i)\x_i\x_i^\top ,
\end{align*}
where we write
\begin{align}
	\ell'(t,z) &= \left(A'(h(t))-z\right) h'(t) \label{eq:l'} \\
	 \ell''(t,z) &= A''(h(t))h'(t)^2+(A'(h(t))-z) h''(t)\nonumber  \\
	 &:= \rho_I(t) + \rho_R(t,z)\label{eq:l''}
\end{align}
for $\rho_I(t) :=A''(h(t))h'(t)^2$ and $\rho_R(t,z) := (A'(h(t))-z) h''(t)$. Although $\rho_I \geq 0$, the sign of $\rho_R$ is arbitrary, and thus $\triangledown^2 \L_n^\ell(\beta)$ is not necessarily a positive semi-definite matrix.

\subsection{Construction of a Convex Surrogate Loss}
Next, we discuss an alternative approach involving a convex surrogate function when a sample is from a {\bf (GLM-L)} model with parameters $(A,a,b )$. Essentially, we construct an unbiased estimator of a convex loss function with the same minimizer, which is a well-known idea in stochastic optimization \citep{Nemirovski2009-ms} and has also been investigated in the latent variable model literature \citep{Loh2012-tl, Chaganty2014-ep,Natarajan2018-su}. More concretely, if the responses $(\y_i)_{i=1}^n$ from a canonical GLM are available, we can minimize a convex loss $\L_n^{c}(\beta)$ which we define as
\begin{align}\label{def:Lc}
	\L_n^{c}(\beta) &:=  \frac{1}{n}\sum_{i=1}^n A(\x_i^\top \beta) - \y_i(\x_i^\top \beta).
\end{align}
For example, we can take this convex approach if labels are not contaminated. 
Since $(\y_i)_{i=1}^n$ are not available, we construct a surrogate function $\L_n^{s}(\beta)$ by replacing $\z$ with a function output $T(\z)$ while keeping $h(\cdot)=I(\cdot)$:
\begin{align}\label{def:Ls}
	\L_n^{s}(\beta) &:=  \frac{1}{n}\sum_{i=1}^n A(\x_i^\top \beta) - T(\z_i)(\x_i^\top \beta).
	\end{align}
 To obtain a consistent estimator, the function $T$ needs to satisfy
  $\E_\beta[T(\z)|\x] = A'(\x^\top \beta) = \E_\beta [\y|\x] $. Such a function $T$ is available by the {\bf (GLM-L)} model class assumption. Specifically, we let $T$ be
	$T(t) := (t-b)/a$
so that $\E_\beta[T(\z)|\x]= \E_\beta[(\z-b)/a] =A'(\x^\top \beta )$ by \eqref{eq:glm-l-mean}. For a future reference we define 
\begin{align}\label{def:ls}
	\ell_s(\x^\top \beta,\z) := A(\x^\top \beta) - T(\z)(\x^\top \beta)
\end{align}
so that $\L_n^s(\beta) = n^{-1} \sum_{i=1}^n \ell_s(\x_i^\top \beta ,\z_i)$.

At any fixed parameter $\beta$, the surrogate loss \eqref{def:Ls} is an unbiased estimate of the loss \eqref{def:Lc}. We note
\begin{align*}
	\E_{\beta_0}[\L_n^{s}(\beta)] &= \E_{\beta_0} \left[ \frac{1}{n}\sum_{i=1}^n A(\x_i^\top \beta) - \E_{\beta_0}[T(\z_i)|\x_i](\x_i^\top \beta)\right]\\
	&= \E_{\beta_0} \left[ \frac{1}{n}\sum_{i=1}^n A(\x_i^\top \beta) - A'(\x_i^\top \beta_0)(\x_i^\top \beta)\right]\\
	& = \E_{\beta_0}[\L_n^{c}(\beta)],
\end{align*} 
where we use the law of iterative expectation and $\E_{\beta_0}[\y|\x] = A'(\x^\top \beta_0)$. 

%% Notation
\subsection{Notation}
Before proceeding, we pause to define some notation that will be useful in presenting our theoretical results. For $v \in \mathbb{R}^p$, we denote the $\ell_1$, $\ell_2$, and $\ell_\infty$ norms as $\|v\|_1 := \sum_{i=1}^p |v_i| $, $\|v\|_2 := \sqrt{v^\top v}$, and  $\|v\|_\infty := \sup_{1\leq j\leq p} |v_j|$. Similarly, for a function $f$, we define $\|f\|_p := (\int |f(x)|^p dx)^{1/p}$ and $\|f\|_\infty := \sup_x |f(x)|$. In the case of matrix norm, for $A \in \R^{m\times n}$, we denote a Frobenius norm as $\|A\|_F := \sqrt{\sum_{i,j}|A_{ij}|^2}$, an operator norm as $\|A\|_2 := \sigma_{\textnormal{max}}(A)$, and an element-wise max norm as $\|A\|_{\textnormal{max}} := \max_{i,j} |A_{ij}|$. We define a condition number of $A$ as $\kappa(A) := \sigma_{\textnormal{max}}(A)/\sigma_{\textnormal{min}}(A)$. 

 For a set $S$, we use $|S|$ to denote the cardinality of $S$. For $v \in \R^p$ and any subset $S \subseteq \{1,\dots,p\}$, $v_S \in \R^{|S|}$ denotes the sub-vector of the vector $v$ by selecting the components with indices in $S$. Likewise for any matrix $A \in \mathbb{R}^{m \times n}$, $A_S \in \R^{m\times|S|}$ denotes a sub-matrix having columns in $S$. For matrices $A\in \mathbb{R}^{m\times n}$ and $B\in \mathbb{R}^{m\times n}$, we say $A\succeq B$ if $A-B$ is a positive semi-definite matrix and $A\succ B$ if $A-B$ is positive definite. Also we write $\mathcal{C}(A)$ to refer to a column space of $A$. Also we use $\B_q(r;v)$ to denote a ball with radius $r$ in the $\ell_q$ norm centered at $v \in \R^p$. If $v=0$, we simply use $\B_q(r)$ to denote the ball.

For functions $f$ and $g$, we write $f(n) = O(g(n))$ if there exists a constant $C>0$ such that $f(n) \leq C g(n)$, $\forall n$, and $f(n) \asymp g(n)$ if $f(n) = O(g(n))$ and $g(n) = O(f(n))$. Also for a random variable $X_n$, we write $X_n = O_p(a_n)$ if $X_n/a_n$ is bounded in probability and $X_n = o_p(a_n)$ if $X_n/a_n$ converges to $0$ in probability. Also for simplicity, we sometimes use $\x_1^n$ to refer to the collection of random variables $(\x_i)_{i=1}^n$. We write a.s. to denote `almost surely', i.e., an event that occurs with probability $1$. Also, for a sequence of events $(\mathcal{E}_n)_{n\geq 1}$, we say $\mathcal{E}_n$ holds with high probability (w.h.p) if $\P(\mathcal{E}_n) \overset{n}{\rightarrow} 1$.

%% Point Estimation in low-d
\section{Estimation and Testing in the Classical Regime}
\label{sec:ld}
In this section, we discuss the statistical properties of two estimators from convex and non-convex approaches in the classical regime where the number of features $p$ is fixed. In particular, we demonstrate that both approaches yield consistent estimators, but the estimator based on the non-convex approach has better efficiency than the convex counterpart in the large $n$ limit. Also, we quantify the efficiency gap between the two approaches and discuss when two approaches can be comparable.

\subsection{Consistency and Relative Asymptotic Efficiency}
We define a global minimizer of $\L_n^\ell(\beta)$ and  $\L_n^{s}(\beta)$  as 
\begin{align}\label{def:low-d_estimators}
\widehat{\beta_{\ell}} \in \argmin_{\beta \in \B_2(r) } \L_n^{\ell}(\beta) \quad \mbox{and} \quad
\widehat{\beta_{s}} \in \argmin_{\beta \in \R^p } \L^{s}_{n}(\beta).
\end{align}
By definition of $\L_n^\ell$ and $\L_n^{s}$, $\widehat{\beta_{\ell}}$ is the solution of a non-convex optimization problem, whereas $\widehat{\beta_{s}}$ is based on the convex problem. In the case of the non-convex optimization problem, we limit the search space to a compact region $\B_2(r)$, where $r$ is some large number such that $r \geq \|\beta_0\|_2$. The use of a compact search space ensures that the gradient of the non-convex loss function is uniformly bounded away from zero for values of $\beta$ not near the true parameter. 

Clearly, it is not obvious whether it is feasible to obtain $\widehat{\beta_{\ell}}$ in practice, since finding a global minimizer of a non-convex function is in general a challenging problem. However, obtaining a stationary point of $\L_n^\ell(\beta)$ is in fact enough when $n$ is sufficiently large, as we will demonstrate in Proposition \ref{prop:4.3} that in the classical regime, with high probability, $\L_n^\ell(\beta)$ has a unique stationary point (i.e. the global minimizer).

In the following Proposition \ref{prop:4.1}, we show that both estimators are consistent for $\beta_0$ and also quantify their asymptotic efficiency. We first state the following minimum eigenvalue condition, which is a standard assumption in the classical regime with the fixed design \citep[e.g.,][]{Fahrmeir2001-gv, Shao2003-vx}. 
\begin{assumption}\label{a1}
There exist $C_\lambda>0$ and $C_X < \infty$ such that $\lambda_{\textnormal{min}}(n^{-1}\sum_{1\leq i\leq n} \x_i\x_i^\top ) \geq C_\lambda$ and $\sup_{1\leq i\leq n} \|\x_i\|_\infty \leq C_X, \forall n.$
\end{assumption}

\begin{proposition}\label{prop:4.1}
(Fixed design) Suppose a sample $(\x_i, \z_i)_{i=1}^n$  is from a {\bf{(GLM-L)}} with \\$(\log(1+\exp(\cdot)),(1-\rho_1-\rho_0),\rho_0)$ and $\z_i \in \{0,1\}$. Assume {\bf A\ref{a1}} and the classical regime where the number of features $p$ is fixed and the sample size $n \rightarrow \infty$.  Then,
\begin{align*}
&\sqrt{n} \I_n^\ell(\beta_0)^{1/2} (\widehat{\beta_\ell} - \beta_0) \overset{d}{\rightarrow} \mathcal{N}(0,I_p)\\
&\sqrt{n}  \I_n^{s}(\beta_0)^{1/2} (\widehat{\beta_s} - \beta_0) \overset{d}{\rightarrow} \mathcal{N}(0,I_p),
\end{align*}
for positive definite matrices $\I_n^\ell(\beta),\I_n^{s}(\beta)$ defined as
\begin{align*}
	\I_n^\ell(\beta)&:= (1-\rho_1-\rho_0)^2 \cdot \frac{1}{n} \sum_{i=1}^n \frac{\mathcal{V}(\mu(\x_i^\top \beta))^2}{\mathcal{V}(\mu_z(\x_i^\top \beta))}\x_i\x_i^\top ,\\
	\I_n^{s}(\beta)&:= (1-\rho_1-\rho_0)^2 \cdot \nonumber\\
	&\quad \left( \frac{1}{n} \sum_{i=1}^n \mathcal{V}(\mu(\x_i^\top \beta))\x_i\x_i^\top \right) \left(  \frac{1}{n} \sum_{i=1}^n  \mathcal{V} (\mu_z(\x_i^\top \beta)) \x_i\x_i^\top 
\right)^{-1}\left( \frac{1}{n} \sum_{i=1}^n  \mathcal{V}(\mu(\x_i^\top \beta))\x_i\x_i^\top \right).
\end{align*}
\end{proposition}
The proof essentially uses classical likelihood and generalized estimating equations theory and is provided in the Appendix \ref{supp_sec:prop:4.1}. One point that deserves special attention is the similarity between $\I_n^\ell(\beta_0)$ and $\I_n^{s}(\beta_0)$ in Proposition \ref{prop:4.1}. In particular, if $\mathcal{V}(\mu_z(\x_i^\top \beta)) \approx \mathcal{V}(\mu(\x_i^\top \beta))$ for all $i$, the two information matrices will turn out to be very similar. 

The following Corollary shows that $\I_n^\ell(\beta_0) \succeq \I_n^{s}(\beta_0)$ and quantifies the discrepancy between the two information matrices. First we define two weight matrices $W_y(\beta)$ and $W_z(\beta)$ as
\begin{align*}
	W_y(\beta):= \mbox{diag}(\{\mathcal{V}(\mu(\x_i^\top \beta))\}_{i=1}^n) \quad \mbox{and} \quad
	W_z(\beta) := \mbox{diag}(\{\mathcal{V}(\mu_z(\x_i^\top \beta))\}_{i=1}^n),
\end{align*}
whose diagonal entries consist of the conditional variances of $\y_i$ and $\z_i$ given $\x_i$, respectively. We suppress the dependence on $\beta$ if $\beta = \beta_0$ and let $W_y:= W_y(\beta_0)$ and $W_z:=W_z(\beta_0)$ for ease of notation. Also, we define the gap $\widehat{\delta}(\mathcal{M},\mathcal{N})$ between two vector subspaces $\mathcal{M}, \mathcal{N}$ as \citep[e.g.,][]{Kato2013-va}
\begin{align}\label{def:gap}
	\widehat{\delta}(\mathcal{M},\mathcal{N}):= \max \{\delta(\mathcal{M},\mathcal{N}),\delta(\mathcal{N},\mathcal{M})\}, \,\, \mbox{for } \delta(\mathcal{M},\mathcal{N}) := \sup_{u\in \mathcal{M}, \|u\|_2 = 1} \inf_{v \in \mathcal{N}} \|u-v\|_2.
\end{align}
The gap measures the distance between two subspaces, with $\widehat{\delta}(\mathcal{M},\mathcal{N}) = 0 $ if and only if $\mathcal{M} = \mathcal{N}$. Now we present the following Corollary.

\begin{corollary}\label{cor:1}
Assume the conditions as in Proposition \ref{prop:4.1}. We have $\I_n^\ell(\beta_0) \succeq \I_n^{s}(\beta_0)$ and 
\begin{align}\label{eq:cor1}
\|I_p-{\I_n^\ell(\beta_0)}^{-1/2}\I_n^{s}(\beta_0){\I_n^\ell(\beta_0)}^{-1/2}\|_2 
\leq c_n \widehat{\delta}^2 (\mathcal{C}(W_z^{-1}W_y\X),\mathcal{C}(\X) )
\end{align}
where $c_n :=\kappa(\X^\top \X/n) \kappa(W_y^2) \kappa(W_z^2) $ and $c_n = O(1)$. 
In particular,  $\widehat{\beta_s}$ achieves asymptotic efficiency if $\mathcal{C}(W_z^{-1}W_y\X) = \mathcal{C}(\X)$. 
\end{corollary}
\begin{figure}[b]
  \begin{minipage}[c]{0.4\textwidth}
  \includegraphics[width=\textwidth]{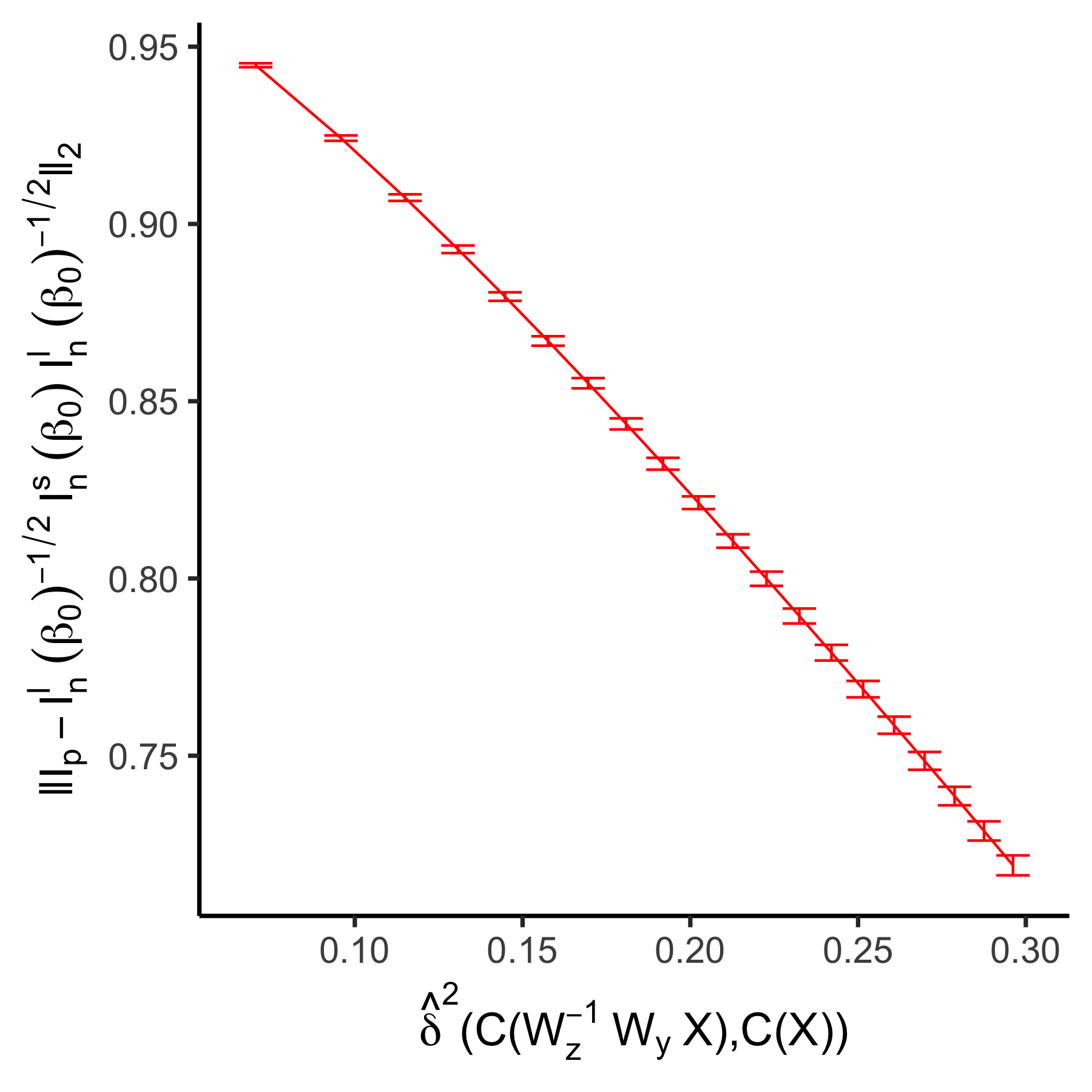}
  \end{minipage}\hfill
  \quad\quad
  \begin{minipage}[c]{0.55\textwidth}
    \caption{The plot of the relative $\ell_2$ difference between $\I_n^\ell(\beta_0)$ and $\I_n^{s}(\beta_0)$ as a function of $gap^2 = \widehat{\delta}^2 (\mathcal{C}(W_z^{-1}W_y\X),\mathcal{C}(\X))$. To generate design matrix $\X$ with various gap values, each $\x_i$ is sampled from an equal mixture of multivariate Gaussian distribution with different centers; see Section \ref{sec:emp} for details. The results were averaged over 10000 repetitions at each center, and the bars denote one standard error.
    } \label{fig:fig1}
  \end{minipage}
\end{figure}
The proof is provided in Appendix \ref{supp_sec:cor1}. We note if $p=1$, the relative $\ell_2$ difference equals to 
\begin{align*}
    \|I_p-{\I_n^\ell(\beta_0)}^{-1/2}\I_n^{s}(\beta_0){\I_n^\ell(\beta_0)}^{-1/2}\|_2  = \left\vert 1 - \frac{\I_n^\ell (\beta_0)^{-1}}{\I_n^s (\beta_0)^{-1}} \right\vert = 1 - \textnormal{ARE} (\widehat{\beta}_s, \widehat{\beta}_\ell ; \beta_0)
\end{align*}
where $\textnormal{ARE} (\widehat{\beta}_s, \widehat{\beta}_\ell ; \beta_0)$ denotes the asymptotic relative efficiency of $\widehat{\beta}_s$ with respect to $\widehat{\beta}_\ell$. In general, we can find the direction $u$ such that $\|u\|_2 = 1 $ and
\begin{align*}
     \|I_p-{\I_n^\ell(\beta_0)}^{-1/2}\I_n^{s}(\beta_0){\I_n^\ell(\beta_0)}^{-1/2}\|_2  \geq 1- \textnormal{ARE} (u^\top \widehat{\beta}_s, u^\top \widehat{\beta}_\ell ; \beta_0).
\end{align*}

The bound \eqref{eq:cor1} shows that the relative $\ell_2$ difference between $\I_n^\ell(\beta_0)$ and $\I_n^{s}(\beta_0)$ depends on how dissimilar $W_z^{-1}W_y$ is from the identity matrix. We observe that $W_z^{-1}W_y$ is a diagonal matrix where the diagonal entries are ratios of the variances of $\z$ and $\y$, i.e.,
\begin{align}\label{eq:Wz1Wy}
{(W_z^{-1}W_y)}_{ii} 
&=  \frac{\mathcal{V}(\mu(\x_i^\top \beta_0))}{\mathcal{V}(\mu_z(\x_i^\top \beta_0))} \nonumber\\
&= \frac{\mu(\x_i^\top \beta_0)(1-\mu(\x_i^\top \beta_0))}{\{(1-\rho_1-\rho_0)\mu(\x_i^\top \beta_0)+\rho_0\}\{(1-\rho_1-\rho_0)(1-\mu(\x_i^\top \beta_0))+\rho_1\}},	
\end{align}
noting $\mathcal{V}(\mu) = \mu (1-\mu)$. In light of these observations, the inefficiency of a surrogate convex loss function can be understood as the result of sub-optimal weighting of the observations due to the mis-specification of the variance matrix for $\z$. In fact, in the special case of the intercept-only model, no covariate information is available for the optimal weighting of the observations. In this case, we have  $W_y = w_1 I_n$, $W_z = w_2 I_n$ for some $w_1,w_2>0$, and thus $\mathcal{C}(W_z^{-1}W_y\X)=\mathcal{C}(\X)$ and the inequality \eqref{eq:cor1} is sharp. 

We also note that the variance ratios are also functions of the noise rates. Each diagonal entry ${(W_z^{-1}W_y)}_{ii}$ is a point on the variance ratio curve \[r(t)=\frac{\mu(t)(1-\mu(t))}{\{(1-\rho_1-\rho_0)\mu(t)+\rho_0\}\{(1-\rho_1-\rho_0)(1-\mu(t))+\rho_1\}}\] at $t = \x_i^\top \beta_0$, where the curve $r(t)$ is a function of the noise rates $\rho_1$ and $\rho_0$. When there is no noise in the labels, i.e., $\rho_0=\rho_1=0$, $r(t) \equiv 1$, all diagonal entries are $1$. In this case, we again have $\mathcal{C}(W_z^{-1}W_y\X)=\mathcal{C}(\X)$. With positive noise rates, higher noise rates tend to be associated with larger amounts of perturbation to the column space of $\X$, with the caveat that the locations of $\{\x_i^\top \beta_0\}_{i=1}^n$ also play a role in determining the amount of perturbation given the noise rates. We provide more discussion on the relationship between the amounts of noise in labels and the relative efficiency of the two estimators in Section \ref{sec:emp-2_NL}.

So far, we have considered the fixed design setting. Now we present the result equivalent to Proposition \ref{prop:4.1} in the random design. We assume that rows in the random design matrix satisfy a sub-gaussian tail condition. We define this sub-gaussian tail condition as follows:
\begin{definition}[sub-gaussian tail condition]
	We say a random vector $\x \in \R^p$ satisfies the sub-gaussian tail condition with parameter $K$ if 
	\begin{align}\label{def:sub-gaussian}
		\sup_{u \in \R^p;\|u\|_2=1} \E [\exp (u^T\x)^2/K^2] \leq 2.
	\end{align}
\end{definition}
For example, a random vector $\x\in \R^p$ with $\mu_X = \|\E[\x]\|_2$ satisfies \eqref{def:sub-gaussian} with $K=c_1 \mu_X + c_2\sigma_X$ for some absolute constants $c_1, c_2>0$ if the centered vector $\x - \E[\x]$ is sub-gaussian with parameter $\sigma_X$, i.e.,
\begin{align*}
 \sup_{u \in \R^p;\|u\|_2=1} \E[\exp(t(u^T\x-\E[u^T\x]))] \leq \exp (t^2 \sigma_X^2/2), \forall t\in \R.
\end{align*}
We replace {\bf A\ref{a1}} with the following assumption:

\begin{assumptionB}\label{a1'}
(Random design) For a random feature vector $\x \in \mathbb{R}^p$, $\x$ satisfies the sub-gaussian tail condition with parameter $K_X$ for a positive constant $K_X<\infty$. Also, there exist $C_\lambda>0$ and $C_X < \infty$ such that $\lambda_{\textnormal{min}}(\E[\x_i\x_i^\top] ) \geq C_\lambda$ and $\sup_{1\leq i\leq n} \|\x_i\|_\infty \leq C_X, \forall n.$\end{assumptionB}

\begin{proposition}\label{prop:4.2}
(Random design) Assume the conditions of Proposition \ref{prop:4.1} where {\bf A\ref{a1}} is replaced with {\bf A\ref{a1'}'}.  Then,
\begin{align*}
&\sqrt{n} (\widehat{\beta_\ell} - \beta_0) \overset{d}{\rightarrow} \mathcal{N}(0,\I^\ell(\beta_0)^{-1})\\
&\sqrt{n}(\widehat{\beta_s} - \beta_0) \overset{d}{\rightarrow} \mathcal{N}(0,\I^s(\beta_0)^{-1}),
\end{align*}
for $\I^\ell(\beta),\I^s(\beta)$ defined as 
\begin{align*}
 \I^\ell(\beta)&:=(1-\rho_1-\rho_0)^2\E_\beta\left( \frac{\mathcal{V}(\mu(\x^\top \beta))^2}{\mathcal{V}(\mu_z(\x^\top \beta))} \x\x^\top \right), \\
 \I^s(\beta)&:= (1-\rho_1 -\rho_0)^2\E_\beta \left(\mathcal{V}(\mu(\x^\top \beta))\x\x^\top \right) \E_\beta\left( \mathcal{V}(\mu_z(\x^\top \beta))\x\x^\top \right)^{-1}\E_\beta \left(\mathcal{V}(\mu(\x^\top \beta))\x\x^\top \right).
\end{align*}
Also,  $\I^\ell(\beta_0) \succeq \I^{s}(\beta_0)$. 
\end{proposition}

The result follows from classical M-estimation theory \citep[see e.g.,][]{Van_der_Vaart1998-fr}. $\I^\ell(\beta_0)\succeq \I^s(\beta_0)$ follows from Theorem 1 in \citet{Morton1981-jb}. 

The final result that we will present in this section is about the comparability between the global and local minimizer in the low-dimensional setting. So far, we have only considered the global minimizer of the empirical risk function $\L_n^\ell(\beta)$ which is the MLE. However, since $\L_n^\ell(\beta)$ is non-convex, obtaining the global minimizer $\widehat{\beta_\ell}$ is in general computationally intractable, and algorithms on the optimization of non-convex functions focus on finding a stationary point of the objective function.  

The population risk, albeit non-convex, can be shown to be unimodal and also strongly convex around $\beta_0$ in some GLMs. Often, fast probability tail decay of $\x$ and $\z$, and boundedness of the derivatives of the loss function allow enough concentration of the empirical risk function around the population counterpart that the empirical risk function has a unique stationary point, which in fact is the global minimum \citep{Mei2018-ec}. We make the following assumption {\bf A\ref{a2}} about the smoothness of $\ell$, for $\ell$ defined in \eqref{def:Ll}. In Corollary \ref{cor:2}, we show that Assumption {\bf A\ref{a2}} is satisfied for the label noise model, {\bf (GLM)} with parameters $(\log(1+\exp(\cdot)), g_{LN})$.

\begin{assumption}\label{a2}
$\ell''$ is Lipschitz w.r.t its first argument, i.e., 
$|\ell''(a,t)-\ell''(a',t)|\leq L_\ell |a-a'|, \forall t$. Furthermore, there exists $C_\ell <\infty$ such that $ \max \{ \|\ell'\|_\infty , \|\rho_I\|_\infty , \|\rho_R\|_\infty  \} \leq C_\ell$.
\end{assumption}

\begin{proposition}\label{prop:4.3}
	Suppose a sample $(\x_i,\z_i)_{i=1}^n$ is from a {\bf (GLM)} with parameters $(A,g)$ and assume {\bf A\ref{a1'}'} and {\bf A\ref{a2}}. Moreover, assume $\z_i|\x_i$ satisfies the sub-gaussian tail condition with parameter $K_Z$ for a positive constant $K_Z<\infty$, i.e.,
	\begin{align*}
		\E[e^{t(\z_i-\E[\z_i|\x_i])}|\x_i] \leq e^{t^2 K_Z^2/2} \text{ a.s., } \forall t \in \R.
	\end{align*}
Then, for any given $\epsilon>0$, there exists a unique stationary point of $\L^\ell_n(\beta)$ in $\B_2(r)$ with probability at least $1-\epsilon$, given a sufficiently large $n \geq C\log(1/\epsilon)p \log n$ where the constant $C$ depends only on the model parameters in our assumptions. The unique stationary point is the global minimum of $\L^\ell_n(\beta)$.
\end{proposition}
\begin{corollary}\label{cor:2}
	Under Assumption {\bf A\ref{a1'}'}, the log-likelihood for the label noise GLM has a unique stationary point in $\B_2(r)$, which is the MLE, with high probability.
\end{corollary}
Proofs of Proposition \ref{prop:4.3} and Corollary \ref{cor:2} are provided in Appendix \ref{sec:pf_prop:4.3} and \ref{supp_sec:cor2}.

%% High-d
\section{Estimation and Testing in the High-Dimensional Regime}\label{sec:hd}
\subsection{\texorpdfstring{$\ell_1$}{l1} and \texorpdfstring{$\ell_2$}{l2} Consistency}

In many modern data sets, the number of the features $p$ may be comparable to sample size $n$, or may even be substantially larger ($p \gg n$). In this section, we discuss the estimation of $\beta_0$ in the high-dimensional regime. For the non-convex optimization problem, as in the previous section, we restrict the search space to be $\B_2(r)$ for some large enough $r$ that the true parameter $\beta_0$ is an interior point of the parameter space.
We propose two estimators $\widehat{\beta}_\ell^H$, $\widehat{\beta}_s^H$ as solutions of the following optimization problems
\begin{align}\label{def:high-d_estimators}
\widehat{\beta}_{\ell}^H \in \argmin_{\beta \in \B_2(r)} \L_n^{\ell}(\beta) +\lambda_\ell\|\beta\|_1 \quad \mbox{and}\quad
\widehat{\beta}_{s}^H \in \argmin_{\beta \in \R^p} \L^s_{n}(\beta)+\lambda_s\|\beta\|_1,
\end{align}
where $\L_n^\ell(\beta)$ is a non-convex negative log-likelihood loss and $\L_n^{s}(\beta)$ is a convex surrogate loss. Here $\lambda_\ell$ and $\lambda_s$ are tuning parameters which need to be chosen appropriately, and we will discuss their choices shortly. Finally, we note that in many cases, it is common to leave a finite number of coordinates unpenalized. An important special example is when the model includes an intercept feature. The theory that we develop in this section has a straightforward extension when the $\ell_1$ penalty is modified to exclude a finite subset of features. 
In the low-dimensional setting, we established that the global minimizer can be obtained with high probability, but it is hard for a similar result to hold in the high-dimensional regime. Therefore, instead of $\widehat{\beta}_{\ell}^H$ we make use of a stationary point and define $\widetilde{\beta}_\ell^H$ to be a stationary point of the first optimization problem in \eqref{def:high-d_estimators}.

% RSC
We now study the statistical guarantees of the two estimators in the high-dimensional regime. First, we impose the standard sparsity assumption on $\beta_0$,
$s_0:= \|\beta_0\|_0$. The core condition which needs to be established is the restricted strong convexity (RSC) condition, the notion of which was first proposed by \citet{Negahban2012-bd} for convex loss functions and extended for non-convex functions by \citet{Loh2017-jk} and \citet{Loh2015-cn}. Similarly as in the definition in \citet{Loh2017-jk}, we define the RSC condition as follows.
% Definition
\begin{definition}[RSC condition] \label{def:RSC} We say $\mathcal{L}_n$ satisfies a restricted strong convexity (RSC) condition with respect to $\beta_0$ with curvature $\alpha$, a tolerance function $\tau_{n,p}(\cdot):\R \rightarrow \R$, and a region $\Omega \subseteq \R^p$ if there exist $\alpha>0$, $\tau_{n,p}(\cdot)$ such that 
	\begin{align}\label{eq:RSC}
		\langle \triangledown\mathcal{L}_n(\beta)-\triangledown\mathcal{L}_n(\beta_0), \beta- \beta_0 \rangle \geq \alpha \|\beta-\beta_0\|_2^2 -\tau_{n,p}(\|\beta-\beta_0\|_1), \quad \forall \beta \in \Omega.
	\end{align}
\end{definition}
For example, the RSC condition in \citet{Loh2017-jk} corresponds to the choices $\tau_{n,p}(t) = \tau (\log p/n)t^2$ for a constant $\tau \geq 0$ and $\Omega = \B_2(\delta;\beta_0)$ with a radius $\delta>0$.
The main idea behind the definition of RSC is that it is the relaxed version of the strong convexity; when $\alpha >0$, $\tau_{n,p} \equiv 0$ and the inequality \eqref{eq:RSC} holds for all $\beta$ and $\beta_0 \in \R^p$. Even if $\L_n$ is convex, $\L_n$ cannot be strongly convex in the high-dimensional regime due to the rank deficiency, which causes the curvature to vanish in some directions. The RSC condition guarantees that gradient information can still be exploited to direct the algorithm to the optimal point $\beta_0$ in the lack of strong convexity.

We will establish the RSC condition \eqref{eq:RSC} with the choices $\tau_{n,p}(t) = \tau_\ell \sqrt{\log p/n}t$ and $\tau_{n,p}(t) = \tau_s (\log p/n)t^2$ for $\L_n^\ell(\beta)$ and $\L_n^s(\beta)$, respectively, for some $\tau_\ell, \tau_s >0$. First, we discuss some additional conditions needed to establish the RSC condition for the negative log-likelihood loss $\L_n^\ell(\beta)$. 
\begin{assumption}\label{a3}
There exist $C_\rho>0$ such that $\sup_t |\ell''(t,z)t|\leq C_\rho$, for all $z\in \mathcal{Z}$.
\end{assumption}
 \begin{assumption}\label{a4}
 There exist $C_d, C_b <\infty$ such that $\max_{1\leq i\leq n}|\x_i^\top (\beta_0/\|\beta_0\|_2)|\leq C_d$, a.s. and $\|\beta_0\|_2\leq C_b$.
 \end{assumption}
 Assumption {\bf A\ref{a3}} is a technical assumption which ensures that $\ell''(t.z)$ decays at least on the order of $1/t$ as $t\rightarrow \pm \infty$. We will show in Corollary \ref{cor:3} that Assumption {\bf A\ref{a3}} is satisfied for the noisy labels model, where $\ell''(t,z) = A''(h_{LN}(t))h'_{LN}(t)^2 + (A'(h_{LN}(t))-z)h''_{LN}(t)$ for $A(t) = \log (1+ \exp(t))$ and $h_{LN}(\cdot)$ defined in Section \ref{sec:noiselabel}.
 Assumption {\bf A\ref{a4}} concerns the boundedness of the signal. In particular, we assume that the size of $\x$ projected onto $\beta_0$ as well as $\|\beta_0\|_2$ are bounded. 

Now we present two propositions to establish the RSC conditions with high probability for $\L_n^\ell(\beta)$ and $\L_n^{s}(\beta)$. 
%% Proposition 4.1
\begin{proposition}\label{prop:5.1}
	Suppose a sample is from a {\bf (GLM)} with $(A,g)$ which satisfies the random design condition {\bf{A\ref{a1'}'}}. We assume a high-dimensional regime where $p \gg n $ and $\log p/n = o(1)$. Also, we assume that $\ell$ is smooth and has a fast decaying tail ({\bf{A\ref{a2}-A\ref{a3}}}), and that the linear signal is bounded ({\bf A\ref{a4}}). Then for any given $\epsilon>0$, there exist positive constants $\alpha_\ell$ and $\tau_{\ell }$ such that the following event
\begin{align}
	\left(\triangledown \mathcal{L}_n^\ell(\beta) - \triangledown \mathcal{L}^\ell_n(\beta_0)\right)^\top (\beta-\beta_0) \geq
\alpha_\ell\|\beta-\beta_0\|_2^2  - \tau_{\ell}\sqrt{\frac{\log p}{n}}  \|\beta-\beta_0\|_1 ,\quad \forall \beta\in \B_2(r)\label{eq:prop5.1}
\end{align}
holds with probability at least $1-\epsilon$, given a sufficiently large sample size $n \geq C (1/\epsilon)^{1/7}$ for a constant $C$ depending only on the model parameters. 
\end{proposition}
The proof is deferred to the Appendix \ref{supp_sec:prop5.1}. We also present an equivalent result for the convex surrogate loss when a sample is from a {\bf (GLM-L)} model. We recall that the convex approach discussed in the previous section is available when a sample is from {\bf (GLM-L)} model.

\begin{proposition}\label{prop:5.2}
	Suppose a sample is from a {\bf (GLM-L)} model with $(A,a,b)$ for $a>0$ and $b\geq 0$, which satisfies the random design condition {\bf A\ref{a1'}'}. Also assume the high-dimensional regime as in the Proposition \ref{prop:5.1} and $\|\beta_0\|_2 = O(1)$. Then for any given $\epsilon>0$, there exist positive constants $\alpha_s$ and $\tau_{s}$ such that for $n \geq C\log (1/\epsilon)$, it holds with probability at least $1-\epsilon$ that
\begin{align}
	\left(\triangledown \mathcal{L}_n^{s}(\beta) - \triangledown \mathcal{L}^{s}_n(\beta_0)\right)^\top (\beta-\beta_0) &\geq
\alpha_s \|\beta-\beta_0\|_2^2 - \tau_s \frac{\log p}{n} \|\beta-\beta_0\|_1^2,\quad \forall \beta\in \B_2(1;\beta_0),\label{eq:prop5.2}
\end{align}
where the constant $C$ depends only on the model parameters.
\end{proposition}

The key observation to establish the RSC result \eqref{eq:prop5.2} is that the form of $\L_n^{s}(\beta)$ coincides with the negative log-likelihood function of a generalized linear model with the canonical link. Although $\P(T(\z)|\x)$ does not belong to the GLM family, the role of $T(\z)$ is limited in establishing the restricted strong convexity, and the proof for the generalized linear model with the canonical link in \citet{Negahban2012-bd} can be almost applied directly. More details are provided in Appendix \ref{supp_sec:prop:5.2}.

Now we state the following results regarding $\ell_1$ and $\ell_2$ error bounds. The first part of the theorem---for the $\ell_1$ and $\ell_2$ error bounds of the non-convex estimator---is a modification of Theorem 1 in \citet{Loh2017-jk}, where Theorem 1 in \citet{Loh2017-jk} established error bounds for a stationary point under the RSC condition with a different tolerance function $\tau_{n,p}(t) = \tau (\log p/n)t^2$. The $\ell_1$ and $\ell_2$ error bounds for the convex estimator can be obtained by applying Theorem 1 in \citet{Negahban2012-bd}. To apply Theorem 1 in \citet{Negahban2012-bd}, we show that the RSC condition \eqref{eq:prop5.2} implies the RSC condition in \citet{Negahban2012-bd}. We defer the detailed discussion to the Appendix \ref{supp_sec:thm5.1}.

%% Theorem 5.1
\begin{theorem}[$\ell_1$ and $\ell_2$ error bound]\label{thm:5.1}
Assume $\L_n^\ell$ and $\L_n^s$ satisfy the RSC conditions \eqref{eq:prop5.1} and \eqref{eq:prop5.2} and also assume the high-dimensional regime as in the Proposition \ref{prop:5.1}. 
\begin{enumerate}
	\item If $\lambda_\ell \geq 4 \max\{\|\triangledown \L_n^\ell(\beta_0)\|_\infty, \tau_\ell \sqrt{\frac{\log p}{n}}\}$, then, 
% Inequality 1
\begin{align}
&\|\widetilde{\beta}_\ell^H-\beta_0\|_2 \leq c_1 \frac{\sqrt{s_0}\lambda_\ell}{\alpha_{\ell}} \quad  and  \quad \|\widetilde{\beta}_\ell^H-\beta_0\|_1 \leq 4c_1 \frac{s_0\lambda_\ell}{\alpha_{\ell}}\label{eq:l1_l2_errbounnds-1}.
\end{align}
\item If $\lambda_s \geq 2 \|\triangledown \L_n^s(\beta_0)\|_\infty$ and $n \geq (32 \tau_s / \alpha_s) s_0 \log p$, then
% Inequality 2
\begin{align}
&\|\widehat{\beta}_s^H-\beta_0\|_2 \leq c_2 \frac{\sqrt{s_0}\lambda_s}{\alpha_s} \quad and \quad \|\widehat{\beta}_s^H-\beta_0\|_1 \leq 4c_2 \frac{s_0\lambda_s}{\alpha_s}\label{eq:l1_l2_errbounnds-2}
\end{align}
\end{enumerate}
Here $c_1, c_2>0$ are generic constants and $s_0 := \|\beta_0\|_0$.
\end{theorem}

In particular, if $\|\triangledown \L_n^\ell(\beta_0)\|_\infty,\|\triangledown \L_n^s(\beta_0)\|_\infty = O (\sqrt{\frac{\log p}{n}}) $ w.h.p, both estimators achieve the minimax-optimal error rates with the choices of $\lambda_\ell, \lambda_s \asymp \sqrt{\frac{\log p}{n}}$. In the following Corollary \ref{cor:3}, we summarize the results about error bounds for the noisy labels model.
\begin{corollary}\label{cor:3}
	Suppose a sample $(\x_i,\z_i)_{i=1}^n$ is from a {\bf (GLM)} with $(\log(1+\exp(\cdot)),g_{LN})$ and $\z_i \in \{0,1\}$. Assume the high-dimensional regime as in the Proposition \ref{prop:5.1}, the random design condition {\bf A\ref{a1'}'}, and the boundedness of the signal {\bf A\ref{a4}}. For the choices of $\lambda_\ell, \lambda_s \asymp \sqrt{\frac{\log p}{n}}$, it holds that
	\begin{align*}
	\|\widetilde{\beta}_\ell^H-\beta_0\|_2 \leq c_1 \frac{\sqrt{s_0}\lambda_\ell}{\alpha_{\ell}} \quad \mbox{and}\quad
		\|\widehat{\beta}_s^H-\beta_0\|_2 \leq c_2 \frac{\sqrt{s_0}\lambda_s}{\alpha_s}
	\end{align*}
	with probability at least $1-\epsilon$, given a sufficiently large sample size $ n \geq C (1/\epsilon)^{1/7}$, for a constant $C$ which only depends on the model parameters.
\end{corollary}

Notably, both estimators achieve the same optimal rates although there could still be a constant gap between the two estimators due to the different multipliers. We compare the performance of the two estimators empirically in Section \ref{sec:emp}.

%% High-D HT
\subsection{Hypothesis Testing}
Sparse estimators are known to have intractable limiting distributions even in the low-dimension regime \citep{Knight2000-lj}. Nonetheless, it is of interest to quantify the uncertainty in the obtained estimators and test the significance of features. In this section, we discuss how we can carry out a test using the point estimates discussed in the previous section. 
 
We take a de-biasing approach and obtain a one-step estimator whose direction is based on an estimating equation $\psi_n(\beta):= \frac{1}{n}\sum_{i=1}^n \psi(\x_i^\top \beta,\z_i)\x_i$. For a function $\psi : \R \rightarrow \R$, we consider $\psi$ satisfying the following two properties:
\begin{enumerate}
	\item $\psi$ has an expectation of zero at $\beta=\beta_0$:
	\begin{align}\label{eq:estimating_eq}
	\E_{\beta_0}[\psi(\x^\top  \beta_0,\z) \x ] = 0,
	\end{align}
	\item  The derivative of $\psi$ with respect to its first argument can be decomposed into the sum of $\psi'_R$ and $\psi'_I$, 
\begin{align}\label{eq:psi'}
\psi'(t,z) = \psi'_I(t)+\psi'_R(t,z)
\end{align}
where $\psi'_I$ and $\psi'_R$ satisfy $\psi'_I (t) > 0, \forall t$ and $ \E_{\beta_0}[\psi'_R(\x^\top \beta_0,\z)] = 0$.
\end{enumerate}
Two particular choices of $\psi$ that we will consider subsequently will be derivatives of the log-likelihood loss and the surrogate loss,
\begin{align*}
	\psi^\ell(\x^\top \beta,\z)&:=\, \ell'\,(\x^\top \beta,\z)=\{A'(h(\x^\top \beta))-\z\}h'(\x^\top \beta)\\	\psi^s(\x^\top \beta,\z)&:= \ell'_s(\x^\top \beta,\z) = A'(\x^\top \beta)-T(\z),
\end{align*}
where $\ell'(t,z)$ is the derivative (with respect to a linear predictor) of the log-likelihood loss defined in \eqref{def:Ll}, and $\ell'_s(t,z)$ is the derivative of the surrogate loss defined in \eqref{def:ls}. Obviously, both $\psi^\ell$ and $\psi^s$ satisfy \eqref{eq:estimating_eq}. Also, when $\psi = \psi^\ell$, the choices of $\psi'_I(t) = \rho_I(t)$ and $\psi'_R(t,z) = \rho_R(t,z)$, for $\rho_I$ and $\rho_R$ are defined in \eqref{eq:l''}, will satisfy \eqref{eq:psi'}. On the other hand, if $\psi = \psi^s$, the choices of $\psi_I'(t) = A''(t)$ and $\psi'_R \equiv 0$ will satisfy \eqref{eq:psi'}. 

The derivative of the estimating equation plays an important role in determining the asymptotic variances of the generalized estimating equation~(GEE) estimators \citep{Godambe1960-rw}. We define an empirical Jacobian matrix $\psi'_{I,n}(\beta)$ of $\E[\psi_n(\beta)]$ and the 
inverse of $\E[\psi'_{I,n}(\beta_0)]$ as 
\begin{align}
\psi'_{I,n}(\beta) := \frac{1}{n} \sum_{i=1}^n \psi_I'(\x_i^\top \beta)\x_i\x_i^\top  \quad \mbox{and} \quad
\Theta(\psi) := \E[\psi_I'(\x^\top \beta_0)\x\x^\top ]^{-1},
\end{align}
We note that the minimum eigenvalue of $\E[\psi_I'(\x^\top \beta_0)\x\x^\top ]$ can be shown to be bounded above by a positive constant under our assumptions, so that $\Theta(\psi)$ is well-defined (see Appendix \ref{pf:debiasing}).
%% Subsection -- de-biasing
\subsection{De-Biasing}
For an initial estimate $\widehat{\beta}$, we define a de-biased estimator using $\psi$ as follows, 
\begin{align}\label{def:debiased_est}
\widehat{\beta}^{\textnormal{db}}(\psi) := \widehat{\beta} -\widehat{\Theta}(\psi) \psi_n(\widehat{\beta}) \end{align}
which is a one-step estimator starting at $\widehat{\beta}$. Here a matrix $\widehat{\Theta}(\psi)$ is an approximation of $\Theta(\psi)$. 

We make the following assumption about the sparsity level of $\beta_0$ and $\Theta(\psi)$ similarly as in \citet{Van_de_Geer2014-if}. We define the column sparsity level of $\Theta(\psi)$ (except the diagonal entries) as $s_* :=\max_{1\leq j\leq p} \|\Theta(\psi)_{j,-j}\|_0$, and recall the definition $s_0 := \|\beta_0\|_0$.
\begin{assumption}\label{a6}
$s_*, s_0  = o(\sqrt{n}/\log p ),$ and $\|\X \Theta(\psi)_j\|_\infty = O_p(1), \forall j$
\end{assumption}

Also we we state conditions regarding the estimation equation $\psi$. In particular, we assume that $\psi$ and $\psi'$ are bounded and  $\psi'$ is also Lipschitz continuous with respect to its first argument. Precisely,
\begin{assumption}\label{a7}
	(Lipschitz continuity of $\psi'$ and boundedness of $\psi$ and $\psi'$)
	Both $\psi_R'$ and $\psi_I'$ are Lipschitz with respect to its first argument, i.e.,
	\begin{align*}
	|\psi_R'(a,z)-\psi_R'(a',z)|&\leq L_\psi |a-a'|, \forall z \quad \mbox{and} \quad
	|\psi_I'(a)-\psi_I'(a')| \leq L_\psi |a-a'|.
\end{align*}
In particular, $\psi'$ is Lipschitz with Lipschitz constant $2L_\psi$. Furthermore, there exists $C_\psi <\infty$ such that $\max\{ \|\psi\|_\infty,\|\psi'_I\|_\infty, \|\psi'_R\|_\infty \} \leq C_\psi $.
\end{assumption}

Now we state Theorem \ref{thm:5.2} which gives the asymptotic distributions of de-biased estimators.
\begin{theorem}\label{thm:5.2} 
Assume the random design condition {\bf A\ref{a1'}'}, {\bf  A\ref{a6}-A\ref{a7}}, and $\|\beta_0\|_2 = O(1)$. Suppose $\widehat{\Theta}(\psi)$ is chosen to satisfy $\|\widehat{\Theta}(\psi)_j-\Theta(\psi)_j\|_1 = o_p(\sqrt{\frac{1}{\log p}})$ and $\|e_j - \psi'_{I,n}(\widehat{\beta}) \widehat{\Theta}(\psi)_j\|_\infty = O_p(\sqrt{\frac{\log p}{n}}), \forall j.$ 
For an initial estimate $\widehat{\beta}$ satisfying $\ell_1$ and $\ell_2$ bounds $\|\widehat{\beta}-\beta_0\|_1 = O_p(s_0 \sqrt{\frac{\log p}{n}} )$ and $\|\widehat{\beta}-\beta_0\|_2^2 = O_p(s_0 \frac{\log p}{n} )$, and $\|\widehat{\beta} - \beta_0\|_1 /\|\widehat{\beta} - \beta_0 \|_2 = O(\sqrt{s_0})$ a.s.,
we have for any $j \in \{1,\dots,p\}$,
$$\sqrt{n}(\widehat{\beta}^{\textnormal{db}}(\psi)_j -\beta_{0j})/\sigma(\psi)_j = Z_j+o_p(1)$$
for $Z_j$ which converges weakly to a $\mathcal{N}(0,1)$ distribution and for
\begin{equation*}
\sigma(\psi)_j := \sqrt{ \Theta(\psi)_j^\top  \E[\psi(\x^\top \beta_0,\z)^2\x\x^\top ]\Theta(\psi)_j}.
\end{equation*}
Moreover, if the bound in {\bf A\ref{a6}} and the conditions in the theorem statement regarding $\widehat{\Theta}(\psi)_j$ hold uniformly in $j$, then the result also holds uniformly in $j$.
\end{theorem}
 We note that obtaining $\widehat{\Theta}(\psi)$ satisfying the conditions of Theorem \ref{thm:5.2} is possible by taking a similar approach as in \citet{Van_de_Geer2014-if} using node-wise regressions. In Appendix \ref{supp_sec:approx_Theta}, we provide more details about such construction. We also note that an initial estimate $\widehat{\beta}$ which satisfies the following cone condition, $\|(\widehat{\beta}-\beta_0)_{S^c}\|_1 \leq L\|(\widehat{\beta}-\beta_0)_S\|_1$ for some $L>0$ and $S \subseteq \{1,\dots,p\}$ such that $|S| = s_0$, also satisfies the $\ell_1$/$\ell_2$ ratio condition of the error vector in Theorem \ref{thm:5.2}, since $\|\widehat{\beta} - \beta_0\|_1 \leq (L+1)\sqrt{s_0} \|\widehat{\beta} - \beta_0\|_2$.
 The proof of Theorem \ref{thm:5.2} is deferred to Appendix \ref{pf:debiasing}. The main argument follows similar lines as in the proof of Theorem 3.1 in \citet{Van_de_Geer2014-if}, with additional arguments to handle the potential non-monotonicity of $\psi$, which can arise from a non-convex loss function (due to a non-canonical GLM). Finally, we state the following Corollary \ref{cor:4} for the asymptotic distributions of the de-biased estimators for the label noise model.

\begin{corollary}\label{cor:4}
Suppose we have a sample $(\x_i,\z_i)_{i=1}^n$ from a {\bf(GLM)} with parameters $(\log(1+\exp(\cdot)),g_{LN})$ and $\z_i \in \{0,1\}$. We assume the conditions of Proposition \ref{prop:5.1}. We also assume that $\widehat{\Theta}(\psi^\ell)$ and $\widehat{\Theta}(\psi^s)$ satisfy the conditions about $\widehat{\Theta}(\psi)$ in Theorem \ref{thm:5.2}, and {\bf A\ref{a6}} holds. We consider two de-biased estimators:
\begin{align*}
	\widehat{\beta}^{\textnormal{db}}_\ell  := \widetilde{\beta}_\ell^H-\widehat{\Theta}(\psi^\ell) \psi_n^\ell(\widetilde{\beta}_\ell^H) \quad \mbox{and} \quad
	\widehat{\beta}^{\textnormal{db}}_s     := \widehat{\beta}_s^H-\widehat{\Theta}(\psi^s) \psi_n^s(\widehat{\beta}_s^H).
\end{align*}
We then have, for any $j \in \{1,\dots,p\}$,
\begin{align*}
	\sqrt{n}(\widehat{\beta}^{\textnormal{db}}_{\ell,j} -\beta_{0j})/\sigma(\psi^\ell)_j = Z_j+o_p(1)\\
	\sqrt{n}(\widehat{\beta}^{\textnormal{db}}_{s,j} -\beta_{0j})/\sigma(\psi^s)_j = \widetilde{Z}_j+o_p(1)
\end{align*}
for $Z_j, \widetilde{Z}_j$ which converge weakly to a $\mathcal{N}(0,1)$ distribution and for,
\begin{align*}
\sigma(\psi^\ell)_j =\sqrt{ \I^\ell(\beta_0)^{-1}_{jj}} \quad \mbox{and} \quad \sigma(\psi^s)_j = \sqrt{ \I^{s}(\beta_0)^{-1}_{jj}}
\end{align*}
where $\I^\ell(\beta)$ and $\I^{s}(\beta)$ are defined in Proposition \ref{prop:4.2}.
\end{corollary}
The conditions about $\psi^\ell$ and $\psi^s$ can be checked similarly as in the proof of Corollary \ref{cor:2}. The rate conditions about the initial estimators can be checked by Corollary \ref{cor:3}. Also, it is well known that both $\widetilde{\beta}^H_\ell -\beta_0$ and $\widehat{\beta}^H_s -\beta_0$ belong to a cone $\{\Delta ; \|\Delta_S\|_1 \leq 3\|\Delta_{S^c}\|_1\}$ where $S \subseteq\{1,\dots,p\}$ is the support of $\beta_0$. We note that these results are analogous to Proposition \ref{prop:4.2} in the low-dimensional setting once penalization and de-biasing are introduced.

\section{Empirical Study}\label{sec:emp}
In this section, we present results about the empirical behavior of the non-convex likelihood-based estimator and the convex surrogate estimator. Our focus in this section is two-fold. First, we study the relative efficiency of the two estimators when different design matrices and noise rates are considered. In particular, we empirically demonstrate that the gap between $\mathcal{C}(\X)$ and $\mathcal{C}(W_z^{-1}W_y\X)$ captures well the impact of design $\X$ and noise rates $(\rho_0,\rho_1$) on the relative efficiency of the two estimators. Second, we study empirical performance of the two estimators in the low- and high-dimensional regimes, with and without regularization. As regards the estimation errors, the likelihood-based estimator is expected to perform better than the convex estimator in low dimensions. However, it is unclear whether this will continue to be true in high dimensions. Indeed, as we discuss hereafter, our simulation study shows that the convex estimator outperforms the likelihood-based estimator in terms of mean squared errors in sparse regimes, where signal strength is relatively low. 

\subsection{Methods}
Based on the regime of each simulation, we obtain non-sparse estimates $\widehat{\beta}_\ell$ and $\widehat{\beta}_s$ from \eqref{def:low-d_estimators} in the low-dimensional regime or sparse estimates $\widetilde{\beta}_\ell^H$ and $\widehat{\beta}_s^H$ from \eqref{def:high-d_estimators} in the high-dimensional regime. We recall that we define $\widetilde{\beta}_\ell^H$ as a stationary point of the optimization problem in \eqref{def:high-d_estimators} due to the non-convexity of $\L_n^\ell(\beta)$. For non-convex problems, we initialize coefficients at the null model where $\beta = [0,\dots,0]^\top $ if a problem is in the low-dimensional regime, and we use a local initialization using a convex estimate otherwise. To compare with the uncorrupted regime, the coefficient estimates $\widehat{\beta}_{\textnormal{ref}}$ and $\widehat{\beta}_{\textnormal{ref}}^H$ are computed using logistic or $\ell_1$-penalized logistic regression on the un-corrupted data $(\x_i,\y_i)_{i=1}^n$. 

In terms of optimization, we use the proximal gradient method combined with a back-tracking line search to solve optimization problems of \eqref{def:low-d_estimators} and 
\eqref{def:high-d_estimators}. This approach guarantees that iterates converge to a stationary point of the objective function if the objective function is non-convex and converge to an optimum in the convex case (e.g., Chapter 10 in \citealp{Beck2017-us}). For $\widehat{\beta}_{\textnormal{ref}}$ and $\widehat{\beta}_{\textnormal{ref}}^H$ we used the `glm()' function from R base package and `glmnet()' from R package \textbf{glmnet} respectively. 

\subsection{Impact of Design}\label{sec:emp-1_design}
To study the relative efficiency of the two estimators in various designs, we fix dimensions $(n=1000,p=10)$ and consider a mixture of multivariate normal distributions with varying distances between the two mixture components. We will demonstrate that increase in distance between the means of the two mixture components leads to an increase in the gap between  $\mathcal{C}(\X)$ and the perturbed column space $\mathcal{C}(W_z^{-1}W_y\X)$, and a larger gap between two subspaces is associated with greater efficiency differences in $\widehat{\beta}_\ell$ and $\widehat{\beta}_s$.

Now we describe our simulation set-up for this subsection. First, we generate a design matrix $\X = [\x_1^\top , \dots, \x_n^\top ]^\top $ by sampling each $\x_i$ from an equal mixture of multivariate Gaussian distribution centered at $\mu_1 = (d,\dots,d)$ and $\mu_2 = (-d,\dots,-d)$ with various $d$ and covariance matrix $\Sigma$ such that $\Sigma_{ij} = 0.2^{|i-j|}$. We let $\beta_0:= [1/\sqrt{p},\dots,1/\sqrt{p}]^\top $ so that $\|\beta_0\|_2^2 =1$. The true unobserved response $\y_i$ is drawn by $\y_i \sim \mbox{Ber}(p_{\beta_0}(\x_i))$ where $p_{\beta_0}(\x_i) = (1+\exp(-\x_i^\top \beta_0))^{-1}$, and a noisy label $\z_i$ is obtained by flipping $\y_i$ based on noise rates $\rho_0 = 10 \%$ and $\rho_1 = 5\%$.
The range of $d^2 = (0, \dots, 2.5)$ is considered so that $\textnormal{dist}^2 :=\|\mu_1-\mu_2\|_2^2 = 4pd^2$ varies from $0$ to $100$. When $\textnormal{dist} = 0$, $\x_i$ is from single Gaussian distribution, i.e., $\x_i \sim N(0, \Sigma), \forall i$. For each $d$, we repeat the experiment $B =10000$ times. At each $d$ and iteration $b=1,...,B$, 
we calculate
\begin{itemize}
    \item relative $\ell_2$ difference: $\textnormal{rd}(\I_n^\ell(\beta_0)_b,\I_n^s(\beta_0)_b):=\|I_p-{\I_n^\ell(\beta_0)_b}^{-1/2}\I_n^{s}(\beta_0)_b{\I_n^\ell(\beta_0)_b}^{-1/2}\|_2$,
    \item gap: $\widehat{\delta} (\mathcal{C}(\X_b), \mathcal{C}(W_z^{-1}W_y \X_b)) = \|\mathcal{P}_{\mathcal{C}(\X_b)} - \mathcal{P}_{\mathcal{C}(W_z^{-1}W_y\X_b)}\|_2 $ \citep{Kato2013-va},
    \item mean squared errors: $\textnormal{mse}^{\ell}_b := \|\widehat{\beta}_{\ell,b} - \beta_0\|_2^2$ and $\textnormal{mse}^s_{b} :=\|\widehat{\beta}_{s,b} - \beta_0\|_2^2$,
    \item asymptotic mean squared errors:\footnote[2]{ $\E\|\widehat{\beta} - \beta_0\|_2^2 = \textnormal{tr}(\E(\widehat{\beta} - \beta_0)(\widehat{\beta} - \beta_0)^\top ) \approx \textnormal{tr}(\I_n(\beta_0)^{-1}/n)$} $\textnormal{amse}^\ell_b :=\frac{\textnormal{tr}({\I_n^\ell}(\beta_0)_b^{-1})}{n}$ and $\textnormal{amse}^s_b :=\frac{\textnormal{tr}({\I_n^s}(\beta_0)_b^{-1})}{n}$,
\end{itemize}
where subscripts of $b$ mean corresponding quantities are from the $b$th experiment. We summarize results by taking an average of $B$ values. 

\begin{figure}[!bp]
	\centering
	\begin{minipage}[t]{0.45\linewidth}
		\includegraphics[width=\linewidth]{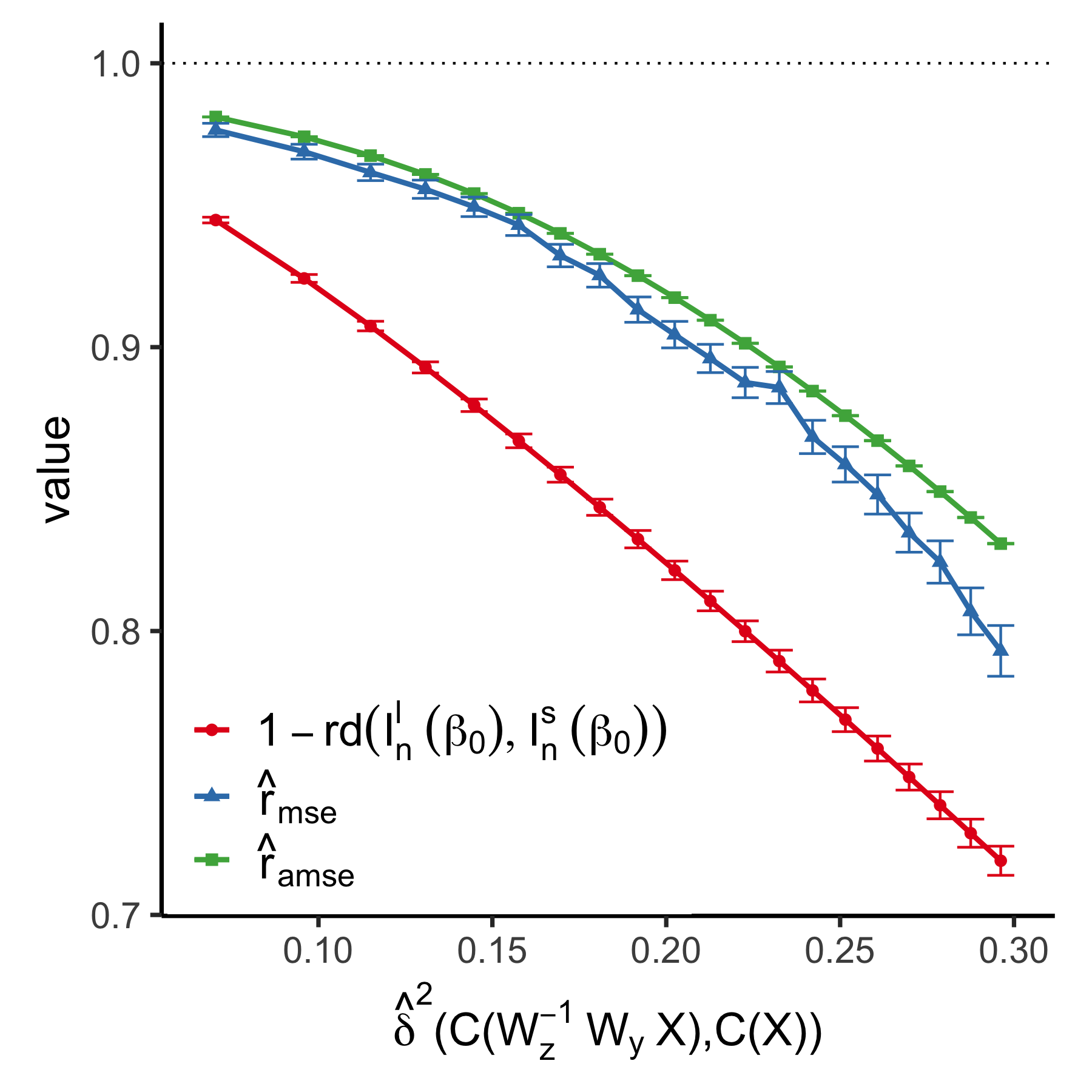}
		\caption{Ratios of mse and asymptotic mse and 1- relative $\ell_2$ difference with varying $gap^2$. Error bars refer to 1se.}\label{fig:fig2}
	\end{minipage}
	\quad
	\begin{minipage}[t]{0.45\linewidth}
		\includegraphics[width=\linewidth]{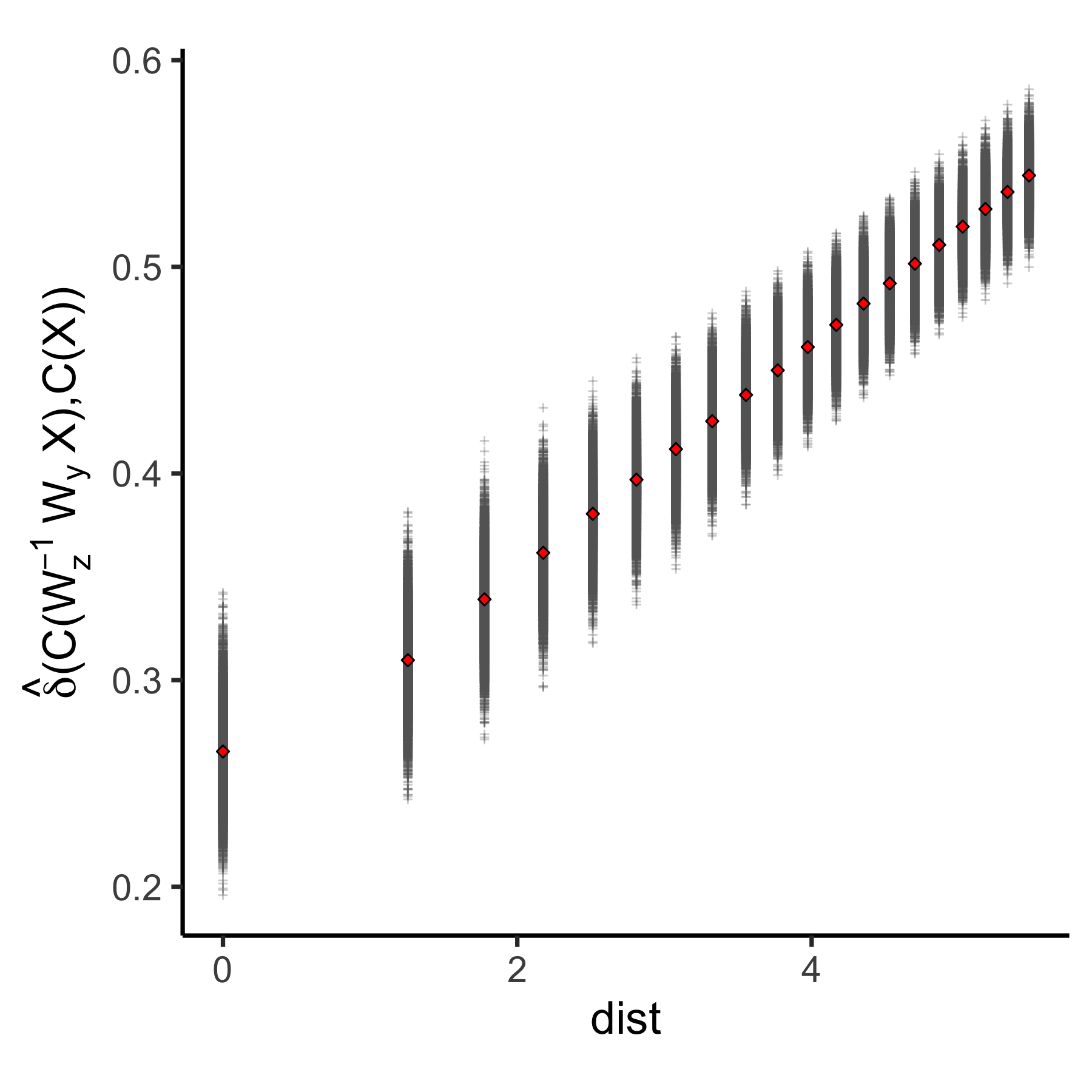}
		\caption{Plot of the distance between the means of two mixture distributions vs. the gap between the two column spaces.}\label{fig:fig3}
	\end{minipage}%
\end{figure}

To compare the efficiency of the two estimators, we calculate $\widehat{r}_{\textnormal{mse}}$, the ratio of estimated mean squared errors, and $\widehat{r}_{\textnormal{amse}}$, the ratio of asymptotic mean squared errors. More concretely, we let  $\widehat{r}_{\textnormal{mse}} := \frac{\overline{\textnormal{mse}^\ell}}{\overline{\textnormal{mse}^s}}$ and $\widehat{r}_{\textnormal{amse}} := \frac{\overline{\textnormal{amse}^\ell}}{\overline{\textnormal{amse}^s}}$. When $n$ is sufficiently large, $\widehat{r}_{\textnormal{mse}}$ is expected to be close to $\widehat{r}_{\textnormal{amse}}$, and both to be close to $r_{\textnormal{amse}} := \lim_n \frac{\E[\|\widehat{\beta}_\ell - \beta_0\|_2^2]}{\E[\|\widehat{\beta}_s - \beta_0\|_2^2]}$.
Note if the two estimators have the same efficiency, ratios will be close to 1. If the ratios are strictly less than 1, we can conclude that $\widehat{\beta}_\ell$ is more efficient than $\widehat{\beta}_s$.

Figure \ref{fig:fig2} plots the ratios of the mean squared errors and asymptotic mean squared errors, as well as $1- \textnormal{rd}(\I_n^\ell(\beta_0),\I_n^s(\beta_0))$ with varying $\textnormal{gap}^2$ values, i.e. $\widehat{\delta}^2 (\mathcal{C}(\X), \mathcal{C}(W_z^{-1}W_y \X))$. We recall that $1-\textnormal{rd}(\I_n^\ell(\beta_0),\I_n^s(\beta_0)) =1 $ iff two estimators have the same asymptotic efficiency, i.e. $\I_n^\ell(\beta_0)=\I_n^s(\beta_0)$, and $1-\textnormal{rd}(\I_n^\ell(\beta_0), \I_n^s(\beta_0))<1$ if $\I_n^\ell(\beta_0)\succ \I_n^s(\beta_0)$.
We see from Figure \ref{fig:fig2} that $1-\textnormal{rd}(\I_n^\ell(\beta_0),\I_n^s(\beta_0))$ linearly decreases with the $\textnormal{gap}^2$, which aligns with the result of Corollary \ref{cor:1}. Also, the efficiency of the surrogate estimator worsens compared to the likelihood-based estimator as the gap increases, but not in the linear fashion as in the case of $1-\textnormal{rd}(\I_n^\ell(\beta_0),\I_n^s(\beta_0))$. Unlike the relative $\ell_2$ difference where we associated the quantity with variance ratio of the two estimators with respect to a particular direction $u$, variance ratios in all directions are considered in $r_{\textnormal{amse}}$ since $\widehat{r}_{\textnormal{amse}} = \textnormal{tr}(\I_n^\ell(\beta_0)^{-1})/\textnormal{tr}(\I_n^s(\beta_0)^{-1})$. Figure \ref{fig:fig3} plots the gap $\widehat{\delta}(\mathcal{C}(\X), \mathcal{C} (W_z^{-1}W_y \X))$ as functions of $\textnormal{dist}= \|\mu_1 - \mu_2\|_2$. We see that the gap between two subspaces increases as the distance between two mixture components increases. 

\subsection{Impact of Noise Rates}\label{sec:emp-2_NL}

In this section, we study the relative efficiency of the two estimators with varying noise rates. From \eqref{eq:Wz1Wy}, for a given distribution $\x_i^\top \beta_0$, higher noise rates lead to a larger perturbation from $\mathcal{C}(\X)$ to $\mathcal{C}(W_z^{-1}W_y \X)$, and therefore a larger efficiency gap between the non-convex and the convex estimator is expected. 

The distribution of the $\x_i^\top \beta_0$ plays a role in determining the gap between the two subspaces. For example, two samples from the models with the same noise rates can have very different gap values if the distributions of $\X$ are different. To illustrate this point concretely, suppose that two samples are from the models with the same rates of $\rho_0 = 0\%$ and $\rho_1=20\%$ but from different $\X$, where most $\x_i^\top\beta_0$ are negative in the first model but most $\x_i^\top\beta_0$ are positive in the second model. A larger amount of the variance mis-specification using the convex approach will happen in the region when $\x_i^\top \beta_0$ is positive, since only the positive labels are flipped into the negative labels (Figure \ref{fig:fig_vary_varz}). This causes $W_z^{-1}W_y$ to deviate further from an identity matrix, and therefore, the gap between $\mathcal{C}(\X)$ to $\mathcal{C}(W_z^{-1}W_y \X)$ tends to be much larger for the second model despite the noise rates being the same in both models. 

\begin{figure}[b]
  \begin{minipage}[c]{0.45\textwidth}
    \includegraphics[width=\textwidth]{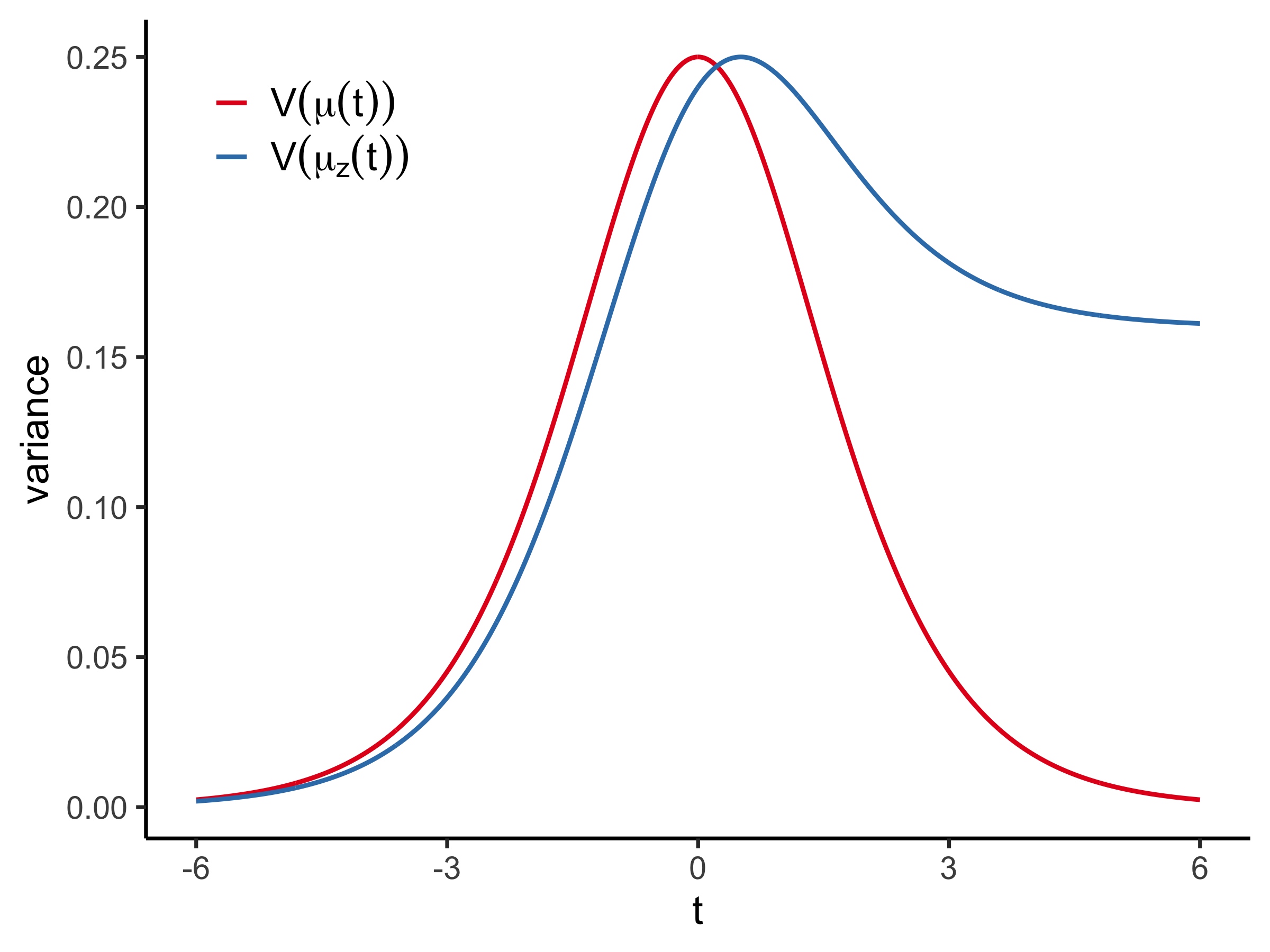}
  \end{minipage}\hfill
  \quad\quad
  \begin{minipage}[c]{0.5\textwidth}
    \caption{The plot of the variance of $\y$, $\text{Var}(\y|\x) = \mathcal{V}(\mu(\x^\top\beta))$ and $\z$, $\text{Var}(\z|\x) =\mathcal{V}_z(\mu(\x^\top\beta))= \mathcal{V}((1-\rho_1-\rho_0)\mu(\x^\top\beta)+\rho_0)$, for varying $t = \x^\top\beta$ with $\rho_1 = 0.2$ and $\rho_0=0$. The difference between the two variances is larger when $t\gg 0$.}\label{fig:fig_vary_varz}
  \end{minipage}
\end{figure}

We empirically study the impact of the noise rates using a similar simulation set-up as in the previous subsection \ref{sec:emp-1_design}, except that we fix $d$---the varying parameter in the previous subsection---to be $d=2/\sqrt{10}$ and instead vary the noise rates from 5\% to 20\%. In particular, the distribution of $\x_i^\top \beta_0$ is the same for all experiments in this section, which is an equal mixture of $\mathcal{N}(\beta_0^\top \mu_1,\beta_0^\top \Sigma \beta_0 )$ and $\mathcal{N}(\beta_0^\top \mu_2,\beta_0^\top \Sigma \beta_0 )$. We consider two settings for the noise rates to cover both symmetric and non-symmetric noise rates cases, where in the first setting we fix $\rho_1 = 5\%$ and vary $\rho_0$ from $5$ to $20\%$, and in the second setting we let $\rho=\rho_0=\rho_1$ and vary $\rho$ from $5$ to $20\%$. 
That is, for each noise rate setting and $b=1,...,B=10000$, we generate $(\x_i,\y_i,\z_i)_{i=1}^n$ and obtain the gap value $\widehat{\delta} (\mathcal{C}(\X_b), \mathcal{C}(W_z^{-1}W_y \X_b)) = \|\mathcal{P}_{\mathcal{C}(\X_b)} - \mathcal{P}_{\mathcal{C}(W_z^{-1}W_y\X_b)}\|_2 $ and the mean squared errors $\textnormal{mse}^{\ell}_b := \|\widehat{\beta}_{\ell,b} - \beta_0\|_2^2$ and $\textnormal{mse}^s_{b} :=\|\widehat{\beta}_{s,b} - \beta_0\|_2^2$, in the same way as we did in the previous subsection. The ratios of the estimated mean squared errors, $\widehat{r}_{\textnormal{mse}}:= \frac{\overline{\textnormal{mse}^\ell}}{\overline{\textnormal{mse}^s}}$ are then summarized over $B$ values at each noise rate setting. 

\begin{figure}[t]
\centering
\begin{subfigure}[h]{0.45\linewidth}
\includegraphics[width=\linewidth]{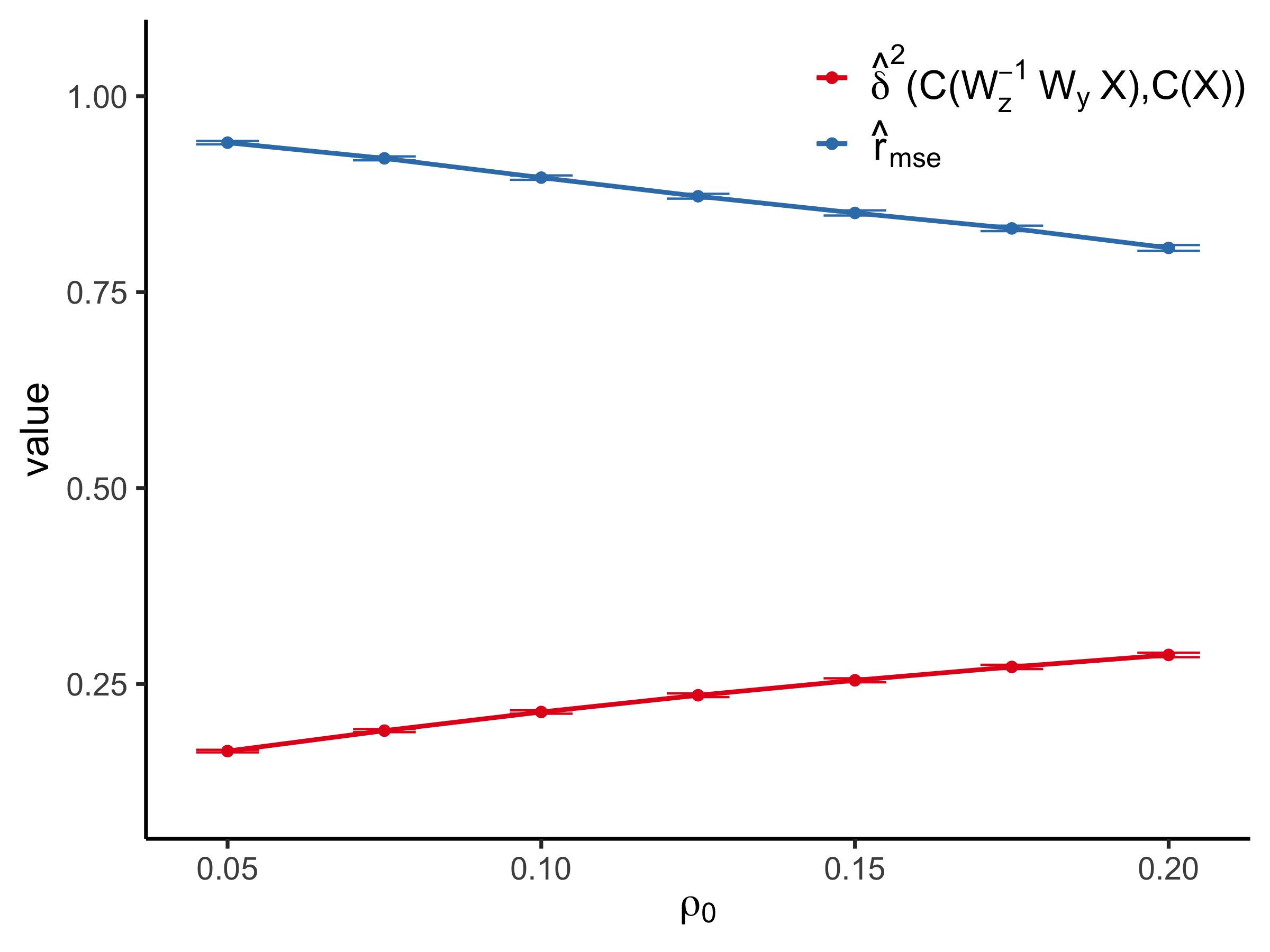}
\caption{Varying $\rho_0$ and fixed $\rho_1=5\%$}
\end{subfigure}
\quad
\begin{subfigure}[h]{0.45\linewidth}
\includegraphics[width=\linewidth]{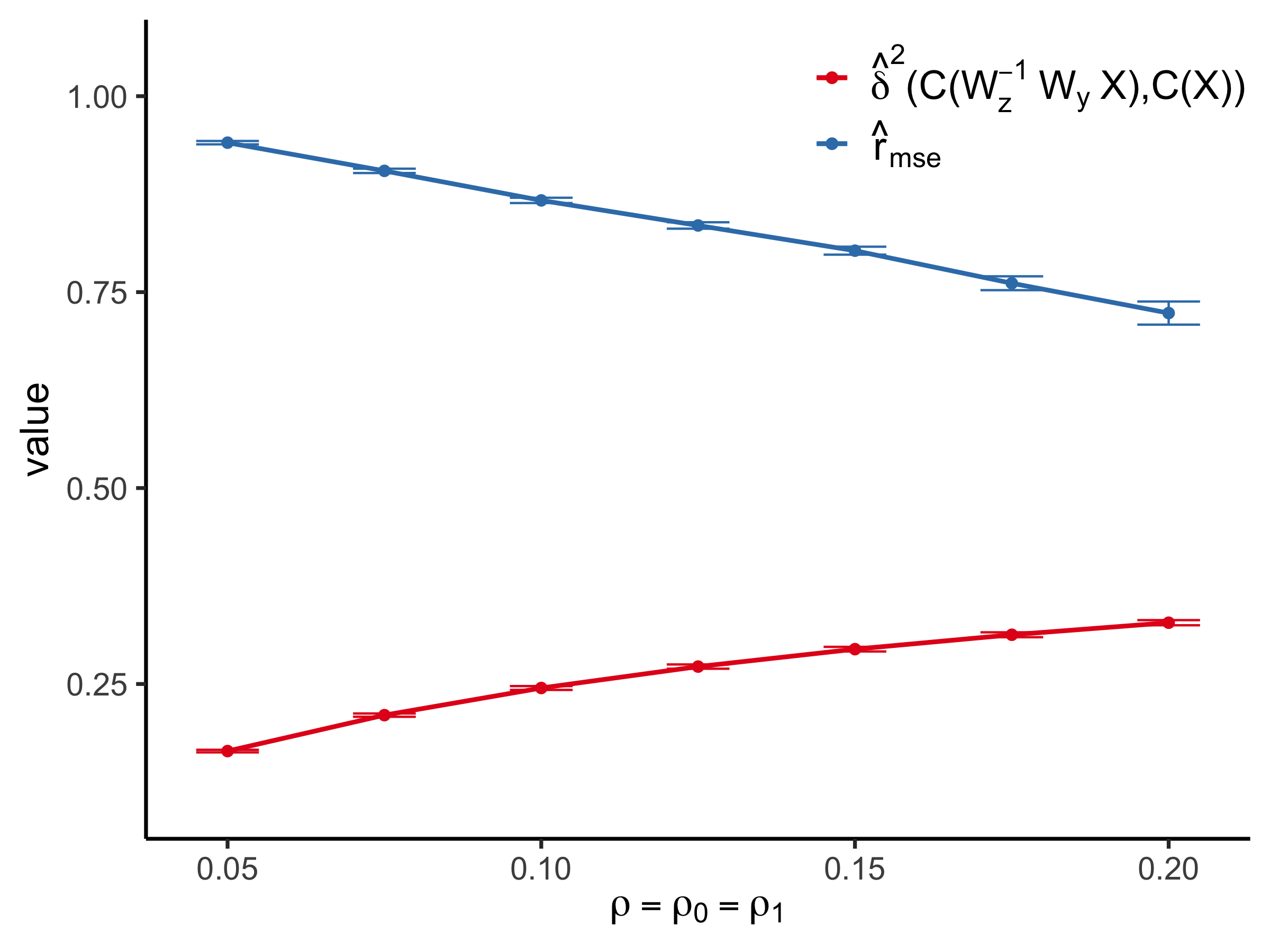}
\caption{Varying $\rho=\rho_0=\rho_1$}
\end{subfigure}%
\caption{Plots of ratios of mse and $\textnormal{gap}^2$ with (a) varying $\rho_0$ and fixed $\rho_1=5\%$ and (b) varying $\rho=\rho_0=\rho_1$. The error bars denote one standard error.}\label{fig:fig5}
\end{figure}

Figure \ref{fig:fig5} plots the ratios of the squared gaps $\widehat{\delta}^2(\mathcal{C}(\X), \mathcal{C} (W_z^{-1}W_y \X))$ and the mean squared errors as functions of noise rates. We note that $\widehat{r}_{\textnormal{mse}}<1$ implies the non-convex estimator had a smaller mse than the convex estimator. It can be seen from the plots the gap between two subspaces increases as the noise rates increase in both settings, and the convex approach performs worse than the non-convex approach.

\subsection{Comparison of Estimation Errors in Low- and High-Dimensional Settings}\label{sec:emp-3_Est}
To study estimation performances of the non-convex and convex approaches, we consider the following two regimes: (i) fixed $p = 10$ and growing $n$; (ii) growing $(n,p)$ with $p = n$. Also, we consider two noise settings, where we use $\rho_1 = 5\%$ and $\rho_0 = 10\%$ for the first noise setting (low-noise) and double the noise rates for the second noise level setting (high-noise).

A sample $\x_i \in \R^p$ is generated from multivariate gaussian distribution $\mathcal{N}(0, \Sigma)$ where $\Sigma \in \R^{p\times p}$ is given as $\Sigma_{i,j} = C_\Sigma (0.2)^{|i-j|}$, where $C_\Sigma$ is chosen so that $\textnormal{Var}(\x_i^\top \beta_0) = 5$. The sample size $n$ varies from $1000$ to $5000$ where values in between are interpolated in a log scale. In both regimes, we first let $10$ features be active ($s=10$) and true parameter be $\beta_0:= [\underbrace{1,\dots,1}_{\text{s/2}},\underbrace{-1\dots,-1}_{\text{s/2}}, 0,\dots,0]$. The true observed responses $\y_i$ and the noisy labels $\z_i$ are generated in the same way as in Section \ref{sec:emp-1_design}.
Each experiment in the low-dimensional regimes is repeated $B=300$ times and $B=50$ times in the high-dimensional regimes. The mean and standard errors of $B$ trials are reported in Figure \ref{fig:emp_study1}. Tuning parameter $\lambda$ needs to be chosen for the high-dimensional estimators. We choose $\lambda$ in each simulation based on the testing loss from 5-fold cross validation. 
\begin{figure}[t]
    \centering
    \includegraphics[width=1\linewidth]{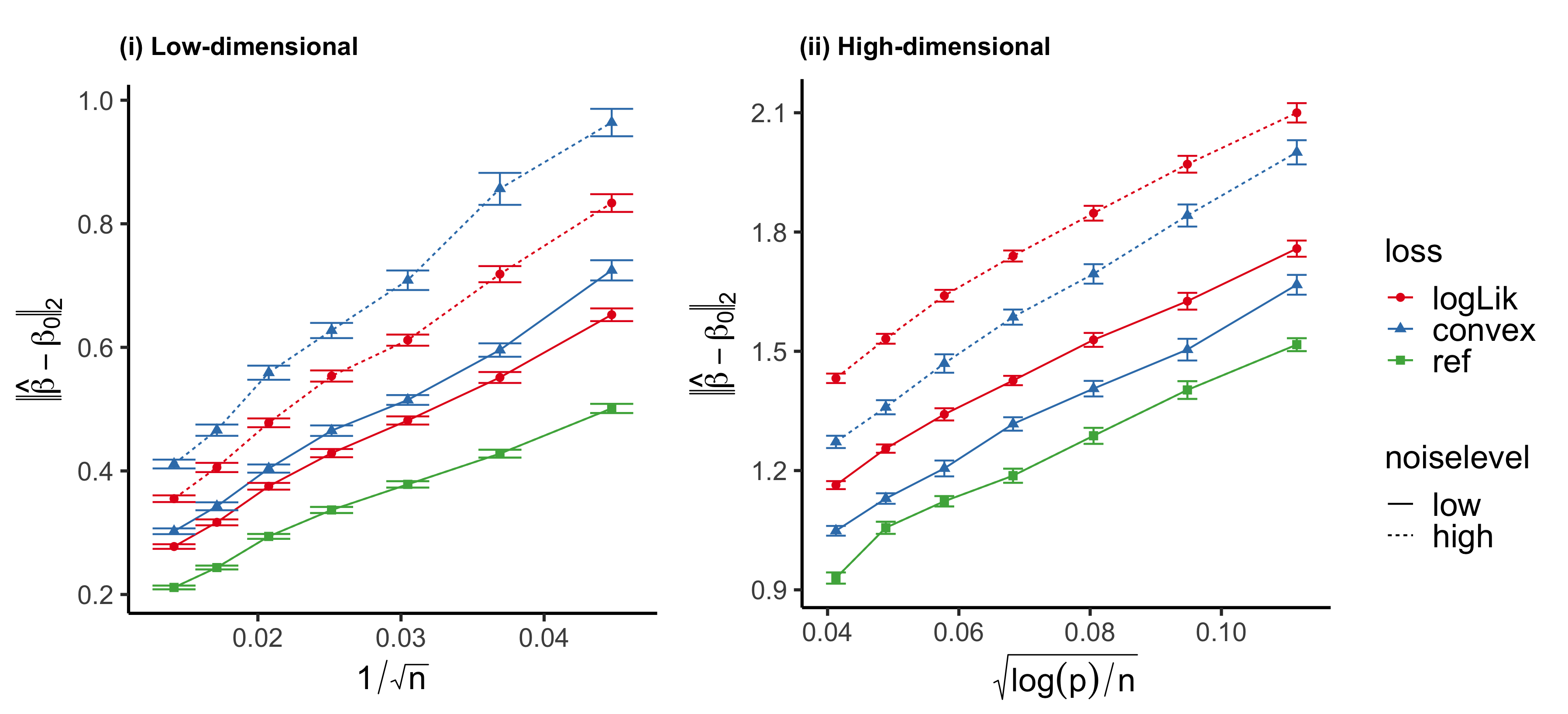}
    \caption{Comparison of the log-likelihood and surrogate loss based estimators in the low and high-dimensional regimes. Reference loss (ref) refers to the logistic loss when clean data is available.}
    \label{fig:emp_study1}
\end{figure}

Figure \ref{fig:emp_study1} shows the comparison results of the non-sparse and sparse estimators in both low and high-dimensional regimes. Not surprisingly, the likelihood-based estimator performs uniformly better than the convex estimator in the low-dimensional regime without any regularization in the both noise settings. The loss of efficiency by using a convex surrogate loss appears to be relatively small when the noise level is low. The performance of the surrogate estimator worsens when noise rates increase since the squared gap $\widehat{\delta}^2(\mathcal{C}(\X), \mathcal{C}(W_z^{-1}W_y \X))$ increases, which agrees with the results in Section \ref{sec:emp-2_NL}.
On the contrary, the convex surrogate estimator appears to perform uniformly better than the likelihood-based loss in the high-dimensional setting.

It is well known that when no regularization is introduced, the likelihood function is the best function to optimize since the procedure results in the smallest asymptotic variance matrix (in regular problems). Its optimality (more precisely, the optimality of the score function) was also argued in classical estimating equation theory, where the score function is shown to be the best estimating equation function in the sense of minimizing the asymptotic variance \citep{Godambe1960-rw}. 
The surrogate loss $\L_n^s(\beta)$ has a stronger curvature than $\L_n^\ell(\beta)$; in fact the curvature of $\L_n^s(\beta)$ is the same as $\L_n^c(\beta)$, the logistic loss from clean data. However, $\triangledown \L_n^s(\beta)$ has also a larger variance than $\triangledown \L_n^\ell(\beta)$ due to noise in the responses, resulting in the larger asymptotic variance matrix. We conjecture that in a penalized problem, especially when signal is relatively small compared to noise, regularization plays a role in reducing the variability in $\triangledown \L_n^s(\beta)$ which leads to the better performance of the convex estimator in some cases.

To test our conjecture, we carry out an additional set of simulation. The set-up of the simulation mainly follows the set-up in Section \ref{sec:emp-1_design} except we choose dimensions $(n=10000, p=20)$, fix $d = 3/\sqrt{p}$, and instead let the $\ell_1/\ell_2$ ratio of the true signal $\beta_0$ vary. More concretely, we fix $\|\beta_0\|_2 = 1$, and consider different sparsity levels $s$ of $\beta_0 = [\underbrace{1/\sqrt{s},\dots,1/\sqrt{s}}_\text{s},\underbrace{0,\dots,0}_{\text{p-s}}]$ so that $\|\beta_0\|_1/\|\beta_0\|_2$ varies from $1$ to $\sqrt{p}$. For each $s = 1,\dots,p$, we obtain both non-sparse and sparse estimates. In the case of sparse estimates, tuning is performed by minimizing test loss on a test set of size $n$. At each $s$, the experiment is repeated $B=10000$ times. Similarly as in Section \ref{sec:emp-1_design}, we calculate the two mean squared ratios each for the non-sparse and sparse estimates, 
    $\widehat{r}_{\textnormal{mse}} := \frac{\overline{\textnormal{mse}^\ell}}{\overline{\textnormal{mse}^s}}$ and $\widehat{r}_{\textnormal{smse}} := \frac{\overline{\textnormal{smse}^\ell}}{\overline{\textnormal{smse}^s}}$,
where we recall the definitions of $\textnormal{mse}^{\ell}_b := \|\widehat{\beta}_{\ell,b} - \beta_0\|_2^2$ and $\textnormal{mse}^s_{b} :=\|\widehat{\beta}_{s,b} - \beta_0\|_2^2$ and define $\textnormal{smse}^\ell_b := \|\widetilde{\beta}_{\ell,b}^H - \beta_0\|_2^2$ and $\textnormal{smse}^s_b := \|\widehat{\beta}_{s,b}^H - \beta_0\|_2^2$.
\begin{figure}[!tp]
\centering
\begin{subfigure}[h]{0.45\linewidth}
\includegraphics[width=\linewidth]{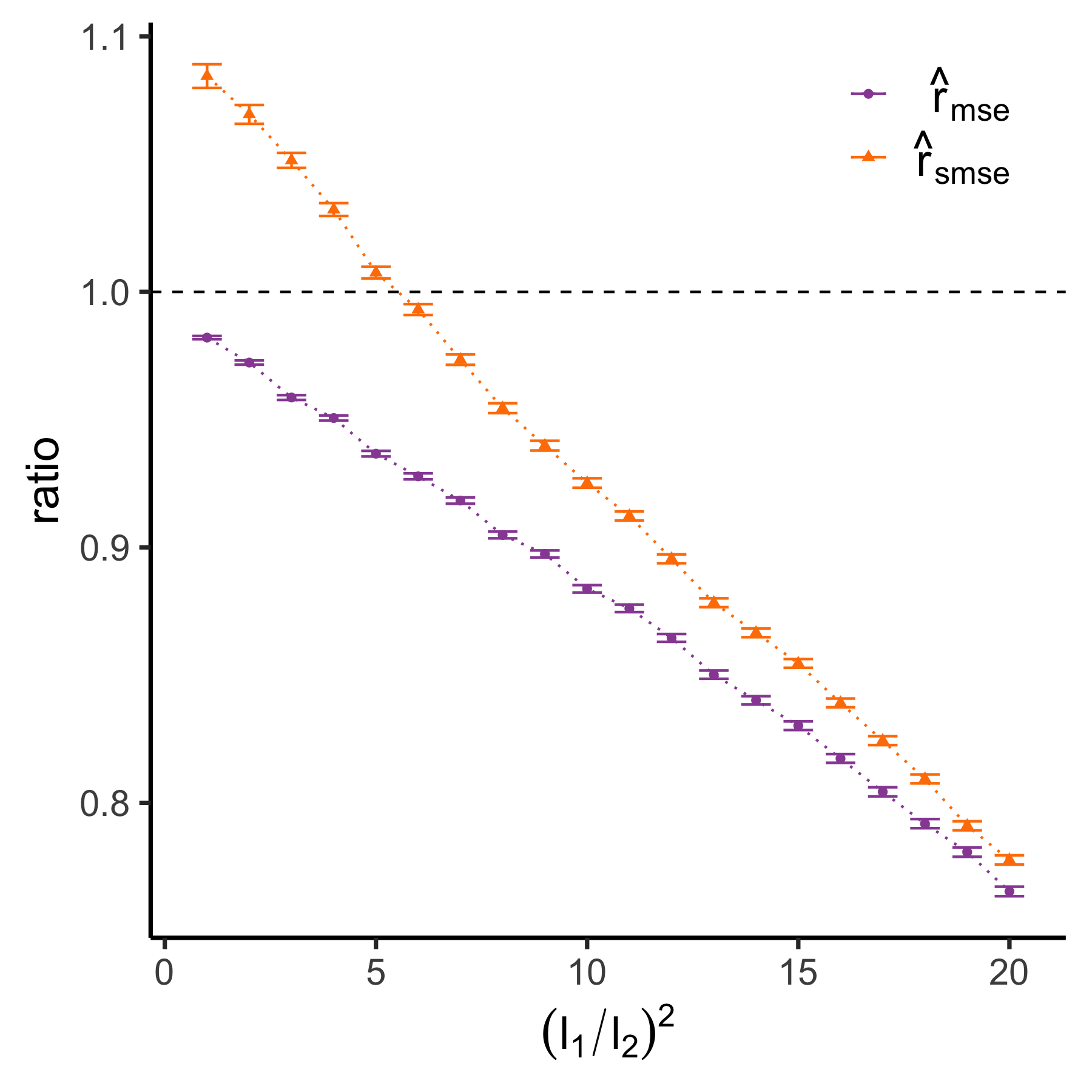}
\caption{$\widehat{r}_{\textnormal{mse}}$ and $\widehat{r}_{\textnormal{smse}}$}
\end{subfigure}
\quad
\begin{subfigure}[h]{0.45\linewidth}
\includegraphics[width=\linewidth]{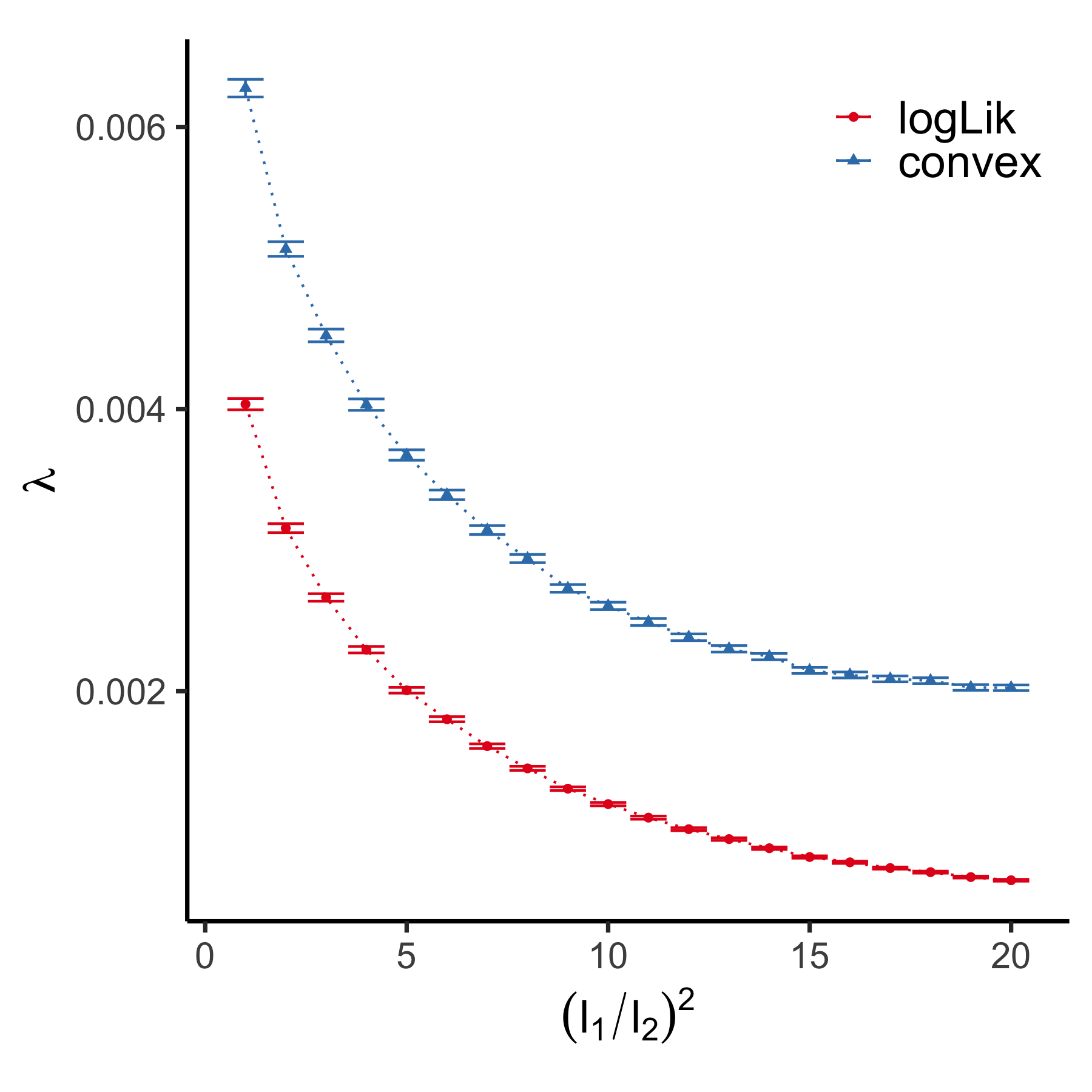}
\caption{$\lambda$ for the likelihood and the convex losses}
\end{subfigure}%
\caption{Plots of (a) $\widehat{r}_{\textnormal{mse}}$ and $\widehat{r}_{\textnormal{smse}}$ and (b) the amount of regularization chosen ($\lambda$) for the likelihood and the convex losses with varying $\ell_1/\ell_2$ ratios. The error bars
denote one standard error. The dotted line represents when the ratio = 1, where the likelihood-based estimator and the convex estimator have the same efficiency.}\label{fig:fig7}
\end{figure}

As we can see from Figure \ref{fig:fig7}, the likelihood-based estimator is always more efficient than the convex estimator in the case of non-sparse estimates. On the other hand, for the sparse estimators, the convex estimator outperforms the likelihood-based estimator when the $\ell_1/\ell_2$ ratio is small. As the $\ell_1/\ell_2$ ratio increases, we see from plot (b) that the amount of regularization decreases and the likelihood-based estimator starts to perform better than the convex estimator.

\subsection{Comparison of Empirical Confidence Interval Coverage Rates in Low- and High-Dimensional Settings}\label{sec:emp-4_Test}

In this section, we compare the empirical coverage rates of confidence intervals from the convex and non-convex methods. Similarly as in the previous sections, we consider the low and high-dimensional regimes, where we let $(n,p,s) = (2000,20,10)$ in the low-dimensional regime and $(n,p,s) = (500,1000,10)$ in the high-dimensional regime. The noise rates $(\rho_0,\rho_1)$, true parameter $\beta_0$, and features $(\x_i)_{i=1}^n$ are set up in the same way as in the previous section \ref{sec:emp-1_design}.
 
We construct $1-\alpha$ confidence intervals in low- and high-dimensional settings, where we obtain both non-sparse and sparse estimates in the low-dimensional setting and obtain sparse estimates in the high-dimensional setting. The nominal level $1-\alpha$ is set to be $0.95$. Confidence intervals for the non-sparse estimators are constructed based on their asymptotic normality results in Proposition \ref{prop:4.2}. For sparse estimators, we first de-bias the estimates and construct confidence intervals based on the asymptotic normality results for de-biased estimators in Proposition \ref{cor:4}. For the estimation of inverse Hessian matrices which are needed for de-biasing, the inverse matrix of the Hessian matrix and an approximate inverse matrix based on node-wise regressions are used in the low- and high- dimensional settings respectively. Mean empirical coverage rates for all features (all), active features (nzero), and non-active features (zero), as well as the mean lengths of the confidence intervals (CI length) are reported in Table \ref{table:CIcoverage} based on 100 realizations of confidence intervals. 
\begin{table}[!t]
\begin{center}
\subcaption{dimensions: (n=2000, p =20)}
\begin{tabular}{l c c c c}
\toprule
 & coverage (all) & coverage (nzero) & coverage (zero) & CI length \\
\midrule
logLik            & $0.951$   & $0.951$   & $0.951$   & $0.362$   \\
                  & $(0.005)$ & $(0.008)$ & $(0.007)$ & $(0.001)$ \\
convex            & $0.962$   & $0.961$   & $0.962$   & $0.387$   \\
                  & $(0.005)$ & $(0.007)$ & $(0.006)$ & $(0.002)$ \\
logLik (debiased) & $0.944$   & $0.942$   & $0.946$   & $0.340$   \\
                  & $(0.006)$ & $(0.009)$ & $(0.008)$ & $(0.001)$ \\
convex (debiased) & $0.946$   & $0.938$   & $0.953$   & $0.360$   \\
                  & $(0.005)$ & $(0.009)$ & $(0.006)$ & $(0.002)$ \\
\bottomrule
\end{tabular}\bigskip
\subcaption{dimensions: (n=500, p =1000)}
\begin{tabular}{l c c c c}
\toprule
 & coverage (all) & coverage (nzero) & coverage (zero) & CI length \\
\midrule
logLik (debiased) & $0.965$   & $0.900$   & $0.965$   & $0.368$   \\
                  & $(0.001)$ & $(0.009)$ & $(0.001)$ & $(0.001)$ \\
convex (debiased) & $0.964$   & $0.925$   & $0.964$   & $0.388$   \\
                  & $(0.001)$ & $(0.008)$ & $(0.001)$ & $(0.001)$ \\
\bottomrule
\end{tabular}
\end{center}
\caption{Average coverage rates for 100 confidence interval realizations of a noisy labels model in low and high dimensional settings for all, active (nzero), and non-active (zero) features. Numbers in parentheses represent one standard error.}\label{table:CIcoverage}
\end{table}

The overall coverage rates appear to be good for the both convex and non-convex methods in all regimes where about 95\% of the constructed intervals contained the true parameters. In all settings, both likelihood based methods (without and with penalization) result in confidence intervals with shorter lengths than those from the convex methods, which agrees with the results in Proposition \ref{prop:4.2} and \ref{cor:4}. Confidence intervals from de-biased estimators tend to be less conservative than those from the non-sparse estimators, which seem to cause lower empirical coverage rates than the nominal level for non-zero coefficients. Similar observations have also been made in \citet{Dezeure2015-jh}.

We conclude this section by remarking on the relative performance of the two approaches and some practical implications. In terms of estimation errors, the non-convex estimator performed better than the convex estimator in unpenalized, low-dimensional settings with large $n$. Also, in penalized schemes, the non-convex likelihood approach empirically performed better when the true model was not highly sparse, as shown in Figure \ref{fig:fig7}. Confidence intervals from both methods showed good empirical coverage rates and the average confidence interval length was shorter for the likelihood approach than for the convex approach. Therefore, the non-convex approach is preferred in unpenalized and large $n$ regimes, or in penalized and relatively not highly sparse regimes ($\ell_1/\ell_2$ ratio over $\sqrt{.3p}$ from Figure \ref{fig:fig7}), potentially with multiple initializations. The convex approach provides a viable alternative to the non-convex likelihood-based approach in all settings, but it can be particularly advantageous in settings where the true model is highly sparse, or when running optimization algorithms multiple times with various initializations is computationally challenging. The convex objective is computationally attractive to work with since every stationary point of the objective is a global minimum.

\section{Application to Beta-glucosidase Protein Data}\label{sec:real_data}
In this section, we describe an application of our non-convex and convex methods to beta-glucosidase (BGL) protein sequence data. Beta-glucosidase is a key enzyme present in cellulase which converts cellobiose to glucose during cellulose hydrolysis. The BGL enzyme protein plays a significant role in bioethanol production \citep{Singhania2013-af}. Due to its industrial importance, it is of great interest to understand the effects of mutations of the protein and design a protein with improved functionality. 

The data set we analyze is a positive and unlabeled beta-glucosidase protein sequence data set generated in the Romero Lab \citep{Romero2015-kl}. Large-scale data were generated by deep mutational scanning (DMS) method, which applies the high-throughput screening method to sort out functional protein variants. Screening is based only on the enzyme functionality of a sequence. Unlabeled sequences from the initial library whose associated functionality is unknown are obtained together with screened sequences to be positive. %
The data consists of $n_\ell = 2533388$ functional (positive label) and $n_u = 1500277$ unlabeled sequences.\footnote[2]{The raw data is available in https://github.com/RomeroLab/seq-fcn-data.git} A sequence consists of 500 positions which takes one of $21$ discrete values which correspond to 20 amino acid letter codes plus an additional letter for the alignment gap. 
From an alternative experiment, the prevalence of functional sequences in the unlabeled data set is known to be $0.35$. This data set is previously analyzed in \citet{Song2018-we} where the likelihood based approach was taken with the $\ell_1$ penalization to obtain a sparse estimate.
We obtained $p=3097$ features using one-hot encoding of the sequences. Because each sequence contains only a few mutations, we obtained a sparse design matrix $\X \in \R^{n \times p}$ by taking the amino acid levels in the WT sequence as the baseline levels. We note that the number of features ($p = 3097$) in the model is approximately one third of the number of maximum possible features ($10000 =20 \times 500$). Since the sequences in the data set are local sequences around the wild-type sequence, some mutations were never observed.

We apply both convex and non-convex methods to the data set to estimate each mutation effect of the BGL sequence. Estimated coefficients are obtained by fitting the model using all sequence examples. In addition, to compare predictive performance of the two methods, we split the data set into training and test sets using 90\% and 10\% of the sequence examples. The model is then refitted using .1\%, 1\%, 10\%, and 100\% of the examples in the training set to compare the performance of the two methods at various sample sizes. For a performance metric, we use the area under the ROC curve (AUC) for the comparison of classification performances. However, for positive and unlabeled data, the ROC curve and AUC value calculated using the observed labels as the responses are biased for the ROC curve and AUC value for the unobserved true responses \citep{Jain2017-cu}. Following the approach in  \citet{Jain2017-cu}, we report the corrected ROC curve and AUC values.

\begin{figure}[!bp]
	\centering
	\begin{subfigure}[h]{0.4\linewidth}
		\includegraphics[width=\linewidth]{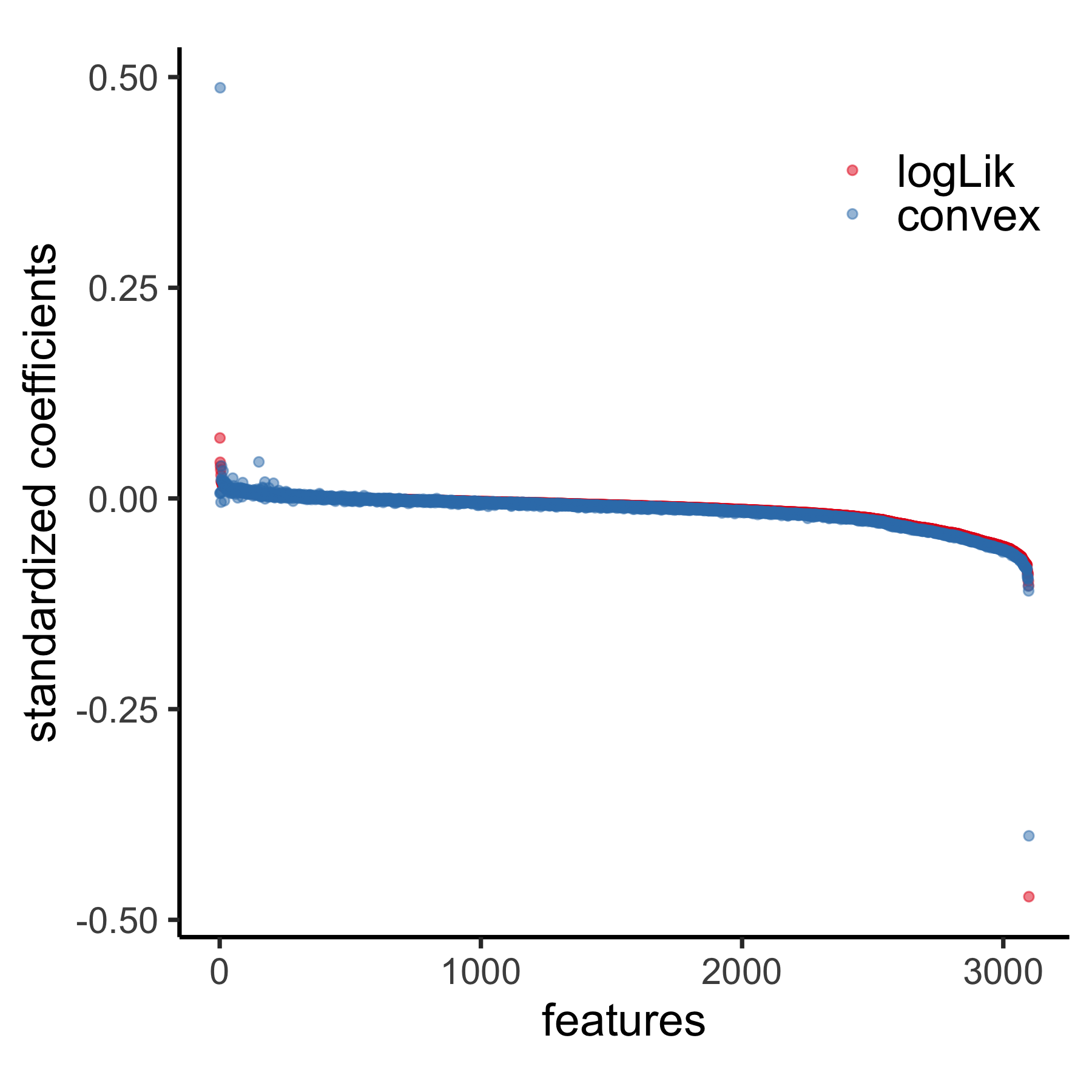}
		\caption{Estimated coefficients}
	\end{subfigure}
	\begin{subfigure}[h]{0.4\linewidth}
		\includegraphics[width=\linewidth]{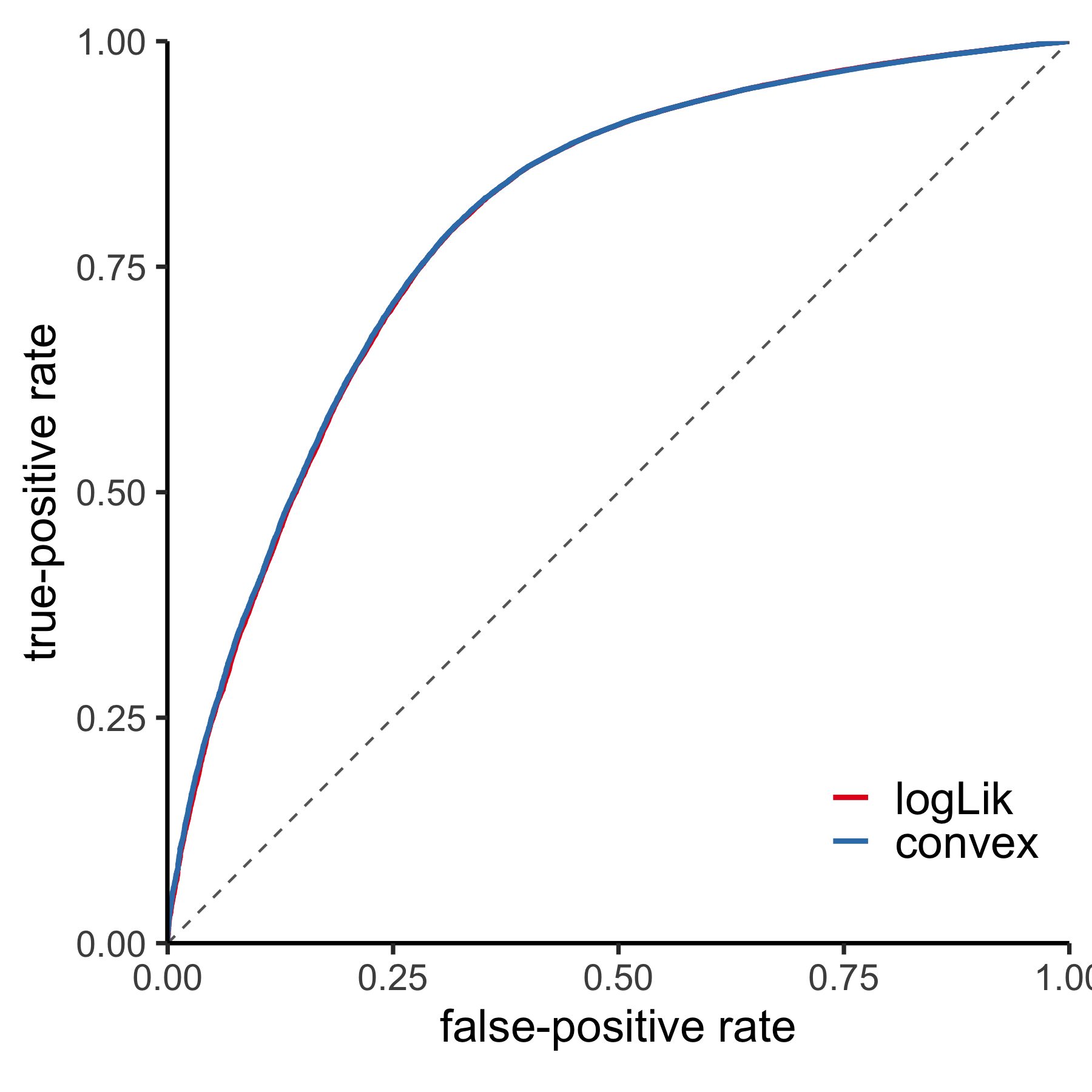}
		\caption{Corrected ROC curves}
	\end{subfigure}
	\caption{Plot of (a) the estimated coefficients and (b) the corrected ROC curves (trained using all examples in the training data set) from the non-convex and convex approaches. In (a), features are sorted based on the coefficients from the likelihood approach.}\label{fig:Bgl3_coeff}
\end{figure}

\begin{figure}
	\centering
	\includegraphics[width=.5\linewidth]{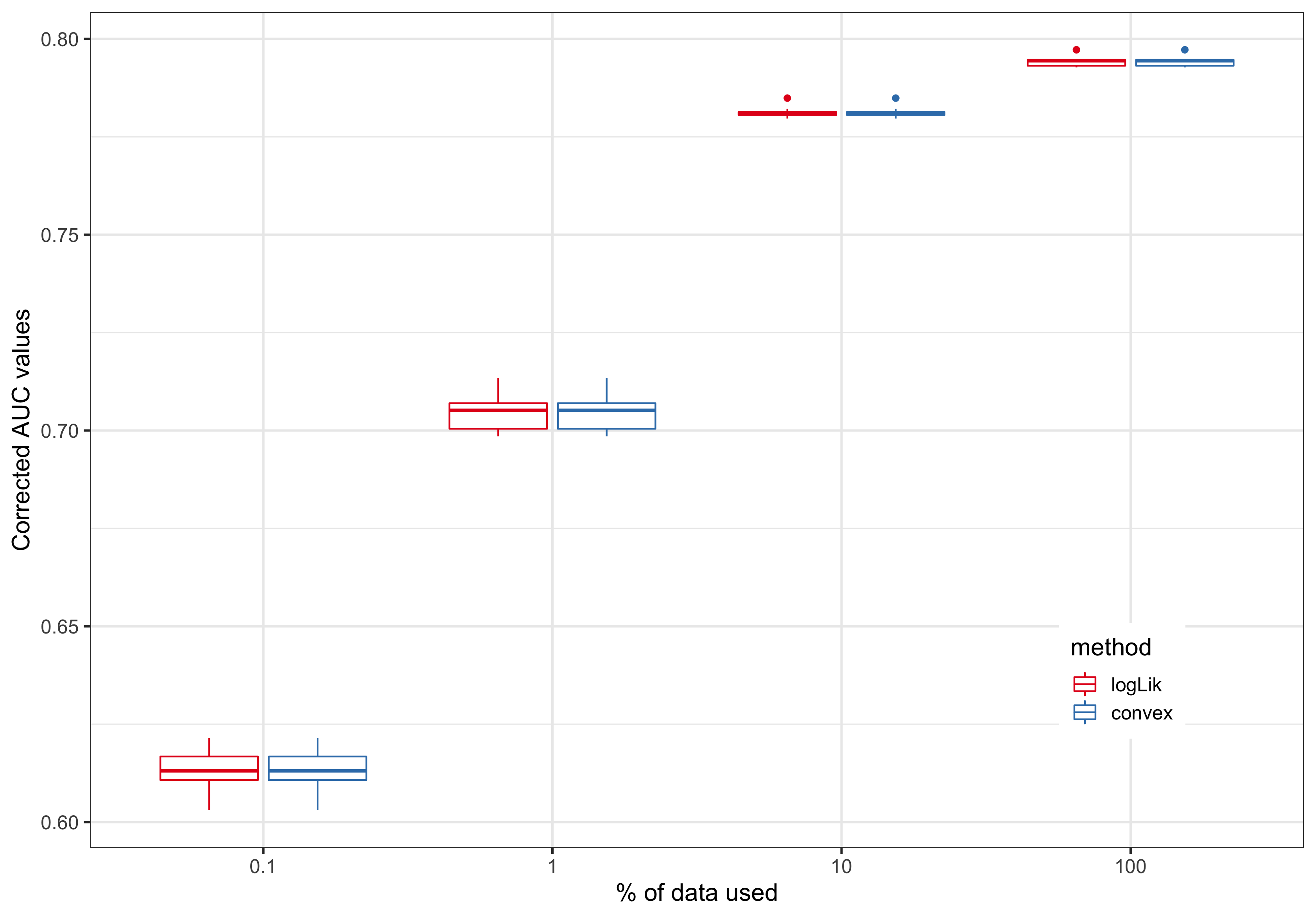}	
	\caption{The corrected AUC values from models using various training sample sizes}\label{fig:Bgl3_at_diff_frac}
\end{figure}
The results are provided in Figure \ref{fig:Bgl3_coeff} and \ref{fig:Bgl3_at_diff_frac}. We observe that the convex approach performed similarly as well as the likelihood approach, where the two approaches produce similar coefficient estimates and comparable classification performance results. Both classifiers demonstrate good classification performance. Therefore, using the convex estimator does not appear to result in a substantial loss of efficiency in this example. 
The corrected AUC value (from the full training set) is $0.7977$ from the likelihood approach and $0.7989$ from the surrogate approach, which is a significant improvement from AUC$=0.5$ in the case of random classification.

\section{Conclusion}\label{sec:concl}
We studied the binary regression problem in the presence of noise in labels in both the classic and high-dimension regimes. We demonstrated that the noisy label model belongs to a sub-class of generalized linear model family. We then discussed two approaches based on a convex surrogate loss for the sub-class of GLMs and a non-convex likelihood for the general class of GLMs. In the low-dimensional setting, the asymptotic distributions of the non-convex likelihood-based estimator and the convex surrogate estimator are derived. We also quantified the efficiency gap between the two approaches and argued that although the convex estimator is provably sub-optimal in terms of efficiency, the gap can be small in some applications.  In the high-dimensional setting, we showed that both estimators, based on regularized non-convex and convex loss functions, achieve a minimax optimal $s \log p /n$ rate for the mean squared errors and derived the asymptotic distribution of the de-biased estimators which can be used for the hypothesis testing in a high-dimensional setting. We empirically demonstrated that both methods perform well in the simulation study and the real data analysis. In particular, although the estimator from the convex approach is sub-optimal in the low-dimensional regime, the efficiency gap between the two estimators is often small. Our empirical results suggest that in sparse regimes the convex surrogate estimator performs better than the likelihood-based estimator. It remains an open question to provide a theoretical justification for this claim.

% Acknowledgements should go at the end, before appendices and references
\acks{
H.S.~was partially supported by the National Institute of Health via grant R01 GM131381-01. G.R.~was partially supported by the National Science Foundation via grant DMS-1811767 and by the National Institute of Health via grant R01 GM131381-01. R.F.B.~was partially supported by the National Science Foundation via grant DMS-1654076 and by an Alfred P.~Sloan fellowship. We thank the three anonymous reviewers whose comments helped improve and clarify this manuscript.}


\appendix
\section{Proofs for Results in Section \ref{sec:ld}}
\subsection{Proof of Proposition \ref{prop:4.1}}
\label{supp_sec:prop:4.1}
First, since $\widehat{\beta_{\ell}}$ and $\widehat{\beta_s}$ are the minimizers of $\L_n^\ell(\beta)$ (over a compact region) and $\L_n^s(\beta)$, and both  converge to $\E_{\beta_0}[\L_n^\ell(\beta)] $ and $\E_{\beta_0}[\L_n^s(\beta)]$ which have a unique maximizer at $\beta_0$,  both $\widehat{\beta_{\ell}}$ and $\widehat{\beta_s}$ are consistent for $\beta_0$ under the Assumption {\bf A\ref{a1}} (e.g., Theorem 2.1 and 2.7 in \citealp{Newey1994-uh}). Also, we note that $\widehat{\beta_{\ell}}$ and $\widehat{\beta_s}$ are zeros of 
\begin{align}
\triangledown \L_{n}^\ell(\beta)&= \frac{1}{n} \sum_{i=1}^n (\mu(h_{LN}(\x_i^\top \beta))-\z_i)h_{LN}'(\x_i^\top \beta)\x_i = \frac{1}{n} \sum_{i=1}^n\psi^{\ell}(\beta, (\x_i,\z_i))\x_i\label{eq:score-1}\\
\triangledown \L_{n}^{s}(\beta) &= \frac{1}{n} \sum_{i=1}^n (\mu(\x_i^\top \beta)-T(\z_i))\x_i =\frac{1}{n} \sum_{i=1}^n \psi^{s}(\beta,(\x_i,\z_i))\x_i,\label{eq:score-2}
\end{align}
where we define
\begin{align*}
	\psi^{\ell}(\beta,(\x,\z)) &:=(\mu(h_{LN}(\x^\top \beta))-\z)h_{LN}'(\x^\top \beta) \\
	\psi^{s}(\beta,(\x,\z)) &:=\mu (\x^\top \beta)-T(\z).
\end{align*}
For notational simplicity, in what follows we write $\psi^{(i)}(\cdot) := \psi(\cdot,(\x_i,\z_i))\x_i$ for any $\psi \in \{\psi^{\ell}, \psi^{s}\}$. Also we define $\psi_n:=n^{-1}\sum_{i=1}^n \psi^{(i)}$.

Then by the second order Taylor expansion, we can establish the asymptotic normality of the two estimators (e.g., Theorem 5.14 in \citealp{Shao2003-vx}, Chapter 2.3.1 in \citealp{Fahrmeir2001-gv}).
We first define the inverse of the asymptotic variance using an estimating equation $\psi$ as 
\begin{align}\label{def:Imatrix}
	&\I_n(\beta;\psi) \\
  &:= \left( \frac{1}{n}\sum_{i=1}^n \E_\beta [\frac{d}{d\beta }\psi^{(i)}(\beta)|\x_i]\right) \left( \frac{1}{n}\sum_{i=1}^n \E_\beta [\psi^{(i)}(\beta)\psi^{(i)}(\beta)^\top |\x_i]\right) ^{-1} \left( \frac{1}{n}\sum_{i=1}^n\E_\beta[\frac{d}{d\beta }\psi^{(i)}(\beta)|\x_i] \right)\nonumber.
\end{align}
Let $g_{CL}(\cdot) := \mu^{-1}(\cdot)$, which is a canonical link function. Recalling the fact $h_{LN}(t) = g_{CL} \circ g_{LN}^{-1}(t)$ and $ g_{LN}^{-1}(t) =  \mu_z (t) = (1-\rho_1-\rho_0)\mu(t) + \rho_0$, we have,
\begin{align}
		\psi^{\ell,(i)}(\beta) &= (1-\rho_1-\rho_0)\left(\mu(h_{LN}(\x_i^\top \beta))-\z_i\right) g_{CL}'(\mu_z (\x_i^\top \beta))\mu'(\x_i^\top \beta)\x_i\label{eq:psi_ell}\\
		\psi^{s,(i)}(\beta) &= \left(\mu(\x_i^\top \beta)-T(\z_i)\right)\x_i.\label{eq:psi_s}
\end{align}
Also, $\mu'(t) =  \mathcal{V}(\mu(t))$, since $\mu'(t) = A''(g_{CL}\circ \mu (t)) =  \mathcal{V}(\mu(t))$ by the definition of $\mathcal{V}$. Also $g_{CL}'(t) = 1/\mathcal{V}(t)$ since $g_{CL}'(\mu(u)) = 1/\mu'(u) = 1/\mathcal{V}(\mu(u))$ by the chain rule. Plugging these expressions in \eqref{eq:psi_ell}, we obtain
\begin{align*}
		\psi^{\ell,(i)}(\beta)
		&=  (1-\rho_1-\rho_0)\left(\mu(h_{LN}(\x_i^\top \beta))-\z_i\right) \frac{\mathcal{V}(\mu(\x_i^\top \beta))}{\mathcal{V}(\mu_z(\x_i^\top \beta))}\x_i.
	\end{align*}	
Then since $\E_\beta[\z_i|\x_i] = \mu(h_{LN}(\x_i^\top \beta))$, we get
\begin{align*}
	\E_\beta[\frac{d}{d\beta} \psi^{\ell,(i)}(\beta)|\x_i] 
	&= A''(h_{LN}(\x_i^\top \beta))\left((1-\rho_1-\rho_0) \frac{\mathcal{V}(\mu(\x_i^\top \beta))}{\mathcal{V}(\mu_z(\x_i^\top \beta))}\right)^2\x_i\x_i^\top \\
	&= \mathcal{V}(\mu_z(\x_i^\top \beta))\left((1-\rho_1-\rho_0) \frac{\mathcal{V}(\mu(\x_i^\top \beta))}{\mathcal{V}(\mu_z(\x_i^\top \beta))}\right)^2\x_i\x_i^\top \\
	&=  (1-\rho_1-\rho_0)^2\frac{\mathcal{V}(\mu(\x_i^\top \beta))^2}{\mathcal{V}(\mu_z(\x_i^\top \beta))}\x_i\x_i^\top.
\end{align*}
Also, since $\psi^{\ell,(i)}(\beta)$ is a negative score function, the variance of the score function is the same as the expected negative derivative of the score function, i.e.,
	\begin{align*}
		\E_\beta[\psi^{\ell,(i)}(\beta)\psi^{\ell,(i)}(\beta)^\top |\x_i] =\E_\beta[\frac{d}{d\beta} \psi^{\ell,(i)}(\beta)|\x_i].
	\end{align*}
Then,
	\begin{align*}
	\I_n(\beta;\psi^\ell)=  (1-\rho_1-\rho_0)^2\frac{1}{n}\sum_{i=1}^n \frac{\mathcal{V}(\mu(\x_i^\top \beta))^2}{\mathcal{V}(\mu_z(\x_i^\top \beta))}\x_i\x_i^\top .
	\end{align*}

For the surrogate function, direct calculations give
\begin{align}\label{eq:surrogate-1}
	\E_\beta[\frac{d}{d\beta} \psi^{s,(i)}(\beta)|\x_i] &= A''(\x_i^\top \beta)\x_i\x_i^\top  = \mathcal{V}(\mu (\x_i^\top \beta))\x_i\x_i^\top \\
	\E_\beta[\psi^{s,(i)}(\beta)\psi^{s,(i)}(\beta)^\top |\x_i] &=\E_\beta[ \left(\mu(\x_i^\top \beta)-T(\z_i)\right)^2 \x_i\x_i^\top |\x_i]=\textnormal{Var}_\beta[T(\z_i)|\x_i]\x_i\x_i^\top  \nonumber
\end{align}
since $\E_\beta[T(\z_i)|\x_i] = \mu(\x_i^\top \beta )$.
Recalling the definition of $T(t) = (t-\rho_0)/(1-\rho_1-\rho_0)$ and $\mathcal{V}$,
\begin{align}
	\textnormal{Var}_\beta[T(\z_i)|\x_i]  = (1-\rho_1-\rho_0)^{-2} \mathcal{V}(\mu_z(\x_i^\top \beta)).\label{eq:surrogate-2}
\end{align}
Thus we have 
\begin{align*}
	&\I_n^{s}(\beta;\psi^{s})\nonumber\\
  &=\left( \frac{1}{n}\sum_{i=1}^n \mathcal{V}(\mu(\x_i^\top \beta))\x_i\x_i^\top \right) \left(\frac{1}{n}\sum_{i=1}^n\frac{\mathcal{V}(\mu_z(\x_i^\top \beta))}{(1-\rho_1-\rho_0)^{2}}  \x_i\x_i^\top \right)^{-1}\left(\frac{1}{n}\sum_{i=1}^n\mathcal{V}(\mu(\x_i^\top \beta))\x_i\x_i^\top \right)
\end{align*}
by applying \eqref{eq:surrogate-1} and \eqref{eq:surrogate-2} to \eqref{def:Imatrix}.

\subsection{Proof of Corollary \ref{cor:1}}\label{supp_sec:cor1}

	First we show $\I_n^\ell(\beta_0) \succeq \I_n^{s}(\beta_0)$ . By the definition of $W_y$ and $W_z$, we have,
	\begin{align*}
		\I_n^\ell(\beta_0) &= (1-\rho_1-\rho_0)^2\X^\top  W_y W_z^{-1}W_y\X/n\\
		\I_n^{s}(\beta_0) &= (1-\rho_1-\rho_0)^2\X^\top W_y\X(\X^\top  W_z \X)^{-1}\X^\top  W_y\X/n.
	\end{align*}
Since the projection matrix $\mathcal{P}_{W_z^{1/2}\X}$ can be written as $\mathcal{P}_{W_z^{1/2}\X} =  W_z^{1/2}\X(\X^\top  W_z\X)^{-1}\X^\top  W_z^{1/2}$, we have 
 $\I_n^{s}(\beta_0) = (1-\rho_1-\rho_0)^2
 \X^\top  W_y W_z^{-1/2}\mathcal{P}_{W_z^{1/2}\X}W_z^{-1/2}W_y\X/n$. Then for any $v \in \R^n$,
 \begin{align}\label{eq:pf-cor1-pd}
	v^\top (\I_n^\ell(\beta_0)-\I_n^{s}(\beta_0) )v&= (1-\rho_1-\rho_0)^2 v^\top  \X^\top  W_y W_z^{-1/2}(I_n-\mathcal{P}_{W_z^{1/2}\X})W_z^{-1/2}W_y\X v/n \nonumber\\ 
	&= (1-\rho_1-\rho_0)^2\|(I_n-\mathcal{P}_{W_z^{1/2}\X})W_z^{-1/2}W_y\X v\|_2^2/n \geq 0	 
 \end{align}
 since $I_n-\mathcal{P}_{W_z^{1/2}\X}  $ is idempotent. Thus, $\I_n^\ell(\beta_0) \succeq \I_n^{s}(\beta_0)$. 
 
Now we address the inequality \eqref{eq:cor1}. First, we have
\begin{align*}
    \|{\I_n^\ell(\beta_0)}^{-1/2}(\I_n^\ell(\beta_0)-\I_n^{s}(\beta_0)){\I_n^\ell(\beta_0)}^{-1/2}\|_2 \leq \|{\I_n^\ell(\beta_0)}^{-1/2}\|_2^2 \|\I_n^\ell(\beta_0)-\I_n^{s}(\beta_0) \|_2,
\end{align*}
and $\|{\I_n^\ell(\beta_0)}^{-1/2}\|_2^2 = \|{\I_n^\ell(\beta_0)}^{-1}\|_2 = (1-\rho_1 - \rho_0)^{-2} \sigma_{\textnormal{min}} ( \X^\top  W_y W_z^{-1}W_y\X/n)^{-1} $.
Also, from \eqref{eq:pf-cor1-pd}, 
\begin{align*}
	\|\I_n^\ell(\beta_0)-\I_n^{s}(\beta_0) \|_2 = (1-\rho_1-\rho_0)^2\|(I_n-\mathcal{P}_{W_z^{1/2}\X})W_z^{-1/2}W_y\X \|_2^2/n.
\end{align*}
Let $A:= W_z^{-1/2}W_y\X $. Then,
\begin{align*}
	\|(I_n-\mathcal{P}_{W_z^{1/2}\X})A\|_2^2 = \sup_{u \in \R^n} \frac{\|(I_n-\mathcal{P}_{W_z^{1/2}\X})Au \|_2^2}{\|u\|_2^2} = \sup_{u \in \mathcal{C}(A)} \frac{\|(I_n-\mathcal{P}_{W_z^{1/2}\X})u \|_2^2}{\|A^\dagger u\|_2^2}
\end{align*}
where $A^\dagger$ is a Moore-Penrose inverse of $A$. Since $\|A^\dagger u\|_2^2 \geq \|u\|_2^2 / \sigma^2_{\textnormal{max}}(A)$ for $u \in \mathcal{C}(A)$,
\begin{align*}
	\|(I_n-\mathcal{P}_{W_z^{1/2}\X})A\|_2^2 &\leq \sigma^2_{\textnormal{max}}(A) \sup_{u \in \mathcal{C}(A)} \frac{\|(I_n-\mathcal{P}_{W_z^{1/2}\X})u \|_2^2}{ \| u\|_2^2}\\
	&= \sigma^2_{\textnormal{max}}(A) \sup_{u \in \mathcal{C}(A), \|u\|_2=1} \|(I_n-\mathcal{P}_{W_z^{1/2}\X})u\|_2^2.
\end{align*}
Since $I_n-\mathcal{P}_{W_z^{1/2}\X}$ is a projection operator onto the orthogonal space of $\mathcal{C}(W_z^{1/2}\X)$, $\|(I_n-\mathcal{P}_{W_z^{1/2}\X})u\|_2^2 = \inf_{v \in \mathcal{C}(W_z^{1/2}\X)}\|u-v\|_2^2$. Therefore,
\begin{align}
	\|(I_n-\mathcal{P}_{W_z^{1/2}\X})A\|_2^2 &\leq \sigma^2_{\textnormal{max}}(A) \sup_{u \in \mathcal{C}(A), \|u\|_2=1}\inf_{v \in \mathcal{C}(W_z^{1/2}\X)}\|u-v\|_2^2 \nonumber \\
	&= \sigma^2_{\textnormal{max}}(W_z^{-1/2}W_y\X)  \delta^2 (\mathcal{C}(W_z^{-1/2}W_y\X),\mathcal{C}(W_z^{1/2}\X) ) \label{eq: pf_cor1-1},
\end{align}
by the definition of the gap \eqref{def:gap}. 

To proceed, we prove a lemma about the gap of two subspaces after linear transformation in relation to the original subspaces. Let $A(\mathcal{M}) := \{Av; v\in \mathcal{M}\}$, and note $\mathcal{C}(A\X) = A(\mathcal{C}(\X))$.

\begin{lemma}\label{lem:A.1}
	Let $\mathcal{M},\mathcal{N} \subseteq \R^n$ be linear subspaces. Let $A \in \R^{n\times n}$ be an invertible matrix and $A(\mathcal{M}) := \{Av; v\in \mathcal{M}\}$. Then
	\begin{align*}
	  \delta(A(\mathcal{M}), A(\mathcal{N})) \leq \kappa(A) \delta(\mathcal{M}, \mathcal{N}),  
	\end{align*}
where $\kappa(A):=\frac{\sigma_{\textnormal{max}}(A)}{\sigma_{\textnormal{min}}(A)}$ is a condition number of $A$.
\end{lemma}
\begin{proof}
	By definition,
	\begin{align*}
		\delta(A(\mathcal{M}), A(\mathcal{N})) &= \sup_{u\in A(\mathcal{M}), \|u\|_2= 1} \inf_{v \in A(\mathcal{N})} \|u-v\|_2\\
		&=  \sup_{u\in \mathcal{M}, \|Au\|_2= 1} \inf_{v \in \mathcal{N}} \|Au-Av\|_2.
	\end{align*}
	For any $v \in \R^n$, we have $ \sigma_{\textnormal{min}}(A)\|v\|_2 \leq \|Av\|_2\leq \sigma_{\textnormal{max}}(A)\|v\|_2$ with $\sigma_{\textnormal{min}}(A)>0$. Thus,
	\begin{align*}
		\sup_{u\in \mathcal{M}, \|Au\|_2= 1} \inf_{v \in \mathcal{N}} \|Au-Av\|_2 &\leq \sup_{u\in \mathcal{M}, \|u\|_2 \leq  1/\sigma_{\textnormal{min}}(A)} \inf_{v \in \mathcal{N}} \sigma_{\textnormal{max}}(A)\|u-v\|_2 \\
		&\leq \frac{\sigma_{\textnormal{max}}(A)}{\sigma_{\textnormal{min}}(A)} \sup_{u\in \mathcal{M}, \|u\|_2 \leq  1} \inf_{v \in \mathcal{N}} \|u-v\|_2 
	\end{align*}
	where we use the fact for any $a\in \R$, $\inf_{v \in \mathcal{N}} \|u-v\|_2 = \inf_{v \in \mathcal{N}} \|u-av\|_2$ since if $v \in \mathcal{N}$ then $av \in \mathcal{N}$ by $\mathcal{N}$ being a linear subspace. We conclude by noting that $ \sup_{u\in \mathcal{M}, \|u\|_2 \leq  1} \inf_{v \in \mathcal{N}} \|u-v\|_2  =  \sup_{u\in \mathcal{M}, \|u\|_2 = 1} \inf_{v \in \mathcal{N}} \|u-v\|_2 $.
\end{proof}

By applying Lemma \ref{lem:A.1} to \eqref{eq: pf_cor1-1}, we have
\begin{align*}
	\delta (\mathcal{C}(W_z^{-1/2}W_y\X),\mathcal{C}(W_z^{1/2}\X) )	 
	& \leq \kappa(W_z^{1/2}) \delta (\mathcal{C}(W_z^{-1}W_y\X),\mathcal{C}(\X) ).
\end{align*}

Therefore,
\begin{align*}
     &\|{\I_n^\ell(\beta_0)}^{-1/2}(\I_n^\ell(\beta_0)-\I_n^{s}(\beta_0)){\I_n^\ell(\beta_0)}^{-1/2}\|_2\\
     &\leq \sigma_{\textnormal{min}}(\X^\top  W_y W_z^{-1}W_y\X/n)^{-1} \|(I_n-\mathcal{P}_{W_z^{1/2}\X})W_z^{-1/2}W_y\X \|_2^2/n\\
     &\leq \sigma_{\textnormal{min}}(\X^\top  W_y W_z^{-1}W_y\X/n)^{-1} \sigma_{\textnormal{max}}^2(W_z^{-1/2}W_y\X/\sqrt{n})\kappa(W_z)\delta^2 (\mathcal{C}(W_z^{-1}W_y\X),\mathcal{C}(\X) ).
\end{align*}
Since 
$\sigma_{\textnormal{min}}(\X^\top  W_y W_z^{-1}W_y\X/n)^{-1} \sigma_{\textnormal{max}}^2(W_z^{-1/2}W_y\X/\sqrt{n}) 
\leq \kappa (\X^\top \X/n) \kappa(W_y^2)\kappa(W_z) $,
\begin{align*}
\|{\I_n^\ell(\beta_0)}^{-1/2}(\I_n^\ell(\beta_0)-\I_n^{s}(\beta_0)){\I_n^\ell(\beta_0)}^{-1/2}\|_2
	\leq \kappa(\X^\top \X/n) \kappa(W_y^2) \kappa(W_z^2)\delta^2 (\mathcal{C}(W_z^{-1}W_y\X),\mathcal{C}(\X) ).
\end{align*}
Note by Assumption {\bf A\ref{a1}},  $\sup_i |\x_i^\top \beta|$ can be bounded by the term independent of $n$ since $\sup_i |\x_i^\top \beta| \leq \sup_i \|\x_i\|_2\|\beta\|_2 \leq r C_X\sqrt{p}$. It follows that $\kappa(W_y^2),\kappa(W_z^2) = O(1)$. 
Also $\lambda_{\textnormal{max}} (\X^\top \X/n)$ is bounded by the term independent of $n$ since $p$ is fixed in the regime of interest and $\lambda_{\textnormal{min}}(\X^\top \X/n) \geq C_\lambda$ by Assumption {\bf A\ref{a1}}.

\subsection{Proof of Proposition \ref{prop:4.3}}\label{sec:pf_prop:4.3}

To show that there exists a unique stationary point in the interior of $\B_2(r)$ w.h.p, we show that there exists an $\ell_2$ ball of radius $\epsilon_0$ centered at $\beta_0$ in which $\L_n^\ell(\beta)$ is strongly convex w.h.p and has at least one local minimum, and no stationary point exists in $\B_2(r) \setminus \B_2(\epsilon_0;\beta_0)$.

We use the following three lemmas to establish the result, whose proofs are provided at the end of this sub-section. The first lemma is about the gradient and Hessian of the population risk. The second and third lemma establish the uniform convergence of the empirical loss, gradient and Hessian to their population counterparts, respectively. We let $\L^\ell(\beta):= \E[\L_n^\ell(\beta)]$. 
\begin{lemma}\label{lem:A.2} There exist an $\epsilon_0>0$ and a constant $\gamma_{\ell}>0$ such that 
\begin{align*}
\inf_{\beta \in \B_2(r) \setminus \B_2(\epsilon_0;\beta_0) } \|\triangledown \L^\ell(\beta)\|_2 \wedge \inf_{\beta \in \B_2(\epsilon_0;\beta_0)} \lambda_{\textnormal{min}}(\triangledown^2 \L^\ell(\beta))\geq \gamma_\ell.
\end{align*}
\end{lemma}

\begin{lemma}\label{lem:A.3} For any given $\delta>0$, we have,
\begin{align*}
		&\P \left( \sup_{\beta \in \B_2(r)} \left\lvert \L_n^\ell(\beta) -  \L^\ell(\beta)  \right\rvert  \leq \tau \sqrt{\frac{Cp \log p}{n}} \right) 
		\geq 1-\delta,
	\end{align*}
where $C = \log(1/\delta) $ and $\tau$ is a constant depending on the model parameters $(K_X,K_Z,r,C_X,C_\ell)$.
\end{lemma}

\begin{lemma}[Theorem 1 in \citealp{Mei2018-ec}]\label{lem:A.4} For $n \geq Cp\log p$ where $C = c_0 \cdot (\log(r \tau /\delta ) \vee 1)$ for an absolute constant $c_0$ and $\tau = 
K_X \max\{C_\ell, L_\ell^{1/3} \}$,
	\begin{align*}
		&\P \left( \sup_{\beta \in \B_2(r)} \| \triangledown \L_n^\ell(\beta) - \triangledown \L^\ell(\beta) \|_2  \leq \tau \sqrt{\frac{C p \log n}{n}} \right) 
		\geq 1-\delta\\
		&\P\left( \sup_{\beta \in \B_2(r)} \| \triangledown^2 \L_n^\ell(\beta) - \triangledown^2 \L^\ell(\beta)\|_2 \leq \tau^2\sqrt{\frac{C p \log n}{n}} \right) \geq 1-\delta.
	\end{align*}
\end{lemma}

First we establish the result of Proposition \ref{prop:4.3} given Lemma \ref{lem:A.2}, \ref{lem:A.3}, and \ref{lem:A.4}.
By Lemma \ref{lem:A.3} and \ref{lem:A.4}, the following inequalities
\begin{align}\label{eq:prop4.3-good_set}
	 &\sup_{\beta \in \B_2(r)} \| \triangledown \L_n^\ell(\beta) - \triangledown \L^\ell(\beta) \|_2   \vee \sup_{\beta \in \B_2(r)}\| \triangledown^2 \L_n^\ell(\beta) - \triangledown^2 \L^\ell(\beta)\|_2 \leq (\gamma_\ell /2) \wedge (\epsilon_g /4r) \nonumber\\
	 & \sup_{\beta \in \B_2(r)} \left\lvert \L_n^\ell(\beta) -  \L^\ell(\beta)  \right\rvert  \leq \epsilon_L/4
\end{align} 
hold with at least probability $1-3\delta$ given a sufficiently large sample size, for $\gamma_\ell$ defined in Lemma \ref{lem:A.2}, $\epsilon_L$ and $\epsilon_g$ defined in \eqref{def:epsilon_L} and \eqref{def:epsilon_g}, respectively. We show that on the event \eqref{eq:prop4.3-good_set} there exists a unique global minimum inside $\B_2(r)$.

First, on the event \eqref{eq:prop4.3-good_set}, we see that $\L_n^\ell(\beta)$ is strongly convex over $\B_2(\epsilon_0;\beta_0)$, since
\begin{align*}
	&\inf_{\beta \in \B_2(\epsilon_0;\beta_0)} \lambda_{\textnormal{min}}(\triangledown^2 \L_n^\ell(\beta))\\
	&\quad \geq \inf_{\beta \in \B_2(\epsilon_0;\beta_0)} \lambda_{\textnormal{min}}(\triangledown^2 \L^\ell(\beta)) - \sup_{\beta \in \B_2(\epsilon_0;\beta_0)}  \| \triangledown^2 \L_n^\ell(\beta)-\triangledown^2 \L^\ell(\beta)\|_2 \\
  &\quad \geq \gamma_\ell/2.
\end{align*}
Then we argue that there exists a local minimum inside the ball $\B_2(\epsilon_0;\beta_0)$. It is sufficient to show that there exists $\beta \in \B_2(\epsilon_0;\beta_0) \setminus \partial \B_2(\epsilon_0;\beta_0) $ such that $\L_n^\ell(\beta) <  \inf_{\beta \in \partial \B_2(\epsilon_0;\beta_0)} \L_n^\ell(\beta)$. Take $\beta =\beta_0$. Note there exists $\epsilon_L >0$ such that
\begin{align}\label{def:epsilon_L}
    \inf_{\beta \in \partial \B_2(\epsilon_0;\beta_0)} \L^\ell(\beta) - \L^{\ell}(\beta_0) = \epsilon_L,
\end{align}
since $ \partial \B_2(\epsilon_0;\beta_0)$ is compact and $A$ is a strictly convex function. Then on the event \eqref{eq:prop4.3-good_set},
\begin{align*}
     \inf_{\beta \in \partial \B_2(\epsilon_0;\beta_0)} \L_n^\ell(\beta) &\geq  \inf_{\beta \in \partial \B_2(\epsilon_0;\beta_0)} \L^\ell(\beta) - \epsilon_L/4 \quad \mbox{and} \quad
     \L_n^\ell(\beta_0) \leq \L^\ell(\beta_0)+ \epsilon_L/4.
\end{align*}
Therefore,
\begin{align*}
    \inf_{\beta \in \partial \B_2(\epsilon_0;\beta_0)} \L_n^\ell(\beta) - \L_n^{\ell}(\beta_0) \geq   \inf_{\beta \in \partial \B_2(\epsilon_0;\beta_0)} \L^\ell(\beta) - \L^{\ell}(\beta_0) - \epsilon_L/2 \geq \epsilon_L/2,
\end{align*}
where we use \eqref{def:epsilon_L} for the last inequality. Also, the empirical gradient in $\B_2(r)$ does not vanish outside of $\B_2(\epsilon_0;\beta_0)$, since
\begin{align*}
	&\inf_{\beta \in \B_2(r) \setminus \B_2(\epsilon_0;\beta_0)} \|\triangledown \L_n^\ell(\beta)\|_2 \\
  &\quad \geq \inf_{\beta \in \B_2(r) \setminus \B_2(\epsilon_0;\beta_0)}\|\triangledown \L^\ell(\beta)\|_2 - \sup_{\beta \in \B_2(r)} \|\triangledown \L_n^\ell(\beta) -\triangledown \L^\ell(\beta)\|_2 \\
  &\quad \geq \gamma_\ell/2.
\end{align*}

Finally, there exists no stationary point on the boundary of $\B_2(r)$. Note there is no stationary point of $\L^\ell(\beta)$ on $\partial \B_2(r)$ since $\langle \triangledown \L^\ell(\beta), \beta_0 - \beta \rangle <0 $ for any $\beta \in \partial \B_2(r)$. Since $\partial \B_2(r)$ is compact, we have $\epsilon_g >0$ such that
\begin{align}\label{def:epsilon_g}
    \sup_{\beta \in \partial \B_2(r)} \langle \triangledown \L^\ell(\beta), \beta_0 - \beta \rangle < -\epsilon_g.
\end{align}
Then for any $\beta \in \partial \B_2(r)$,
\begin{align*}
    \langle \triangledown \L^\ell_n(\beta), \beta_0 - \beta \rangle &=  \langle \triangledown \L^\ell(\beta), \beta_0 - \beta \rangle + \langle \triangledown \L^\ell_n(\beta) - \triangledown \L^\ell(\beta), \beta_0 - \beta \rangle\\
    &\leq -\epsilon_g + \|\triangledown \L^\ell_n(\beta) - \triangledown \L^\ell(\beta)\|_2\|\beta-\beta_0\|_2 \leq -\epsilon_g/2.
\end{align*}
Hence, on the event \eqref{eq:prop4.3-good_set}, there exists a unique stationary point in $\B_2(\epsilon_0;\beta_0) \subsetneq \B_2(r)$ which is a global minimum.

Now we turn to the proofs of three lemmas.
%# proof of Lemma A.2
\begin{proof}[Proof of Lemma \ref{lem:A.2}]
First, we lower bound the minimum eigenvalue of the Hessian. 
\begin{align*}
\inf_{u;\|u\|_2=1}u^\top \triangledown^2 \L^{\ell}(\beta)u 
&= \inf_{u;\|u\|_2=1} \left(u^\top \triangledown^2 \L^{\ell}(\beta_0)u+u^\top \left(\triangledown^2 \L^{\ell}(\beta)-\triangledown^2 \L^{\ell}(\beta_0)\right)u\right)\\
&\geq \inf_{u;\|u\|_2=1} u^\top \triangledown^2 \L^{\ell}(\beta_0)u -\sup_{u;\|u\|_2=1} \left |u^\top \left(\triangledown^2 \L^{\ell}(\beta)-\triangledown^2 \L^{\ell}(\beta_0)\right)u\right|
\end{align*}
At $\beta=\beta_0$,
\begin{align*}
\inf_{u;\|u\|_2=1}u^\top \triangledown^2 \L^{\ell}(\beta_0)u = \E[\ell''(\x^\top \beta_0,\z)(\x^\top u)^2]=\E[\rho_I(\x^\top \beta_0)(\x^\top u)^2]
\end{align*}
since $\E[\rho_R(\x^\top \beta_0,\z)|\x]=0$. Recalling the fact that $\rho_I(t) = A''(h(t))h'(t)^2 \geq 0$ for all $t$, we have a lower bound 
\begin{align*}
	\E[\rho_I(\x^\top \beta_0)(\x^\top u)^2] \geq \E[\rho_I(\x^\top \beta_0)(\x^\top u)^2\1{|\x^\top \beta_0|\leq \tau_c}] \geq \inf_{|t|\leq \tau_c}\rho_I(t) \E[(\x^\top u)^2\1{|\x^\top \beta_0|\leq \tau_c}]
\end{align*}
for any $\tau_c>0$. We let $\tau_c := \left(  r^2 K_X^2 \log \frac{16^2 K_X^4}{C_\lambda^2} \right)^{1/2}$. Then by Cauchy-Schwarz, Assumption {\bf A\ref{a1'}'} and Lemma \ref{lem:A.5},
\begin{align}\label{eq:lemA.2-1}
	\E[(\x^\top u)^2\1{|\x^\top \beta_0|\leq \tau_c}] 
	&= \E[(\x^\top u)^2] - \E[(\x^\top u)^2\1{|\x^\top \beta_0|\geq \tau_c}]\nonumber \\
	&\geq C_\lambda -  \E[(\x^\top u)^4]^{1/2} \P(|\x^\top \beta_0|\geq \tau_c)^{1/2} \geq C_\lambda/2.
\end{align}

Now we bound the difference term. Using Lipschitz assumption in {\bf A\ref{a2}}, we have
\begin{align*}
\left|u^\top \left(\triangledown^2 \L^{\ell}(\beta)-\triangledown^2 \L^{\ell}(\beta_0)\right)u\right| 
&\leq \E[\left|\ell''(\x^\top \beta,\z)-\ell''(\x^\top \beta_0,\z) \right|(\x^\top u)^2]\\
&\leq L_\ell\E[|\x^\top (\beta-\beta_0)|(\x^\top u)^2].
\end{align*}
Then by Cauchy-Schwarz and sub-Gaussian moment property,
\begin{align}\label{eq:lemA.2-2}
 L_\ell\E[|\x^\top (\beta-\beta_0)|(\x^\top u)^2]
&\leq L_\ell\|\Delta_0\|_2 \E[(\x^\top \Delta_0/\|\Delta_0\|_2)^2]^{1/2} \E[(\x^\top u)^4]^{1/2} \leq 4\sqrt{2}K_X^3L_\ell\|\Delta_0\|_2
\end{align}
where $\Delta_0 = \beta-\beta_0$. Hence, combining \eqref{eq:lemA.2-1}, \eqref{eq:lemA.2-2}, we conclude for $\beta$ such that $\|\beta-\beta_0\|_2\leq \epsilon_0$, for $\epsilon_0 := (\inf_{|t|\leq \tau_c}\rho_I(t)C_\lambda)/(16\sqrt{2}K_X^3L_\ell)$,
\begin{align*}
	\inf_{u;\|u\|_2=1}u^\top \triangledown^2 \L^{\ell}(\beta)u  \geq \inf_{|t|\leq \tau_c}\rho_I(t)  C_\lambda/4.
\end{align*}

Now we address the lower bound of the gradient. Let $\beta \in \B_2(r) \setminus \B_2(\epsilon_0;\beta_0)$ be fixed.
\begin{align*}
\langle \beta-\beta_0, \triangledown \L^{\ell}(\beta)\rangle &= \E[\{A'(h(\x^\top \beta))-A'(h(\x^\top \beta_0))\}h'(\x^\top \beta)\x^\top (\beta-\beta_0) ]\\
&=  \E[A''(h(\x^\top \beta_i))h'(\x^\top \beta_i)h'(\x^\top \beta)(\x^\top (\beta-\beta_0))^2]
\end{align*}
for $\beta_i = \beta_0 + v(\beta -\beta_0)$ where $v \in [0,1]$ by the mean value theorem. 

Define an event $\mathcal{E} := \{ |\x^\top \beta_0| \leq \tau_c , |\x^\top \Delta_0| \leq 2\tau_c\}$ where $\Delta_0 := \beta- \beta_0$.
\begin{align}\label{eq:lemA.2-3}
\langle \beta-\beta_0, \triangledown \L^{\ell}(\beta)\rangle 
&\geq  C_r \E[(\x^\top (\beta-\beta_0))^2\mathbbm{1}_\mathcal{E}] \geq C_r C_\lambda \|\Delta_0\|_2^2/2
\end{align}
for $C_r:=\left(\inf_{t;|t|\leq 3\tau_c}A''(h(t))h'(t)\right)(\inf_{t;|t|\leq 3\tau_c}h'(t))>0$, since
\begin{align*}
	\E[(\x^\top (\beta-\beta_0))^2\mathbbm{1}_\mathcal{E}] &= \E[(\x^\top (\beta-\beta_0))^2]-\E[(\x^\top (\beta-\beta_0))^2\mathbbm{1}_{\mathcal{E}^c}]\\
	&\geq \|\Delta_0\|_2^2 \left(C_\lambda - \E[(\x^\top \Delta_0/\|\Delta_0\|_2)^4]^{1/2} \P(\mathcal{E}^c)^{1/2} \right)\\
	&\geq \|\Delta_0\|_2^2C_\lambda/2
\end{align*}
by Lemma \ref{lem:A.5}. We apply Cauchy-Schwarz inequality to \eqref{eq:lemA.2-3} to obtain
\begin{align*}
\|\triangledown \L^{\ell}(\beta)\|_2 \geq  C_\lambda C_r\|\beta-\beta_0\|_2/2 \geq (C_\lambda C_r \epsilon_0)/2
\end{align*}
since $\|\beta-\beta_0 \|_2 \geq \epsilon_0$. Finally, take $\gamma_\ell := C_\lambda \left( \inf_{|t|\leq \tau_c}\rho_I(t)/4 \wedge  (C_r \epsilon_0)/2\right)$ to conclude.
\end{proof}

% Lemma A.5
\begin{lemma}\label{lem:A.5} Let $\x \in \R^p$ be a random vector which satisfies the sub-gaussian tail condition with the parameter $K_X$, and also let $u_2,u_2,u_3 \in \R^p$ be non-random vectors such that $\|u_1\|_2=1$, $\|u_2\|_2 \leq c_1 r$, and $\|u_3\|_2 \leq c_2r$ for some $c_1, c_2 >0$. For $\tau_c := \left(  r^2 K_X^2 \log \frac{16^2 K_X^4}{C_\lambda^2} \right)^{1/2}$, we have
	\begin{align*}
		\E[(\x^\top u_1)^4]^{1/2}  \lbrace \P(|\x^\top u_2| \geq c_1 \tau_c) +\P(|\x^\top u_3| \geq c_2 \tau_c)\rbrace^{1/2} \leq \frac{C_\lambda}{2}.
	\end{align*}
\end{lemma}
\begin{proof}
Moment and tail properties of sub-Gaussian distribution give
$\E[(\x^\top u_1)^4]^{1/2} \leq 4K_X^2$ and  $\P(|\x^\top u_2| \geq c_1 \tau_c) +\P(|\x^\top u_3| \geq c_2 \tau_c) \leq 4\exp(-\tau_c^2/r^2K_X^2)$. Then the choice of $\tau_c$ gives the desirable bound.
\end{proof}
%% Proof of A.3
\begin{proof}[Proof of Lemma \ref{lem:A.3}]
First, we use an extension of McDiarmid's inequality to obtain
\begin{align*}
    \P\left(\sup_{\beta \in \B_2(r)}\left\lvert \L_n^\ell(\beta)- \L^\ell(\beta)\right\rvert \geq  \E\left[ \sup_{\beta \in \B_2(r)}\left\lvert \L_n^\ell(\beta)- \L^\ell(\beta)\right\rvert \right] + t\right) \leq \exp ( -Ct^2n),
\end{align*}
for some constant $C>0$. We will use the following extension of McDiarmid’s inequality, due to \citet{Kontorovich2014-ao}.
\begin{lemma}[Theorem 1 in  \citealp{Kontorovich2014-ao}]\label{lem:extension_Mcdiarmid}
	 Let ($\mathcal{X}_i, \rho_i, \mu_i$) be a sequence of metric spaces, $i =1,\dots,n$. Let $\mathcal{X}^n = \mathcal{X}_1\times \dots\times \mathcal{X}_n$, $\mu^n = \mu_1\times \dots\times \mu_n$, and $\rho^n(x,x') = \sum_{i=1}^n \rho_i(x,x')$ be the product probability space, the product measure, and $\ell_1$ product metric. Let $X_i$ be $\mathcal{X}_i$-valued random variables where $X_i \sim \mu_i$. Suppose $\varphi : \mathcal{X}^n \rightarrow \mathbb{R}$ is 1-Lipschitz with respect to $\rho^n$ metric, i.e., $|\varphi(x)-\varphi(x')|\leq \sum_{i=1}^n \rho_i(x,x')$ for $x,x' \in \mathcal{X}^n$, and there exists a sub-gaussian parameter $\Delta_{SG}(\mathcal{X}_i)<\infty$ such that
\begin{align*}
	\E[\exp (\lambda\sigma_i \rho_i(X_i,X_i')) ] \leq \exp(\lambda^2 \Delta_{SG}(\mathcal{X}_i)^2 /2), \mbox{ for all }\lambda\in \mathbb{R},
\end{align*}
where $X_i,X_i' \sim \mu_i$ are independent and $\sigma_i$ is a Rademacher variable independent of $(X_i,X_i')$. Then $\E[\varphi] < \infty$, and  
\begin{align*}
	\P\left(|\varphi(X_1,\dots,X_n) - \E[\varphi(X_1,\dots,X_n)] |>t\right)  \leq 2\exp \left( - \frac{t^2}{2 \sum_{i=1}^n \Delta_{SG}^2(\mathcal{X}_i)} \right).
\end{align*}
\end{lemma}

Let $\mathbf{u}_i := [\x_{i1},\dots,\x_{ip},\z_i] \in \mathcal{X}_i$, for $\mathcal{X}_i:=\R^p \times \mathcal{Z}$, $\forall i$. We define $\phi:\mathcal{X}^n \rightarrow \R$ as 
\begin{align*}
	\phi (\mathbf{u}_1,\dots,\mathbf{u}_n) := \sup_{\beta \in \B_2(r)} \left\lvert\L_n^\ell(\beta) -\E[\L_n^\ell(\beta)] \right\rvert  = \sup_{\beta \in \B_2(r)} \left\lvert \frac{1}{n} \sum_{i=1}^n \ell(\x_i^\top \beta,\z_i)-\E[\ell(\x_i^\top \beta,\z_i)] \right\rvert.
\end{align*}
We show that $\phi$ is $L_0/n$ Lipschitz, for $L_0$ to be defined later. Let $\mu_i$ be the joint distribution of $(\x_i,\z_i)$, i.e., $\mu_i =\P_\x \times \P_{\z|\x}$, $\forall i$. For $(\mathbf{u}_1,\dots,\mathbf{u}_n),(\mathbf{u}'_1,\dots,\mathbf{u}_n') \sim \mu^n$, 
\begin{align*}
	&|\phi (\mathbf{u}_1,\dots,\mathbf{u}_n)-\phi (\mathbf{u}'_1,\dots,\mathbf{u}_n')| \\
	&\leq \sup_{\beta \in \B_2(r)} \left\lvert \frac{1}{n} \sum_{i=1}^n \ell(\x_i^\top \beta,\z_i)-\ell({\x'_i}^\top \beta,\z'_i)] \right\rvert  \\
	&= \sup_{\beta \in \B_2(r)} \left\lvert \frac{1}{n} \sum_{i=1}^n \triangledown \ell(\tilde{\x}_i^\top \beta,\tilde{\z}_i) ^\top
	\begin{bmatrix}
		(\x_i-{\x_i'})^\top \beta\\ 
		\z_i-\z'_i
	\end{bmatrix}
	\right\rvert\\
	&\leq \sup_{\beta \in \B_2(r)} \left\lvert \frac{1}{n} \sum_{i=1}^n (\|\beta\|_\infty\vee 1)\max \{| \triangledown_1 \ell(\tilde{\x}_i^\top \beta,\tilde{\z}_i)|, |\triangledown_2 \ell(\tilde{\x}_i^\top \beta,\tilde{\z}_i)|\}
	\|\mathbf{u}_i-\mathbf{u}'_i\|_1
	\right\rvert
\end{align*}
where the first equality uses mean value theorem, the second inequality uses H\"{o}lder's inequality, $\tilde{\x}_i^\top \beta \in [\x_i^\top \beta, {\x'_i}^\top \beta]$ and $\tilde{\z}_i \in [\z_i,\z'_i]$, and $\triangledown_i$ refers to a derivative with respect to $i$th argument. We note that $|\triangledown_1 \ell(\tilde{\x}_i^\top \beta,\tilde{\z}_i) |= |\ell'(\tilde{\x}_i^\top \beta,\tilde{\z}_i)| \leq C_\ell$ a.s. by Assumption {\bf A\ref{a2}}. On the other hand, for all $\beta \in \B_2(r)$, $|\triangledown_2 \ell(\tilde{\x}_i^\top \beta,\tilde{\z}_i) | = |A(h(\tilde{\x}_i^\top \beta))- h(\tilde{\x}_i^\top \beta)| \leq C_{\ell,2}$, for $C_{\ell,2}:= \sup_{t; |t| \leq 2rC_X\sqrt{p}} |A(h(t))-h(t)|$ since $|\tilde{\x}_i^\top \beta| \leq |\x_i^\top \beta| +|{\x'_i}^\top \beta| \leq 2rC_X\sqrt{p}$ by Assumption {\bf A\ref{a1'}'}.
Then,
\begin{align*}
	|\phi (\mathbf{u}_1,\dots,\mathbf{u}_n)-\phi (\mathbf{u}'_1,\dots,\mathbf{u}_n')|  &\leq \frac{1}{n} (r\vee 1)\cdot (C_\ell \vee C_{\ell,2})\sum_{i=1}^n\|\mathbf{u}_i-\mathbf{u}'_i\|_1\\
	& = \frac{L_0}{n}\rho^n(\mathbf{u}_i,\mathbf{u}'_i)
\end{align*}
where $L_0 := (r\vee 1) \cdot (C_\ell \vee C_{\ell,2}) $, and $\rho^n$ is an $\ell_1$ product metric for $\rho_i(x,x') = \|x-x'\|_1$. In particular, $\phi(\cdot)$ is $L_0/n$ Lipschitz. Then by Lemma \ref{lem:extension_Mcdiarmid}, 
\begin{align}\label{supp_eq:lemA.3-1}
    \P\left(\sup_{\beta \in \B_2(r)}\left\lvert \L_n^\ell(\beta)- \L^\ell(\beta)\right\rvert \geq  \E\left[ \sup_{\beta \in \B_2(r)}\left\lvert \L_n^\ell(\beta)- \L^\ell(\beta)\right\rvert \right] + t\right) \leq \exp ( -\frac{t^2n^2}{2L_0^2\sum_{i=1}^n \Delta_{SG}^2(\mathcal{X}_i)}),
\end{align}
provided that $\Delta_{SG}^2(\mathcal{X}_i)<\infty$. Now we calculate $\Delta_{SG}^2(\mathcal{X}_i)$. For any $\lambda \in \R$, 
\begin{align*}
	\E[\exp(\lambda  \sigma_i \|\mathbf{u}_i-\mathbf{u}_i'\|_1)] 
	&= \E[\exp(\lambda  \sigma_i\{ \sum_{j=1}^p |\x_{ij}-\x'_{ij}| + |\z_i - \z_i'|\})]\\
	&= \E[\exp(\lambda  \sigma_i \sum_{j=1}^p |\x_{ij}-\x'_{ij}|) \E[\exp(\lambda  \sigma_i|\z_i - \z_i'|)|\x_i]]\\
	& \leq \exp(\lambda^2 (pK_X^2+K_Z^2) ),
\end{align*}
thus $\Delta_{SG}^2(\mathcal{X}_i) = 2(pK_X^2 + K_Z^2)$.
Then for $t \geq \sqrt{2} L_0 (K_X+K_Z) \sqrt{\frac{p\log (1/\delta) }{n}}$,  the probability of the LHS is bounded below by $1-\delta$. 

We bound the expectation term by the standard arguments using symmetrization and contraction inequality. We have,
\begin{align*}
    \E\left[ \sup_{\beta \in \B_2(r)}\left\lvert \L_n^\ell(\beta)- \L^\ell(\beta)\right\rvert \right] &= \E\left[ \sup_{\beta \in \B_2(r)}\left\lvert \frac{1}{n}\sum_{i=1}^n \ell(\x_i^\top \beta,\z_i) - \E [\ell(\x_i^\top \beta,\z_i)]\right\rvert \right]\\
    &\leq 2\E\left[ \sup_{\beta \in \B_2(r)}\left\lvert \frac{1}{n}\sum_{i=1}^n \sigma_i\ell(\x_i^\top \beta,\z_i) \right\rvert\right],
\end{align*}
where we let $(\sigma_i)_{i=1}^n$ be i.i.d. Rademacher variables independent from $(\x_i,\z_i)_{i=1}^n$. Since
$|\ell(t,\z_i) - \ell(s,\z_i)| \leq C_\ell |t-s| $ a.s. by Assumption {\bf A\ref{a2}}, contraction inequality gives
\begin{align*}
    \E\left[ \sup_{\beta \in \B_2(r)}\left\lvert \L_n^\ell(\beta)- \L^\ell(\beta)\right\rvert \right] &
    \leq 4C_\ell \E\left[ \sup_{\beta \in \B_2(r)}\left\lvert \frac{1}{n}\sum_{i=1}^n \sigma_i\x_i^\top \beta \right\rvert\right] +2 \E \left[ \frac{1}{n} \left \lvert \sum_{i=1}^n \sigma_i\ell(0,\z_i) \right\rvert \right] \\
    &\leq 4C_\ell \E\left[ \sup_{\beta \in \B_2(r)} \|\beta \|_1\left\lVert \frac{1}{n}\sum_{i=1}^n \sigma_i\x_i \right \rVert_\infty \right]+2 \E \left[ \frac{1}{n} \left \lvert \sum_{i=1}^n \sigma_i\ell(0,\z_i) \right\rvert \right] \\
    &\leq 4rC_\ell \sqrt{p}\E\left[ \left\lVert \frac{1}{n}\sum_{i=1}^n \sigma_i\x_i \right \rVert_\infty \right]+2 \E \left[ \frac{1}{n} \left \lvert \sum_{i=1}^n \sigma_i\ell(0,\z_i) \right\rvert \right] .
    \end{align*}
Since $|\sigma_i|\leq 1$ a.s. and $\E[\sigma_i \x_{ij}] = 0$, $\sigma_i \x_{ij}$ is mean-zero $cK_X$ sub-gaussian with some absolute constant $c>0$, $\forall i,j$. As $(\sigma_i \x_{ij})_{i=1}^n$ are independent for any $j$, $\sum_{i=1}^n \sigma_i \x_{ij}/n$ is also sub-gaussian with parameter $cK_X/\sqrt{n}$. Similarly, $\sigma_i\ell(0,\z_i)$ is mean-zero $c'K_Z$ sub-gaussian where $c'$ only depends on $(A,g)$, and $\sum_{i=1}^n \sigma_i \ell(0,\z_i)/n$ is sub-gaussian with parameter $c'K_Z/\sqrt{n}$.
Therefore, by the bound on the maximum of sub-gaussian variables,
\begin{align}\label{supp_eq:lemA.3-2}
     \E\left[ \sup_{\beta \in \B_2(r)}\left\lvert \L_n^\ell(\beta)- \L^\ell(\beta)\right\rvert \right] \leq c''r (K_X \vee K_Z )C_\ell  \sqrt{\frac{p\log p}{n}}.
\end{align}
where $c''$ is a constant depending only on the choice of model $(A,g)$. Combining \eqref{supp_eq:lemA.3-1} and \eqref{supp_eq:lemA.3-2}, we obtain the desired inequality.
\end{proof}

\begin{proof}[Proof of Lemma \ref{lem:A.4}]
We verify Assumptions 1-3 in \citet{Mei2018-ec}. The first assumption is to verify whether the gradient of the loss has a sub-Gaussian tail. The second assumption is to show that the Hessian evaluated on a unit vector is sub-Exponential. The third assumption is about the Lipschitz continuity of the Hessian. We mainly check whether quantities in interest satisfy a sub-gaussian/exponential moment bounds.
\begin{itemize}
	\item[A1]	
	For any $u\in \R^p$ such that $\|u\|_2=1$,
	$\langle \ell'(\x^\top \beta,\z) \x, u\rangle$ is sub-gaussian since \\
$\E(|\ell'(\x^\top \beta,\z) \x^\top u|^k)^{1/k} \leq \|\ell'\|_\infty \E[|\x^\top u|^k]^{1/k}\leq  C_\ell K_X \sqrt{k}$ for any $k \geq 1$.
\item[A2] Similarly for any $u\in \R^p$ such that $\|u\|_2=1$ $\langle u , \ell''(\x^\top \beta,\z) \x\x^\top u\rangle$ is sub-exponential since \\
$ \E[|\ell''(\x^\top \beta,\z)(\x^\top u)^2|^{k}]^{1/k} \leq  2C_\ell \E[(\x^\top u )^{2k}]^{1/k} \leq  4 C_\ell K_X^2 k$ for any $k \geq 1$.
\item[A3]  
$\|\triangledown^2 \L^\ell(\beta_0)\|_2 =  \sup_{u;\|u\|_2=1}  \E[\rho_I(\x^\top \beta_0)(\x^\top u)^2] \leq 2C_\ell K_X^2$. Also from the Lipschitz continuity assumption of $\ell''$, 
\begin{align*}
\E \left[ \sup_{\beta_1 \neq \beta_2} \frac{\|\triangledown^2 \L^\ell (\beta_1) - \triangledown^2 \L^\ell(\beta_2)\|_2}{\|\beta_1-\beta_2\|_2}\right]
& = \E \left[\sup_{\substack{\beta_1 \neq \beta_2,\\u;\|u\|_2=1}} \frac{|\ell''(\x^\top \beta_1,\z)- \ell''(\x^\top \beta_2,\z)|(\x^\top  u)^2}{\|\beta_1-\beta_2\|_2} \right]\\
& \leq  L_{\ell}\E\left[ \sup_{\substack{\beta_1 \neq \beta_2,\\u;\|u\|_2=1}} \frac{|\x^\top \beta_1-\x^\top \beta_2|(\x^\top u)^2}{\|\beta_1-\beta_2\|_2}\right].
\end{align*}
By Cauchy-Schwarz, $|\x^\top (\beta_1-\beta_2)| \leq \|\x\|_2 \|\beta_1-\beta_2\|_2$ and $(\x^\top u)^2 \leq \|\x\|_2^2$ since $\|u\|_2=1$. Thus
\begin{align*}
	\E \left[ \sup_{\beta_1 \neq \beta_2} \frac{\|\triangledown^2 \L^\ell (\beta_1) - \triangledown^2 \L^\ell(\beta_2)\|_2}{\|\beta_1-\beta_2\|_2}\right]\leq  L_{\ell}\E\left[\|\x\|_2^3\right] \leq 3^{3/2} L_\ell  K_X^3 p^{3/2},
\end{align*}
\end{itemize}
since $\E[\|\x\|_2^3] = \E[(\sum_{i=1}^p x_i^2)^{3/2}] \leq p^{1/2} \E[(\sum_{i=1}^p |x_i|^3)] \leq 3^{3/2} K_X^3 p^{3/2}$.
\end{proof}

%% Proof of Cor2
\subsection{Proof of Corollary \ref{cor:2}}\label{supp_sec:cor2}
For a {\bf(GLM)} with parameters ($\log(1+\exp(\cdot)), g_{LN}$) with $\z_i \in \{0,1\}$, we first note that the sub-gaussian tail condition for $\z_i$ is satisfied with $K_Z=1$ since $|\z_i - \E[\z_i|\x_i]|\leq 1$, almost surely. Now we show that {\bf A\ref{a2}} is satisfied. Then the result follows from the Proposition \ref{prop:4.3}.

From \eqref{eq:l'} and \eqref{eq:l''}, we have 
\begin{align*}
	\ell'(t,z) &= \left(A'(h_{LN}(t))-z\right) h_{LN}'(t) \\
	\ell''(t,z) &= \rho_I(t) + \rho_R(t,z),
\end{align*}
for $A(t) = \log(1+\exp(t))$ and $\rho_I(t)$ and $\rho_R(t,z)$ such that
\begin{align*}
	\rho_I(t) = A''(h_{LN}(t))h_{LN}'(t)^2, \quad \mbox{and} \quad
	\rho_R(t,z) = (A'(h_{LN}(t))-z) h_{LN}''(t).
\end{align*}
From Lemma \ref{lem:h} which is presented at the end of this subsection, $\|h_{LN}'\|_\infty \leq 1$ and $\|h_{LN}''\|_\infty \leq 2$.  Also $A''(t) = e^t  / (1+e^t )^2$ is bounded by $1/4$ and $|A'(h_{LN}(t))-z|\leq 1 $ for any $t$ and $z \in \{0,1\}$, since  $0 \leq A'(h_{LN}(t)) \leq 1, \forall t$. Thus 
\begin{align*}
	|\ell'(t,z) | \leq 1, \quad |\rho_I(t)| \leq \frac{1}{4}, \quad\mbox{and} \quad |\rho_R(t,z)| \leq \|h_{LN}''\|_\infty \leq 2,\quad \forall z \in \{0,1\}, \forall t, 
\end{align*}
and $\max\{\|\ell'\|_\infty, \|\rho_I\|_\infty, \|\rho_R\|_\infty \}$ is bounded by $2$.

To verify that $\ell''$ is $L_\ell$-Lipschitz where $L_\ell$ does not depend on $t$, it is sufficient to show that the gradients of $\rho_I$ and $\rho_R$ are bounded independent of $t$. By calculation, we have
\begin{align*}
	\rho_I'(t) &= A'''(h_{LN}(t))h_{LN}'(t)^3 + 2A''(h_{LN}(t))h_{LN}'(t)^2h_{LN}''(t)\\
	\rho_R'(t,z) &= A''(h_{LN}(t))h_{LN}'(t)h_{LN}''(t)+\{A'(h_{LN}(t))-z\} h_{LN}'''(t).
\end{align*}
We bound each term separately. As other terms can be bounded similarly other than the term involving $A'''$, it is sufficient to show that $A'''(t)$ is bounded by an absolute constant. We have,
\begin{align*}
|A'''(t)| &= \left\lvert\frac{e^t }{(1+e^t )^2 }- \frac{2e^{2t}}{(1+e^t )^3}\right\rvert \leq  \left\lvert\frac{e^t }{(1+e^t )^2 }\right\rvert+ \left\lvert \left(\frac{2e^{t}}{(1+e^t )^2}\right)\left(\frac{e^{t}}{1+e^t } \right)\right\rvert \leq \frac{1}{4} + \frac{1}{2} \leq 1.
\end{align*}

Finally, we present the Lemma about the boundedness of $h_{LN}',h_{LN}''$ and $h_{LN}'''$. 
% Lemma about the boundedness of h
\begin{lemma}\label{lem:h}	
There exists $C \leq 7$ such that $\max \{ \|h_{LN}'\|_\infty,\|h_{LN}''\|_\infty,\|h_{LN}'''\|_\infty\} \leq C$ for $h_{LN} = (A')^{-1}\circ g_{LN}^{-1}$.
\end{lemma}
\begin{proof}
From the definition of $g_{LN}$ and $h_{LN}$ in \eqref{def:g_nl}, we have	
\begin{align*}
	h_{LN}(t) := \log \left(\frac{(1-\rho_1-\rho_0)\mu(t) +\rho_0}{1-(1-\rho_1-\rho_0)\mu(t) -\rho_0}\right).
\end{align*}
Let 
$ a= 1-\rho_1-\rho_0$ and $b = \rho_0$.
We have $a \mu(t) + b \leq a+b <1$ and $a \mu(t) \leq a\mu(t)+b$ for any $t$. Then $\forall t$,
\begin{align}
&\frac{a\mu(t)}{a\mu(t)+b} \leq 1 \quad \mbox{and} \quad \frac{a(1-\mu(t))}{1-a\mu(t)-b} <1 \label{eq:p_ineq}
\end{align}
By definition of $h_{LN}(t)= \log (a \mu(t) + b) - \log (1-a \mu(t)-b)$,
\begin{align*}
	h_{LN}'(t) &= \frac{d}{d\mu(t)} \log \left(\frac{a \mu(t) + b}{1-(a \mu(t) + b)}\right) \frac{d\mu(t)}{dt}\\ 
	&= \frac{a \mu(t) (1-\mu(t)) }{ (a \mu(t)+b) (1-a\mu(t)-b)} \leq 1
\end{align*}
by the fact that $$ \frac{d\mu(t)}{dt} = A''(t) = \mu(t)(1-\mu(t))$$ and the inequalities \eqref{eq:p_ineq}. In particular, $h_{LN}' \geq 0$ and $\|h_{LN}'\|_\infty \leq 1$.

Now we bound $h_{LN}''$. From elementary calculation, it can be shown that
\begin{align*}
	h_{LN}''(t) &= h_{LN}'(t) (1- 2\mu(t)) -h_{LN}'(t)^2(1-2(a\mu(t)+b)),\\
	h_{LN}'''(t) &=  h_{LN}''(t) \left\lbrace 1- 2\mu(t) -2h_{LN}'(t)(1-2(a\mu(t)+b))\right\rbrace\\
  &\quad -2\mu(t) (1-\mu(t)) h_{LN}'(t)(1-ah_{LN}'(t)).
\end{align*}
In particular,
\begin{align*}
|h_{LN}''(t)| \leq h_{LN}'(t) |1- 2\mu(t)| +h_{LN}'(t)^2|1-2(a\mu(t)+b)| \leq 2 \|h_{LN}'\|_\infty  \leq 2
\end{align*}
since $\max_{0\leq \mu \leq 1} |1-2\mu| = 1$ and $0\leq \mu(t),a \mu(t) +b \leq 1$, for all $t$. Also,
\begin{align*}
	|h_{LN}'''(t)| &\leq |h_{LN}''(t)| \left\lbrace |1- 2\mu(t)| + 2h_{LN}'(t)|1-2(a\mu(t)+b)|\right\rbrace \\
  &\qquad +2\mu(t) (1-\mu(t)) h_{LN}'(t)(1-ah_{LN}'(t))\\
	& \leq  3 \|h_{LN}''\|_\infty + \frac{1}{2}.
\end{align*}
\end{proof}
\section{Proofs for Results in Section \ref{sec:hd}}
%% Proof of Proposition 5.1
\subsection{Proof of Proposition \ref{prop:5.1} }
\label{supp_sec:prop5.1}
First, we note that the inequality \eqref{eq:prop5.1} holds trivially for $\beta=\beta_0$. For any $\beta\in \B_2(r) \setminus \{\beta_0\}$ and $\Delta_0 := \beta- \beta_0$,
\begin{align*}
	&\langle \triangledown \L_n^\ell(\beta) - \triangledown \L_n^\ell(\beta_0),\Delta_0\rangle  \\
  &= \langle  \triangledown \L^\ell(\beta)-\triangledown \L^\ell(\beta_0), \Delta_0\rangle+ \langle \triangledown \L_n^\ell(\beta) - \triangledown \L^\ell(\beta), \Delta_0\rangle + \langle \triangledown \L^\ell(\beta_0) - \triangledown \L_n^\ell(\beta_0), \Delta_0\rangle\\
	&\geq \langle \triangledown \L^\ell(\beta), \Delta_0\rangle - \left( \left\lvert \frac{ \langle \triangledown \L_n^\ell(\beta) - \triangledown \L^\ell(\beta), \Delta_0\rangle}{\|\Delta_0\|_1}\right\rvert+\|  \triangledown \L_n^\ell(\beta_0)\|_\infty \right) \|\Delta_0\|_1 
\end{align*}
using $\triangledown \L^\ell(\beta_0)=0$ and H\"{o}lder's inequality. From \eqref{eq:lemA.2-3}, we have
\begin{align*}\langle  \triangledown \L^{\ell}(\beta),\Delta_0 \rangle 
\geq \alpha_\ell \|\Delta_0\|_2^2,
\end{align*}
where $\alpha_\ell:= C_r C_\lambda/2$, for $C_r$ defined in \eqref{eq:lemA.2-3}. Let 
\begin{align}
	\mathcal{E} := \left \lbrace \left\lvert \frac{ \langle \triangledown \L_n^\ell(\beta) - \triangledown \L^\ell(\beta), \Delta_0\rangle}{\|\Delta_0\|_1}\right\rvert+\|  \triangledown \L_n^\ell(\beta_0)\|_\infty  \leq \tau_\ell \sqrt{\frac{\log p}{n}}, \forall \beta \in \B_2(r) \setminus \{\beta_0\} \right\rbrace,
\end{align}
for $\tau_\ell := c\cdot ( C_\ell K_X + C_1(T_\star+ L_\star \tau))$, where $c>0$ is an absolute constant, and $T_\star, L_\star$, $\tau$, and $C_1$ are constants which are defined in Lemma \ref{lem:A.8}. On $\mathcal{E}$, we note that 
\begin{align*}
	\langle \triangledown \L_n^\ell(\beta) - \triangledown \L_n^\ell(\beta_0),\beta - \beta_0\rangle \geq \alpha_\ell \|\beta - \beta_0\|_2^2 - \tau_\ell \sqrt{\frac{\log p}{n}}\|\beta - \beta_0\|_1.
\end{align*}             
Therefore, it is sufficient to show that $\P(\mathcal{E}) \geq 1-\epsilon$.

First, we show $\|\triangledown \L_n^\ell(\beta_0) \|_\infty \leq c_0 C_\ell K_X\sqrt{\frac{\log p}{n}}$ w.p $1-\epsilon/2$, where $c_0$ is an absolute constant. Note,
	\begin{align*}
	    \|\triangledown \L_n^\ell(\beta_0) \|_\infty = \max_{1\leq j\leq p} \left\lvert \frac{1}{n} \sum_{i=1}^n \ell'(\x_i^\top \beta_0,\z_i)\x_{ij}  \right\rvert.
	\end{align*}
We use the following Lemma \ref{lem:sub_gaussian_max} to bound $\|\triangledown \L_n^\ell(\beta_0) \|_\infty$.
	
\begin{lemma}\label{lem:sub_gaussian_max}
Suppose $(\xi_{ij})_{1\leq i\leq n,1\leq j \leq p}$ are random variables such that $\xi_{ij}$ is a mean-zero sub-gaussian with parameter $C_\xi$ and $(\xi_{ij})_{i=1}^n$ are independent for any $j \in \{1,\dots,p\}$. Then,
\[\P\left(\max_{1\leq j\leq p} \left\lvert \frac{1}{n} \sum_{i=1}^n \xi_{ij}  \right\rvert \geq  3C_\xi \sqrt{\frac{\log p}{n}} \right)\leq \frac{1}{p^7}. \]
\end{lemma}
\begin{proof}
$ \frac{1}{n} \sum_{i=1}^n \xi_{ij}$ is sub-gaussian with parameter $C_\xi/\sqrt{n}$. By taking a union bound, for any $t\geq 0 $ we have
\begin{align*}
    \P\left(  \max_{1\leq j\leq p} \left\lvert \frac{1}{n} \sum_{i=1}^n \xi_{ij}  \right\rvert \geq t \sqrt{\frac{\log p}{n}}\right) \leq  \exp (- t^2 \log p/C_\xi^2 + \log 2p) 
\end{align*}
Take $t^2 = 9C_\xi^2$. Then $\|\triangledown \L_n(\beta_0) \|_\infty \leq t\sqrt{\frac{\log p}{n}}$ with probability at least $1-1/p^7$.
	
\end{proof}
Taking $\xi_{ij} = \ell'(\x_i^\top \beta_0,\z_i)\x_{ij}$ in Lemma \ref{lem:sub_gaussian_max}, we have, $\E[|\xi_{ij}|^{k}]^{1/k} \leq C_\ell \E[|\x_{ij}|^k]^{1/k} \leq  \sqrt{k} C_\ell K_X$ for any $k\geq 1$ by Assumption {\bf A\ref{a1'}'} and {\bf A\ref{a2}}. Also, $\E[\xi_{ij}] = 0$ since $\E[\z_i|\x_i] = A'(h(\x_i^\top \beta_0))$. Therefore, $\xi_{ij}$ is a mean-zero sub-gaussian variable with parameter $c_0 C_\ell K_X$, where $c_0$ is an absolute constant. Then from Lemma \ref{lem:sub_gaussian_max}, $\|\triangledown \L_n(\beta_0) \|_\infty \leq c_0 C_\ell K_X\sqrt{\frac{\log p}{n}}$ with probability at least $1-1/p^7$. Thus, for $n \geq C \cdot (2/\epsilon)^{1/7}$ for a sufficiently large constant $C$, we have $\|\triangledown \L_n(\beta_0) \|_\infty \leq c_0 C_\ell K_X\sqrt{\frac{\log p}{n}}$ w.p. at least $1-\epsilon/2$, in the regime of interest $p \gg n$. 

The bound for the second term can be obtained by taking advantage of the uniform convergence result of the directional derivative of the loss function in \citet{Mei2018-ec}, and we summarize the result for the case of $\L_n(\beta)$ in Lemma \ref{lem:A.8}. Taking $\delta =\epsilon/2$ in Lemma \ref{lem:A.8}, we obtain $\P(\mathcal{E}) \geq 1-\epsilon$, as desired.

\begin{lemma}[Theorem 3 in \citealp{Mei2018-ec}]\label{lem:A.8}
	There exists a constant $C_1>0$, which depends on model parameters $(r,K_X,C_\ell,L_\ell)$ and $\delta$ such that 
	\begin{align}
		\P\left( \sup_{\beta \in \B_2(r)\setminus \{\beta_0\}} \left\lvert \frac{ \langle \triangledown \L_n^\ell(\beta) - \triangledown \L^\ell(\beta), \beta-\beta_0\rangle}{\|\beta-\beta_0\|_1}\right\rvert \leq C_1(T_\star + L_\star \tau) \sqrt{\frac{ \log (np)}{n}}\right) \geq 1-\delta,
	\end{align}
	where $\tau = K_X \max\{C_\ell, L_\ell^{1/3} \}$, $T_\star = C_\ell C_X$, and $L_\star = C_\rho + C_\ell (C_bC_d+1)$.
\end{lemma}

\begin{proof}[Proof of Lemma \ref{lem:A.8}]
	We verify Assumptions 2-5 in \citet{Mei2018-ec}. From the proof of Lemma \ref{lem:A.4}, we have already checked that Assumption 2 and 3 in \citet{Mei2018-ec} are satisfied under {\bf A\ref{a1'}'} and {\bf A\ref{a2}}. We thus check Assumptions 4 and 5 in \citet{Mei2018-ec}, which verify the existence of $T_\star$ and $L_\star$, where $T_\star$ is a constant such that $\|\ell'(\x^\top \beta, \z)\x\|_\infty \leq T_\star$ a.s., and $L_\star$ is a Lipschitz constant for the function $g(\cdot,\cdot) \rightarrow \mathbb{R}$, which is defined as follows:
	\begin{align*}
		g(\x^\top (\beta-\beta_0), (\x,\z))=\langle \ell'(\x^\top \beta, \z)\x, \beta-\beta_0\rangle.
	\end{align*}
	For the existence of $T_\star$,
	$\|\ell'(\x^\top \beta, \z)\x\|_\infty \leq C_\ell C_X$, by Assumption {\bf A\ref{a1'}'} and {\bf A\ref{a2}}. Thus we can let $T_\star = C_\ell C_X$. For Assumption 5 in \citet{Mei2018-ec},
\begin{align*}
	\langle \ell'(\x^\top \beta, \z)\x , \beta-\beta_0\rangle &= (A'(h(\x^\top \beta))-\z)h'(\x^\top \beta)\x^\top (\beta-\beta_0)\\
	&=g(\x^\top (\beta-\beta_0); (\x,\z))
\end{align*}
for $g(t;(\x,\z)) := \ell'(t+\x^\top \beta_0,\z)t$. We show that $g(t;(\x,\z))$ is Lipschitz with respect to $t$ under Assumption {\bf A\ref{a3}}. Taking a derivative with respect to $t$,
\begin{align*}
	g'(t;(\x,\z)) = \ell''(t+\x^\top \beta_0,\z)t+\ell'(t+\x^\top \beta_0,\z).
\end{align*}
Then,
\begin{align*}
	|g'(t;(\x,\z))| &\leq |\ell''(t+\x^\top \beta_0,\z)(t+\x^\top \beta_0)|+|\ell''(t+\x^\top \beta_0,\z)\x^\top \beta_0|+ |\ell'(t+\x^\top \beta_0,\z)|\\
	&\leq C_\rho + C_\ell (C_bC_d+1)
\end{align*}
by Assumptions {\bf A\ref{a1'}'} and {\bf A\ref{a2}}, noting $|\x^\top \beta_0|\leq C_bC_d$ a.s. by Assumption {\bf A\ref{a4}}. Therefore $L_\star$ can be taken as $L_\star = C_\rho + C_\ell (C_bC_d+1)$.
\end{proof}

\subsection{Proof of Proposition \ref{prop:5.2} }\label{supp_sec:prop:5.2}

\begin{align*}
	\langle \triangledown \L_n^{s}(\beta) - \triangledown \L_n^{s} (\beta_0) ,\beta - \beta_0 \rangle 
	&= \frac{1}{n}\sum_{i=1}^n  (\mu(\x_i^\top \beta) -\mu(\x_i^\top \beta_0))\x_i^\top (\beta-\beta_0)\\
	&= \frac{1}{n}\sum_{i=1}^n  \mu'(\x_i^\top \beta_0 + v \x_i^\top  (\beta-\beta_0 ))(\x_i^\top (\beta-\beta_0))^2
\end{align*}

Then from the proof of Proposition 2 in \citet{Negahban2012-bd}, there exist positive constants $\kappa_1$ and $\kappa_2$ such that
\begin{align*}
 \langle \triangledown \L_n^{s} (\beta) - \triangledown \L_n^{s} (\beta_0) ,\beta - \beta_0 \rangle 
\geq \kappa_1 \|\Delta\|_2 \left(\|\Delta\|_2 - \kappa_2 \sqrt{\frac{\log p}{n}}\|\Delta\|_1 \right),\quad \forall \beta \in \B_2(1;\beta_0)
\end{align*}
with probability at least $1-c_1\exp(-c_2n)$, for some $c_1,c_2>0$. The result \eqref{eq:prop5.2} follows from the basic arithmetic inequality $2ab\leq (a+b)^2$.

\subsection{Proof of Theorem \ref{thm:5.1}}
\label{supp_sec:thm5.1}

First, we address the $\ell_1$ and $\ell_2$ error bounds for the non-convex estimator. We characterize the $\ell_1$ and $\ell_2$ error bounds of a stationary point following similar lines as in the proof of Theorem 1 in \citet{Loh2017-jk}, which established the result with a different tolerance function and penalty. Since $ \beta_0 $ is feasible, by the first order optimality condition, we have the following inequality
\begin{equation*}
\langle \triangledown \L_n(\widetilde{\beta}_\ell^H) +\lambda v(\widetilde{\beta}_\ell^H),\beta_0 - \widetilde{\beta}_\ell^H\rangle \geq 0,
\end{equation*}
for $v(\widetilde{\beta}_\ell^H) \in \partial \|\widetilde{\beta}_\ell^H\|_1$. We let $ \widetilde{\Delta} := \widetilde{\beta}_\ell^H - \beta_0 $. By applying RSC condition \eqref{eq:prop5.1},
\begin{equation*}
\alpha_\ell \|\widetilde{\Delta}\|_2^2 -\tau_{\ell}  \sqrt{\frac{\log p}{n}}\|\widetilde{\Delta}\|_1+\langle \triangledown \L_n(\beta_0) +\lambda v(\widetilde{\beta}_\ell^H),\widetilde{\beta}_\ell^H-\beta_0 \rangle \leq 0,
\end{equation*}
By convexity of $\|\cdot\|_1$,
\begin{equation}
\lambda\|\beta_0\|_1 - \lambda\|\widetilde{\beta}_\ell^H\|_1 \geq -\lambda v(\widetilde{\beta}_\ell^H)^\top \widetilde{\Delta}.
\end{equation}
Therefore,
\begin{equation*}
\alpha_\ell\|\widetilde{\Delta}\|_2^2 -\tau_\ell \sqrt{\frac{\log p}{n}}\|\widetilde{\Delta}\|_1+\langle \triangledown \L_n(\beta_0) ,\widetilde{\beta}_\ell^H-\beta_0 \rangle + \lambda(\|\widetilde{\beta}_\ell^H\|_1 -\|\beta_0\|_1)\leq 0,
\end{equation*}
That is,
\begin{align*}
\alpha_\ell\|\widetilde{\Delta}\|_2^2 &\leq \tau_\ell\sqrt{\frac{\log p}{n}} \|\widetilde{\Delta}\|_1+|\langle \triangledown \L_n(\beta_0) ,\widetilde{\beta}_\ell^H-\beta_0 \rangle| + \lambda(\|\beta_0\|_1-\|\widetilde{\beta}_\ell^H\|_1)\\
&\leq \tau_\ell\sqrt{\frac{\log p}{n}} \|\widetilde{\Delta}\|_1+\|\triangledown \L_n(\beta_0)\|_\infty \|\widetilde{\Delta}\|_1 + \lambda (\|\widetilde{\Delta}_S\|_1 - \|(\widetilde{\beta}_\ell^H)_{S^c}\|_1)
\end{align*}
Since $\tau_\ell \sqrt{\frac{\log p}{n}} +\|\triangledown \L_n(\beta_0)\|_\infty\leq \frac{\lambda}{2} $,
\begin{align*}
\alpha_\ell\|\widetilde{\Delta}\|_2^2&\leq \frac{\lambda}{2} (\|\widetilde{\Delta}_S\|_1+\|\widetilde{\Delta}_{S^c}\|_1) + \lambda (\|\widetilde{\Delta}_S\|_1 - \|(\widetilde{\beta}_\ell^H)_{S^c}\|_1)\\
& = \frac{3\lambda}{2}\|\widetilde{\Delta}_S\|_1-\frac{\lambda}{2}\|\widetilde{\Delta}_{S^c}\|_1.
\end{align*}
In particular,
\begin{align}
	&\alpha_\ell\|\widetilde{\Delta}\|_2^2 \leq \frac{3\lambda}{2}\|\widetilde{\Delta}_S\|_1 \leq \frac{3\sqrt{s_0}\lambda}{2}\|\widetilde{\Delta}\|_2,\label{eq:l2}\\
	& \|\widetilde{\Delta}_{S^c}\|_1\leq 3\|\widetilde{\Delta}_S\|_1\label{eq:cone}
\end{align}
$\ell_2$ bound follows from \eqref{eq:l2} and
\begin{align*}
	\|\widetilde{\Delta}\|_1 =	\|\widetilde{\Delta}_S\|_1 +\|\widetilde{\Delta}_{S^c}\|_1 	\leq 4\|\widetilde{\Delta}_S\|_1  \leq 4\sqrt{s_0} \|\widetilde{\Delta}\|_2.
\end{align*}
Now, we address the $\ell_1$ and $\ell_2$ error bounds for the convex estimator. To do so, we need to establish a different RSC condition, introduced by \citet{Negahban2012-bd} as follows:

\begin{definition}[restricted strong convexity in \citealp{Negahban2012-bd}]
	For a given set $\mathbb{S}$, the loss function $\L_n$ satisfies restricted strong convexity (RSC) with parameter $\alpha > 0$ if
	\begin{align}\label{def:RSC2}
\L_n(\beta) -\L_n(\beta_0) - \langle \triangledown \L_n(\beta_0), \beta- \beta_0 \rangle \geq \alpha \| \beta - \beta_0\|_2^2 \quad \mbox{for all } \beta - \beta_0 \in \mathbb{S}.
\end{align}
\end{definition}

In the following Lemma \ref{supp_lem:RSC}, we show that the RSC condition \ref{def:RSC} with $\tau_{n,p}(t) = \tau (\log p/n) t^2$ and $\Omega = \B_2(\delta;\beta_0)$ implies the RSC condition in \citet{Negahban2012-bd}. 
%% Lemma A.7
\begin{lemma}\label{supp_lem:RSC}
	The RSC condition \ref{def:RSC} with $\tau_{n,p}(t) = \tau (\log p/n) t^2$ and $\Omega = \B_2(\delta;\beta_0)$ implies \eqref{def:RSC2} with parameter $\alpha/4$ and 
$$\mathbb{S} = \{\Delta\in \R^p; \|\Delta_{S^c}\|_1 \leq 3\|\Delta_S\|_1\} \cap \{\Delta\in \R^p; \|\Delta\|_2 \leq \delta\},$$
where $ S \subseteq \{1,\dots,p\}$ is the support of $\beta_0$ and $s_0 := |S|$, given the sample size $n \geq (32\tau s_0/\alpha)\log p$.
\end{lemma}

Provided that Lemma \ref{supp_lem:RSC} is true and given the condition of $\lambda_s$ in Theorem \ref{thm:5.1}, the $\ell_2$ error bound
\begin{align}\label{supp_eq:l2}
	\| \widehat{\beta}_s^H -\beta_0\|_2 \leq  \frac{8\sqrt{s_0}\lambda_s}{\alpha_s}
\end{align}
can be obtained by applying Theorem 1 in \citet{Negahban2012-bd}. Also it is well known that an error vector $\widehat{\beta} - \beta_0$, where $\widehat{\beta}$ is a solution of Lasso optimization problem, belongs to the cone $ \{\Delta\in \R^p; \|\Delta_{S^c}\|_1 \leq 3\|\Delta_S\|_1\}$. Thus $\|\widehat{\beta}_s^H-\beta_0\|_1 \leq 4\|(\widehat{\beta}_s^H-\beta_0)_{S}\|_1 \leq 4 \sqrt{s_0} \|\widehat{\beta}_s^H-\beta_0\|_2 $. Applying this inequality to \eqref{supp_eq:l2} gives an $\ell_1$ bound. Now we present the proof of Lemma \ref{supp_lem:RSC}.
\begin{proof}[Proof of Lemma \ref{supp_lem:RSC}]
For any $\beta$ such that $\beta -\beta_0 \in \mathbb{S}$, we have,
\begin{align}
	\L_n(\beta) &= \L_n(\beta_0) + \int \triangledown \L_n(\beta_0 + t(\beta - \beta_0))^\top (\beta - \beta_0) dt \nonumber \\
	&= \L_n(\beta_0) + \triangledown \L_n(\beta_0)^\top (\beta - \beta_0) +\int_0^1 \frac{1}{t} (\triangledown \L_n(\beta_0 + t(\beta - \beta_0))-\triangledown \L_n(\beta_0))^\top t(\beta - \beta_0) dt \label{eq:supp_thm5.1-1} .
\end{align}
By the RSC condition \ref{def:RSC} with $\tau_{n,p}(t) = \tau (\log p/n) t^2$ and $\Omega = \B_2(\delta;\beta_0)$, for any $\beta \in \B_2(\delta;\beta_0)$ it holds that
\begin{align}
	(\triangledown \L_n(\beta_0 + t(\beta - \beta_0))-\triangledown \L_n(\beta_0))^\top t(\beta - \beta_0)  \geq t^2 \left( \alpha\|\beta-\beta_0\|_2^2 -\tau \left(\frac{\log p}{n}\right) \|\beta-\beta_0\|_1^2 \right). \label{eq:supp_thm5.1-2}
\end{align}
Applying \eqref{eq:supp_thm5.1-2} to \eqref{eq:supp_thm5.1-1},
\begin{align*}
	\L_n(\beta)- \L_n(\beta_0) - \triangledown \L_n(\beta_0)^\top (\beta - \beta_0) &\geq\int_0^1 t \left( \alpha\|\beta-\beta_0\|_2^2 -\tau \left(\frac{\log p}{n}\right) \|\beta-\beta_0\|_1^2 \right) dt\\
	&= \frac{\alpha}{2}\|\beta-\beta_0\|_2^2 -\frac{\tau}{2} \left(\frac{\log p}{n}\right) \|\beta-\beta_0\|_1^2.
	\end{align*}
Since $\beta -\beta_0 \in \mathbb{S}$, $\|\beta- \beta_0\|_1 \leq 4\sqrt{s_0}\|\beta- \beta_0\|_2$. Therefore,
\begin{align*}
	\L_n(\beta)- \L_n(\beta_0) - \triangledown \L_n(\beta_0)^\top (\beta - \beta_0) & \geq \left( \frac{\alpha}{2} -8 \tau s_0 \frac{\log p}{n} \right)\|\beta-\beta_0\|_2^2 \geq \frac{\alpha}{4} \|\beta-\beta_0\|_2^2
\end{align*}
where the last inequality is from a given sample condition $n \geq (32 \tau s_0/\alpha)\log p$.
 \end{proof}

\subsection{Proof of Corollary \ref{cor:3}}\label{supp_sec:cor3}

The Corollary \ref{cor:3} essentially follows from Proposition \ref{prop:5.1}, Proposition \ref{prop:5.2}, and Theorem \ref{thm:5.1}. The main conditions to verify are $\|\triangledown \L_n^\ell(\beta_0)\|_\infty, \|\triangledown \L_n^s(\beta_0)\|_\infty = O(\sqrt{\frac{\log p}{n}})$ with high probability and Assumption {\bf A\ref{a3}}, since we have already shown that Assumption {\bf A\ref{a2}} is satisfied for the noisy labels problem in the proof of Corollary \ref{cor:2}.

We first address bounds for $\|\triangledown \L_n^\ell(\beta_0)\|_\infty$ and $\|\triangledown \L_n^s(\beta_0)\|_\infty$. We note that $\| \triangledown \L_n(\beta_0)\|_\infty$, for $\L_n \in (\L_n^\ell, \L_n^s)$, has the form
	\begin{align*}
	    \|\triangledown \L_n(\beta_0) \|_\infty = \max_{1\leq j\leq p} \left\lvert \frac{1}{n} \sum_{i=1}^n \xi_{ij}  \right\rvert,
	\end{align*}
where $\xi_{ij} = \{(A'(h_{LN}(\x_i^\top \beta_0))-\z_i) h_{LN}'(\x_i^\top \beta_0)\}\x_{ij}$ if $\L_n = \L_n^\ell$, and $ \xi_{ij} = \{A'(\x_i^\top \beta_0)-T(\z_i) \}\x_{ij}$ if $\L_n = \L_n^s$. Also $\E[\xi_{ij}]=0, \forall i,j$ and $(\xi_{ij})_{i=1}^n$ are independent for any $j \in \{1,\dots,p\}$.

From Lemma \ref{lem:h}, we have $|(A'(h_{LN}(\x_i^\top \beta_0))-\z_i)h_{LN}'(\x_i^\top \beta_0)| \leq 1$ and $ |A'(\x_i^\top \beta_0)-T(\z_i)| \leq 1$ a.s. Thus $\E[|\xi_{ij}|^{k}]^{1/k} \leq \E[|\x_{ij}|^k]^{1/k} \leq  \sqrt{k} K_X$ for any $k\geq 1$ by Assumption {\bf A\ref{a1'}'}. In particular, $\xi_{ij}$ is mean-zero sub-gaussian with parameter $c K_X$ where $c>0$ is an absolute constant. Therefore, by Lemma \ref{lem:sub_gaussian_max}, $\|\triangledown \L_n(\beta_0) \|_\infty \leq c'K_X\sqrt{\frac{\log p}{n}}$ with probability at least $1-1/p^7$ for a different constant $c'>0$.

Now we show that Assumption {\bf A\ref{a3}} holds. We recall $\ell''(t,z) = \rho_I(t) + \rho_R(t,z)$ for $\rho_I(t) = A''(h_{LN}(t))h'_{LN}(t)^2$, $\rho_R(t,z) = (A'(h_{LN}(t))-z)h''_{LN}(t)$ where $A(t) = \log (1+ \exp(t))$ and $h_{LN}(\cdot)$ defined in Section \ref{sec:noiselabel}. In the following, we show that both $\sup_t |\rho_I(t)t|$ and $\sup_t |\rho_R(t,z)t|$ are bounded by an absolute constant. First, we let $a = 1-\rho_1-\rho_0$, $b = \rho_0$.
Since
\begin{align*}
	h'_{LN}(t) = \frac{ a\mu(t) (1-\mu(t))}{ (a \mu(t)+b) (1-a\mu(t)-b)},
\end{align*}
we have,
\begin{align*}
	t\rho_I(t) &= tA''(h_{LN}(t)) h_{LN}'(t)^2\\
	& = t (a\mu(t)+b)(1-a\mu(t)-b) \left(   \frac{a \mu(t) (1-\mu(t)) }{ (a \mu(t)+b) (1-a\mu(t)-b)}\right)^2\\
	& = a t\mu(t) (1-\mu(t)) h'_{LN}(t).
\end{align*}
By Lemma \ref{lem:h}, $\|h'_{LN}\|_\infty\leq 1$. Also, with an elementary calculation, it can be shown that
\begin{align*}
	|t\mu(t)(1-\mu(t))| = \frac{|t|e^t}{(1+e^t)^2} \leq 2, \forall t.
\end{align*}
	Therefore, $\sup_t |\rho_I(t)t| \leq 2$. Now we address $\sup_t |\rho_R(t,z)t|$. Since $|A'(h_{LN}(t))-z|\leq 1$ for all $t$ and $z \in \{0,1\}$, we have,
\begin{align*}
	|t\rho_R(t,z)| = |(A'(h_{LN}(t))-z)h''(t)t|
	\leq |h''(t)t|.
\end{align*}
Therefore, it is sufficient to bound $h''(t)t$. Note if $\rho_1=\rho_0=0$, $h''(t)=0$. Therefore $th''(t)$ is trivially bounded. Otherwise, we discuss three cases separately: 1. $\rho_1, \rho_0 >0$, 2. $\rho_1 >0, \rho_0 = 0$, and 3. $\rho_1=0, \rho_0>0$.
First, we note from the proof of Lemma \ref{lem:h}, we have,
\begin{align*}
	th''_{LN}(t) = t h_{LN}'(t)\{ 1- 2\mu(t) - h_{LN}'(t)(1-2(a\mu(t)+b))\}.
\end{align*}

\noindent Case 1: $\rho_1, \rho_0 >0$\\
In this case, $\sup_t|th'_{LN}(t)| <\infty$, since
\begin{align}\label{eq:cor5.1-1}
	\sup_t|th'_{LN}(t)| \leq  \frac{a\sup_t| t\mu(t) (1-\mu(t)) |}{\inf_t (a \mu(t)+b) (1-a\mu(t)-b)} \leq \frac{2a}{\inf_{(\rho_0 \wedge \rho_1)\leq p\leq 1-(\rho_0 \wedge \rho_1)}x(1-x)}<\infty.
\end{align}
Also, $ |1- 2\mu(t) - h_{LN}'(t)(1-2(a\mu(t)+b) | \leq 1+ \|h_{LN}'\|_\infty \leq 2$.

For Case 2 and 3, we cannot use the bound \eqref{eq:cor5.1-1} since the denominator becomes zero. With elementary calculations, we can obtain
\begin{align}\label{eq:cor5.1-2}
	th''_{LN}(t) 
	%&=th'_{LN}(t)\{1-2\mu(t) -h_{LN}'(t)(1-2(a\mu(t)+b)) \}\\
	&=-\frac{a(1-b-a)t\mu(t)^2(1-\mu(t))}{ (a \mu(t)+b) (1-a\mu(t)-b)^2}  +\frac{ abt\mu(t)(1-\mu(t))^2}{(a \mu(t)+b)^2(1-a\mu(t)-b)}
\end{align}

\noindent Case 2: $\rho_1 >0, \rho_0 = 0$\\
Equivalently, $b=0, a=1-\rho_1$, therefore the second term in \eqref{eq:cor5.1-2} does not exist, and
\begin{align*}
	th''_{LN}(t)=-\frac{(1-a)t\mu(t)(1-\mu(t))}{ (1-a\mu(t))^2}  \end{align*}
Therefore $|th''_{LN}(t) |\leq 2/(1-a) = 2/\rho_1$.

\noindent Case 3: $\rho_1 =0, \rho_0 >0$\\
In Case 3, $a=1-\rho_0, b= \rho_0, a+b =1$. The first term in \eqref{eq:cor5.1-2} does not exist, and
\begin{align*}
	th''_{LN}(t)=\frac{ bt\mu(t)(1-\mu(t))}{(a \mu(t)+b)^2},
	\end{align*}
	noting $1-a\mu(t)-b = a(1-\mu(t))$. Therefore, $|th''_{LN}(t) |\leq 2/b = 2/\rho_0$.

%%% Proof of Normality of de-biasing estimator
\subsection{Proof of Theorem \ref{thm:5.2}}\label{pf:debiasing}
First, for a given $\psi$, we let $\widehat{\beta}^{\textnormal{db}} = \widehat{\beta}^{\textnormal{db}}(\psi)$, $\widehat{\Theta} = \widehat{\Theta}(\psi)$, and $\Theta = \Theta(\psi)$ for ease of notation. For any fixed $j \in \{1,\dots,p\}$, we have
\begin{align}\label{eq:de-basing_estimator}
\widehat{\beta}^{\textnormal{db}}_j -\beta_{0j} &= \widehat{\beta}_j-\beta_{0j} -\widehat{\Theta}_j^\top  \left( \frac{1}{n}\sum_{i=1}^n \psi(\x_i^\top \widehat{\beta},\z_i)\x_i\right).
\end{align}
Let $\widehat{\Delta}:= \widehat{\beta}-\beta_0$. By the Taylor expansion,
\begin{align*}
\widehat{\Theta}_j^\top  &\left( \frac{1}{n}\sum_{i=1}^n \psi(\x_i^\top \widehat{\beta},\z_i)\x_i\right)\\
&= n^{-1}\sum_{i=1}^n \left( \psi(\x_i^\top \beta_0,\z_i)+ \psi'(v_i,\z_i)(\x_i^\top \widehat{\beta}-\x_i^\top \beta_0)\right)\widehat{\Theta}_j^\top  \x_i\\
&=n^{-1}\sum_{i=1}^n \left( \psi(\x_i^\top \beta_0,\z_i)+ \psi'(\x_i^\top \widehat{\beta},\z_i)\x_i^\top \widehat{\Delta}+\{\psi'(v_i,\z_i)-\psi'(\x_i^\top \widehat{\beta},\z_i)\}\x_i^\top \widehat{\Delta}\right)\widehat{\Theta}_j^\top \x_i
\end{align*}
for $v_i$ such that $|v_i -\x_i^\top \widehat{\beta}| \leq |\x_i^\top (\widehat{\beta} - \beta_0)|$. 

First, we address the last term and show that it is $o_p(n^{-1/2})$. 
\begin{align}\label{eq:28}
&n^{-1}\sum_{i=1}^n \{\psi'(v_i,\z_i)-\psi'(\x_i^\top \widehat{\beta},\z_i)\}\x_i^\top \widehat{\Delta}\widehat{\Theta}_j^\top  \x_i\nonumber \\
&\leq  n^{-1}\sum_{i=1}^n |\psi'(v_i,\z_i)-\psi'(\x_i^\top \widehat{\beta},\z_i)||\x_i^\top \widehat{\Delta}| |\widehat{\Theta}_j^\top  \x_i|
\end{align}
From {\bf A\ref{a7}}, $\psi'(t,z)$ is Lipschitz in $t$ with the Lipschitz constant $2L_\psi$,$\forall z$. Thus we have,
\begin{align}\label{eq:29}
	|\psi'(v_i,\z_i)-\psi'(\x_i^\top \widehat{\beta},\z_i)| \leq 2L_\psi |v_i-\x_i^\top \widehat{\beta}|\leq 2L_\psi |\x_i^\top \beta_0-\x_i^\top \widehat{\beta}|,
\end{align}
and by combining \eqref{eq:28}, \eqref{eq:29}, we obtain
\begin{align*}
\frac{1}{n}\sum_{i=1}^n (\psi'(v_i,\z_i)-\psi'(\x_i^\top \widehat{\beta},\z_i))\x_i^\top \widehat{\Delta}\widehat{\Theta}_j^\top  \x_i
&\leq \frac{2L_\psi}{n} \sum_{i=1}^n (\x_i^\top \widehat{\Delta})^2|\widehat{\Theta}_j^\top  \x_i|\\ 
&\leq \frac{2L_\psi}{n}\| \X\widehat{\Delta}\|_2^2  \max_{1\leq i \leq n} |\widehat{\Theta}_j^\top  \x_i|.
\end{align*}
To bound $\| \X\widehat{\Delta}\|_2^2$, we use the following result, which can be obtained by combining Lemma 12 and 15 in \citet{Loh2012-tl}. 
\begin{lemma}\label{lem:subexp_unif}
Suppose $\x_i$ satisfies the sub-gaussian tail condition with the parameter $K_X$, for all $i=1,\dots,n$. For any $u >0$, the following inequality holds with probability at least $1 - 2\exp(- c'n u (1 \wedge u)/2)$,
\begin{align}
\sup_{\|v\|_1\leq \sqrt{s(u)} \|v\|_2}\left\lvert v^\top \left(\sum_{i=1}^n\frac{\x_i\x_i^\top }{n}-\E(\x_i\x_i^\top )  \right)v\right\rvert \leq 27 u K_X^2 \|v\|_2^2
\end{align}
where $ s(u): = (c'n/4 \log p) (u \wedge u^2)$ and $c'$ is a universal constant in Bernstein's inequality (see Corollary 2.8.3 in \citealp{Vershynin2018-xl}), given a sufficient sample size $n \geq (4\log p /c') \max \{(u\wedge u^2), (u\wedge u^2)^{-1}\}$.
\end{lemma}
Then from an application of Lemma \ref{lem:subexp_unif}, we have \begin{align}\label{eq:quadratic_emp_bound}
	\frac{1}{n} \sum_{i=1}^n (\x_i^\top \Delta)^2	&\leq \alpha' \|\Delta\|_2^2 + \tau' \frac{\log p}{n } \|\Delta\|_1^2, \quad \forall \Delta \in \R^p
	\end{align} 
with probability at least $1-2\exp(-c'n)$
	where $\alpha' , \tau' $, and $c'$ are constants which only depend on model parameter $K_X$ and not dimensions $(n,p)$.
Thus we have $ \|\X\widehat{\Delta}\|_2^2/n = O_p(s_0(\log p/n))+O_p(s_0^2(\log p/n)^2) = o_p(n^{-1/2})$ by the rate assumption of $s_0$ in {\bf A\ref{a6}}. Also,
\begin{align*}
\max_{1\leq i \leq n} | \x_i^\top \widehat{\Theta}_j|
&\leq \max_{1\leq i \leq n}|\x_i^\top  (\widehat{\Theta}_j-\Theta_j)| +\max_{1\leq i \leq n}|\x_i^\top \Theta_j|\\
&\leq \max_{1\leq i \leq n}\|\x_i\|_\infty \|\widehat{\Theta}_j-\Theta_j\|_1 + \|\X \Theta_j\|_\infty =O_p(1).
\end{align*}
It holds because $\max_{i,j} |\x_{ij}| \leq C_X$ by {\bf A\ref{a1'}'}, $\|\widehat{\Theta}_j-\Theta_j\|_1=o_p(\sqrt{1/\log p)}$ from the assumption about $\widehat{\Theta}$, and $\|\X \Theta_j\|_\infty =O_p(1)$ from {\bf A\ref{a6}}. Therefore,
\begin{align}\label{eq:82}
	\frac{1}{n}\| \X\widehat{\Delta}\|_2^2  \|\X\widehat{\Theta}_j\|_\infty = o_p(n^{-1/2}),
\end{align}
and we have,
\begin{align}\label{eq:31}
\widehat{\Theta}_j^\top  \psi_n(\widehat{\beta})
=\frac{1}{n} \sum_{i=1}^n \left( \psi(\x_i^\top \beta_0,\z_i)+ \psi'(\x_i^\top \widehat{\beta},\z_i)\x_i^\top \widehat{\Delta}\right)\widehat{\Theta}_j^\top \x_i+o_p(n^{-1/2}).
\end{align}
Combining \eqref{eq:de-basing_estimator} with \eqref{eq:31}, 
\begin{align*}
\widehat{\beta}^{\textnormal{db}}_j -\beta_{0j} 
&= \widehat{\beta}_j-\beta_{0j} - n^{-1}\sum_{i=1}^n \left( \psi(\x_i^\top \beta_0,\z_i)+ \psi'(\x_i^\top \widehat{\beta},\z_i)\x_i^\top \widehat{\Delta}\right)\widehat{\Theta}_j^\top \x_i+o_p(n^{-1/2})\\
&= e_j^\top  \widehat{\Delta}-\widehat{\Theta}_j^\top \psi_n(\beta_0)- n^{-1}\sum_{i=1}^n \left( \psi'_I(\x_i^\top \widehat{\beta})+\psi'_R(\x_i^\top \widehat{\beta},\z_i)\right)\widehat{\Theta}_j^\top \x_i\x_i^\top \widehat{\Delta}+o_p(n^{-1/2}),
\end{align*}
where we use the relationship $\psi'(t,z) = \psi_I'(t) + \psi'_R(t,z)$ in \eqref{eq:psi'}.
Recalling the definition 
$\psi'_{I,n}(\beta) := n^{-1} \sum_{i=1}^n \psi_I'(\x_i^\top \beta)\x_i\x_i^\top $,
\begin{align*}
	n^{-1}\sum_{i=1}^n \psi'_I(\x_i^\top \widehat{\beta})\widehat{\Theta}_j^\top \x_i\x_i^\top \widehat{\Delta} =\widehat{\Theta}_j^\top \left( n^{-1}\sum_{i=1}^n \psi'_I(\x_i^\top \widehat{\beta})\x_i\x_i^\top \right)\widehat{\Delta} = \widehat{\Theta}_j^\top \psi'_{I,n}(\widehat{\beta})\widehat{\Delta}
\end{align*}
thus we have
\begin{align*}
\widehat{\beta}^{\textnormal{db}}_j -\beta_{0j} = -\widehat{\Theta}_j^\top \psi_n(\beta_0) -\underbrace{ n^{-1}\sum_{i=1}^n \psi'_R(\x_i^\top \widehat{\beta},\z_i)\widehat{\Theta}_j^\top \x_i\x_i^\top \widehat{\Delta}}_\textrm{Term-I} + \underbrace{\widehat{\Delta}^\top (e_j - \psi'_{I,n}(\widehat{\beta})\widehat{\Theta}_j)}_\textrm{Term-II}+o_p(n^{-1/2}).
\end{align*}

We will show that the first term $\sqrt{n} \widehat{\Theta}_j^\top \psi_n(\beta_0)$  will converge to the normal distribution. Both remainder terms (Term-I and Term-II) need to be $o_p(n^{-1/2})$. For the second remainder term (Term-II), we have $|\widehat{\Delta}^\top (e_j - \psi'_{I,n}(\widehat{\beta})\widehat{\Theta}_j)| \leq
\|\widehat{\Delta}\|_1 \|e_j -\psi'_{I,n}(\widehat{\beta})\widehat{\Theta}_j\|_\infty =O_p(s_0\sqrt{\log p/n})O_p(\sqrt{\log p/n}) = o_p(n^{-1/2})$ by the rate condition {\bf A\ref{a6}} and the assumptions in the theorem. Now we address the first remainder term (Term-I):
\begin{align}\label{eq:term1}
&n^{-1}\sum_{i=1}^n \psi'_R(\x_i^\top \widehat{\beta},\z_i)\widehat{\Theta}_j^\top \x_i\x_i^\top \widehat{\Delta}\nonumber\\
&=n^{-1}\sum_{i=1}^n \left\lbrace \psi'_R(\x_i^\top \beta_0,\z_i)+\left(\psi'_R(\x_i^\top \widehat{\beta},\z_i)-\psi'_R(\x_i^\top \beta_0,\z_i)\right)\right\rbrace\widehat{\Theta}_j^\top \x_i\x_i^\top \widehat{\Delta}.
\end{align}

We need the following Lemma which establishes a kind of sparse eigenvalue condition. 
\begin{lemma}\label{lem:A.9}
	Let $E \in \R^{n\times n}$ be a random matrix which has a representation $E = \frac{1}{n}\sum_{i=1}^n e_i \x_i\x_i^\top $, for random $(e_i)_{i=1}^n$ such that $\E[e_i|\x_i] = 0$ and $|e_i| \leq c_e$ a.s., and $\x_i$ satisfies the sub-gaussian tail condition with the parameter $K_X$ for all $i$.  Then for any $s,s'\geq 1,$ if $n\geq C (s+s')\log p$ for an absolute constant $C$, there exist constants $c_1, c_2>0$ such that
	\begin{align*}
    P\left( \sup_{\substack{ u \in \B_1(\sqrt{s}) \cap \B_2(1), \\ v \in \B_1(\sqrt{s'}) \cap \B_2(1)}}|u^\top  E v|    \geq  c_1 \sqrt{(s+s') \frac{\log p}{n}} \right) \leq  \frac{c_2}{p^{s+s'}}.
\end{align*}
\end{lemma}
The proof of the Lemma is presented at the end of this section. Now we apply Lemma \ref{lem:A.9} to show that $n^{-1}\sum_{i=1}^n \psi'_R(\x_i^\top \widehat{\beta},\z_i)\widehat{\Theta}_j^\top \x_i\x_i^\top \widehat{\Delta}$ is $o_p(n^{-1/2})$. We have,
\begin{align*}
	 n^{-1}\sum_{i=1}^n \psi'_R(\x_i^\top \beta_0,\z_i )\widehat{\Theta}_j^\top \x_i\x_i^\top \widehat{\Delta} = \widehat{\Theta}_j^\top  \left(n^{-1}\sum_{i=1}^n \psi'_R(\x_i^\top \beta_0,\z_i )\x_i\x_i^\top \right)\widehat{\Delta} = \widehat{\Theta}_j^\top  E^R\widehat{\Delta}
\end{align*}
where we define $E^R:=n^{-1}\sum_{i=1}^n \psi'_R(\x_i^\top \beta_0,\z_i )\x_i\x_i^\top $. From the condition of $\widehat{\beta}$ in Theorem \ref{thm:5.2}, we have  $\widehat{\Delta} /\|\widehat{\Delta} \|_2 \in \B_1(\sqrt{c s_0}) \cap \B_2(1)$ for a constant $c>0$. Also, $\|\widehat{\Theta}_j\|_1 \leq \|\widehat{\Theta}_j - \Theta_j\|_1+ \|\Theta_j\|_1$ and $\|\widehat{\Theta}_j\|_2 \geq \|\Theta_j\|_2-\|\widehat{\Theta}_j - \Theta_j\|_2\geq  \|\Theta_j\|_2-\|\widehat{\Theta}_j - \Theta_j\|_1 $. Define an event $\mathcal{E}_n:= \{\|\widehat{\Theta}_j - \Theta_j\|_1  \leq 0.5\|\Theta_j\|_2\}$.
Then
\begin{align*}
    \frac{\|\widehat{\Theta}_j\|_1 }{\|\widehat{\Theta}_j\|_2} \leq \frac{ \|\Theta_j\|_1+\|\widehat{\Theta}_j - \Theta_j\|_1}{\|\Theta_j\|_2-\|\widehat{\Theta}_j - \Theta_j\|_1} \leq 3 \frac{ \|\Theta_j\|_1}{\|\Theta_j\|_2}
\end{align*}
on $\mathcal{E}_n$. We note that $\Theta_j$ is at most $s_*+1$ sparse vector, recalling the definition $s_* := \max_{1\leq j \leq p} \|\Theta_{j,-j}\|_0$. Also, $\|\Theta\|_2 \asymp 1$, since $\|\Theta\|_2 = \lambda_{\textnormal{min}} ^{-1}(\E[\psi'_I(\x^\top\beta_0)\x\x^\top])$ and the minimum eigenvalue of $\E[\psi'_I(\x^\top\beta_0)\x\x^\top]$ can be shown to be bounded above and also bounded below by a positive constant. More concretely, for any unit vector $u$,
\begin{align*}
    u^T \E [\psi_I'(\x^T\beta_0)\x\x^\top] u \geq   \E [\psi_I'(\x^\top\beta_0)(\x^\top u)^2 \mathbbm{1} \{ |\x^\top \beta_0 | \leq \tau_c\}]\geq \inf_{|t|\leq \tau_c} \psi'_I(t) C_\lambda/2, 
\end{align*}
and
\begin{align*}
    u^T \E [\psi_I'(\x^T\beta_0)\x\x^\top] u \leq C_\psi \E[(\x^\top u)^2] \leq 2C_\psi K_X^2
\end{align*}
for $\tau_c := (2 c_b^2 K_X^2 \log (16K_X^2/C_\lambda))^{1/2}$ using Lemma \ref{lem:A.5}, Assumptions {\bf A\ref{a1'}'} and {\bf A\ref{a7}}, where $c_b$ is a constant such that $\|\beta_0\|_2 \leq c_b$, which exists by the condition $\|\beta_0\|_2= O(1)$.

 Thus $\|\Theta_j\|_1 \leq \sqrt{ s_*+1 } \|\Theta_j\|_2$, $\widehat{\Theta}_j/\|\widehat{\Theta}_j\|_2  \in \B_1(\sqrt{9(s_*+1)}) \cap \B_2(1)$ on $\mathcal{E}_n$. Also, we have $\P(\mathcal{E}_n) \underset{n}{\rightarrow} 1$ by $\| \widehat{\Theta}_j - \Theta_j\|_1 = o_p(1/\sqrt{\log p})$ and $\|\Theta_j\|_2  \asymp 1$. Then on $\mathcal{E}_n$,
\begin{align*}
    | \widehat{\Delta} ^\top   E^R \widehat{\Theta}_j| \leq \|\widehat{\Delta}\|_2 \|\widehat{\Theta}_j\|_2   \sup_{\substack{u \in \B_1(\sqrt{s}) \cap \B_2(1), \\ v \in \B_1(\sqrt{s'}) \cap \B_2(1)}} |u^\top E v| ,
\end{align*}
for $s = cs_0$ and $s' = 9(s_*+1)$. Since $\|\widehat{\Delta}\|_2 = O_p (\sqrt{s_0 \log p/n})$ and $\|\widehat{\Theta}_j\|_2   = O_p(1)$,
\begin{align*}
    \|\widehat{\Delta}\|_2 \|\widehat{\Theta}_j\|_2   \sup_{\substack{u \in \B_1(\sqrt{s}) \cap \B_2(1), \\ v \in \B_1(\sqrt{s'}) \cap \B_2(1)}}|u^\top E v| 
    & = O_p ( \sqrt{\frac{s_0 \log p}{n}}) \cdot O_p( \sqrt{\frac{ (s_0 + s_*) \log p}{n}})= o_p(n^{-1/2})
\end{align*}
on $\mathcal{E}_n$, where the last inequality is from the rate conditions $s_0, s_*  = o(\sqrt{n}/\log p)$ from {\bf A\ref{a6}}. Since $\P(\mathcal{E}_n) \underset{n}{\rightarrow} 1$, we conclude $| \widehat{\Delta} ^\top   E^R \widehat{\Theta}_j|  = o_p(n^{-1/2})$. 

For the second term in \eqref{eq:term1},
\begin{align*}
&n^{-1}\sum_{i=1}^n \left\lvert\psi'_R(\x_i^\top \widehat{\beta},\z_i)-\psi'_R(\x_i^\top \beta_0,\z_i)\right\rvert|\widehat{\Theta}_j^\top \x_i||\x_i^\top \widehat{\Delta}|
\leq L_\psi \|\X\widehat{\Theta}_j\|_\infty \frac{1}{n}\|\X \widehat{\Delta}\|_2^2
\end{align*}
where we use {\bf A\ref{a7}} that $\psi'_R(t,z)$ is $L_\psi$-Lipschitz in $t$ for any $z$. Then from \eqref{eq:82}, we have that the second term is $ o_p(n^{-1/2})$. Therefore, combining the results we obtain 
$$n^{-1}\sum_{i=1}^n \psi'_R(x_i^\top \widehat{\beta},\z_i)\widehat{\Theta}_j^\top x_ix_i^\top \widehat{\Delta} =  o_p(n^{-1/2}).$$
So far, we have obtained,
\begin{align*}
\widehat{\beta}^{\textnormal{db}}_j -\beta_{0j} =-\widehat{\Theta}_j^\top \psi_n(\beta_0) +o_p(n^{-1/2}).
\end{align*}
It remains to show that 
\begin{align*}
\frac{\sqrt{n}\widehat{\Theta}_j^\top \psi_n(\beta_0)}{\sqrt{(\Theta^\top \E[\psi(\x^\top \beta_0,\z)^2\x\x^\top ] \Theta})_{jj}} \overset{d}{\rightarrow} \mathcal{N}(0,1).
\end{align*}
By CLT, 
\begin{align*}
	\frac{1}{\sigma \sqrt{n}}\sum_{i=1}^n \psi(\x_i^\top \beta_0,\z_i)\x_i^\top \Theta_j \rightarrow \mathcal{N}(0,1)
\end{align*}
where 
\begin{align*}
	\sigma^2 &= \textnormal{Var}(\psi(\x^\top \beta_0,\z)\x^\top \Theta_j) =  \E[\psi(\x^\top \beta_0,\z)^2(\x^\top \Theta_j)^2]
\end{align*}
since  $\E[\psi(\x^\top \beta_0,\z)] = 0$ by \eqref{eq:estimating_eq}. Thus it is sufficient to show
$\sqrt{n}\widehat{\Theta}_j^\top \psi_n(\beta_0)=\sqrt{n}\Theta_j^\top \psi_n(\beta_0) + o_p(1)$ to conclude. Indeed, we have,
\begin{align*}
|\sqrt{n}(\widehat{\Theta}_j
-\Theta_j)^\top \psi_n(\beta_0)| 
\leq \sqrt{n}\|\widehat{\Theta}_j-\Theta_j\|_1\|\psi_n(\beta_0)\|_\infty = o_p(1). \end{align*}
This holds because, by the condition of $\widehat{\Theta}$, $\|\widehat{\Theta}_j-\Theta_j\|_1= o_p(\sqrt{1/\log p})$ and $\|\psi_n(\beta_0)\|_\infty =O_p(\sqrt{\log p/n})$. Recalling the definition of $\psi_n$, we have,
\begin{align*}
\|\psi_n(\beta_0)\|_\infty =  \max_{1\leq j \leq p} \left\lvert\frac{1}{n} \sum_{i=1}^n \psi(\x_i^\top \beta_0, \z_i)\x_{ij} \right\rvert.
\end{align*}
From {\bf A\ref{a7}}, we have $\|\psi\|_\infty \leq C_\psi$. Also $\E[\psi(\x_i^\top \beta_0, \z_i)\x_{ij}]=0$ by \eqref{eq:estimating_eq}. Thus $\frac{1}{n}\sum_{i=1}^n\psi(\x_i^\top \beta_0, \z_i)\x_{ij}$ is mean-zero sub-gaussian with a parameter $c C_\psi K_X/\sqrt{n}$ for an absolute constant $c>0$. Thus from Lemma \ref{lem:sub_gaussian_max}, $ \|\psi_n(\beta_0)\|_\infty= O_p(\sqrt{\log p/n}).$

\begin{proof}[Proof of Lemma \ref{lem:A.9}]
First we establish the following inequality. For any $\tilde{s} \geq 1$, there exists $c_0>0$ which do not depend on dimensions $(n,p)$ such that 
\begin{align}\label{eq:lema7-bound1}
    \P\left( \sup_{u \in \B_0 (\tilde{s}) \cap \B_2(1)} |u^\top E u|  \geq c_0 \sqrt{\tilde{s} \frac{\log p}{n}}\right) \leq c_2 /p^{\tilde{s}},
\end{align}
holds where $c_2$ is an absolute constant.

Since for any unit vector $u \in \R^p$ and $i$, $\E[e_i (\x_i^\top u)^2] =0$ and $\E[|e_i(\x_i^\top u)^2|^k]^{1/k} \leq c_e \E[(\x_i^\top u)^{2k}]^{1/k} \leq 2c_eK_X^2 k$, $\forall k\geq 1$, $e_i(\x_i^\top u)^2$ is mean-zero sub-exponential whose parameter is $cc_eK_X^2$ for an absolute constant $c$. From Bernstein's inequality, for every $t\geq 0$, we have
\begin{align}\label{eq:lema7-bernstein}
    \P(|u^\top  E u| \geq tc_eK_X^2) \leq \exp(-c'n (t^2 \wedge t) ),
\end{align}
where $c'>0$ is an absolute constant.
Note,
\begin{align*}
    \B_0 (\tilde{s}) \cap \B_2(1)  &= \bigcup_{k=0}^{\tilde{s}}  \{ v \in \B_2(1); \|v\|_0 = k\}\\
    & = \bigcup_{k=0}^{\tilde{s}} \bigcup_{S; |S| = k}  \{ v\in \B_2(1); \textnormal{supp}(v) = S\}.
\end{align*}
Taking a union bound,
\begin{align*}
   \P( \sup_{u \in \B_0 (\tilde{s}) \cap \B_2(1)} |u^\top E u|  \geq t c_eK_X^2) \leq  \sum_{k=0}^{\tilde{s}} \sum_{S; |S| = k}P(\|E_{S,S}\|_2 \geq t c_eK_X^2),
\end{align*}
where $E_{S,S}$ is a sub-matrix of $E$ supported on $S$.
Letting $\mathcal{N}_\epsilon$ is an $\epsilon$-net of the sphere $\mathcal{S}^{|S|-1}$, we have
\begin{align*}
    \|E_{S,S}\|_2 \leq \frac{1}{1-2\epsilon} \sup_{v \in \mathcal{N}_\epsilon} |v^\top  E_{S,S} v|
\end{align*}
by the covering argument (e.g., \citealp{Vershynin2018-xl}). Take $\epsilon = 1/4$. Then, 
\begin{align*}
    \P( \sup_{u \in \B_0 (\tilde{s}) \cap \B_2(1)} |u^\top E u|  \geq t c_eK_X^2) &\leq  \sum_{k=0}^{\tilde{s}} \binom{p}{k} 9^{k}  P(|v^\top  E_{S,S} v| \geq tc_eK_X^2/2) \\
    &\leq   2  \exp(-c''n (t \wedge t^2 )+\tilde{s}\log (9p)),
\end{align*}
where we use the bounds $|\mathcal{N}_{1/4}|\leq 9^{|S|}$, $\binom{p}{k} \leq p^k$, and \eqref{eq:lema7-bernstein}, and $c''$ is a universal constant. Taking $t^2 = 2\tilde{s}\log (9p) /c''n$, we have
\begin{align*}
    \P( \sup_{u \in \B_0 (\tilde{s}) \cap \B_2(1)} |u^\top E u|  \geq c_eK_X^2 \sqrt{2\tilde{s}\frac{\log 9p}{n}})\leq 2/(9p^{\tilde{s}}),
\end{align*}
given a sample size condition $n \geq 2\tilde{s} \log 9p / c''$. The we obtain the inequality \eqref{eq:lema7-bound1} with $c_0 = 4c_e K_X^2$ and $c_2 = 2/9$, where we use the inequality $p^5 \geq 9p$ for $p \geq 2$.

Now we show that on the event that 
\begin{align}\label{eq:l0bound}
     \sup_{u \in \B_0 (s+s') \cap \B_2(1)} |u^\top E u|  \leq c_0 \sqrt{(s+s') \frac{\log p}{n}},
\end{align}
we have
\begin{align}\label{eq:l1bound}
      \sup_{\substack{u \in \B_1(\sqrt{s}) \cap \B_2(1), \\ v \in \B_1(\sqrt{s'}) \cap \B_2(1)}} |u^\top E v|  \leq c_1 \sqrt{(s+s') \frac{\log p}{n}}
\end{align}
where $c_1$ is a multiple of $c_0$.

From Lemma 11 in \citet{Loh2012-tl}, we have
\begin{align*}
    \B_1(\sqrt{s}) \cap \B_2(1) \subseteq 3 \overline{ \textnormal{conv}(\B_0(s) \cap \B_2(1))}
\end{align*}
where $\textnormal{conv}(D)$ denotes a convex hull of $D \subseteq \R^p$. Using this Lemma, for any $u \in  \B_1(\sqrt{s}) \cap \B_2(1)$ and $v \in  \B_1(\sqrt{s'}) \cap \B_2(1)$ , we have the following representation
\begin{align*}
    u = \sum_{i=1}^p \alpha_i u_i \quad \mbox{and} \quad  v = \sum_{j=1}^p \beta_j v_j
\end{align*}
for $\alpha_i\geq 0, \beta_j\geq 0, u_i, v_j $ such that $\sum_i \alpha_i = \sum_j \beta_j = 1$, $u_i \in \B_0(s) \cap \B_2(3) $ and $v_j \in \B_0(s') \cap \B_2(3)$, $\forall i,j$. Then,
\begin{align*}
    u^\top  E v = (\sum_{i=1}^p \alpha_i u_i )^\top  E (\sum_{j=1}^p \beta_j v_j ) = \sum_{i,j} 9\alpha_i \beta_j \tilde{u}_i^\top  E \tilde{v}_j
\end{align*}
for $\tilde{u}_i,\tilde{v}_j \in \B_0(s+s') \cap \B_2(1), \forall i,j$.
Since,
\begin{align*}
    |\tilde{u}_i^\top  E \tilde{v}_j| &\leq \frac{1}{2 }\left\lbrace |(\tilde{u}_i+\tilde{v}_j)^\top  E (\tilde{u}_i+\tilde{v}_j) | + |\tilde{u}_i^\top  E \tilde{u}_i | +|\tilde{v}_j^\top  E \tilde{v}_j | \right\rbrace, 
\end{align*}
by the basic inequality $2x^\top E y= (x+y)^\top E (x+y)-x^\top Ex - y^\top Ey $ for any $x,y \in \R^p$, we have
\begin{align*}
    \sup_{\substack{u \in \B_1(\sqrt{s}) \cap \B_2(1), \\ v \in \B_1(\sqrt{s'}) \cap \B_2(1)}} |u^\top E v|  &\leq \sum_{i,j}  \frac{27}{2 }(\alpha_i \beta_j) \sup_{u \in \B_0 (s+s') \cap \B_2(1)} |u^\top E u|\\
    &\leq \sum_{i,j} (\alpha_i \beta_j) \left( \frac{27}{2}c_0 \sqrt{(s+s')  \frac{\log p}{n}} \right) = c_1 \sqrt{(s+s') \frac{\log p}{n} }
\end{align*}
where for the second inequality we use \eqref{eq:l0bound} and $c_1 = 27c_0/2$. Thus \eqref{eq:l1bound} holds with probability at least $1-c_2/p^{s+s'}$. 
\end{proof}

\subsection{Construction of an Approximate Inverse of Fisher Information Matrix Using Node-Wise Regression}
\label{supp_sec:approx_Theta}
First we let $W(\beta) := \textnormal{diag}(\{\psi'_I(\x_i^\top \beta)\}_{i=1}^n)$. We note the square root of $W(\beta)$ exists since $\psi_I(t) \geq 0$ for all $t$. Following the node-wise lasso construction in \citet{Van_de_Geer2014-if}, we define
\begin{align*}
\widehat{\gamma}_j &:= \argmin_{\gamma \in \R^p }\frac{1}{2n}\|W(\widehat{\beta})^{1/2}\X_j -  W(\widehat{\beta})^{1/2}\X_{-j}\gamma\|_2 +  \lambda_j \|\gamma\|_1\\
\widehat{\tau}_j^2 &:= \|W(\widehat{\beta})^{1/2}\X_j -  W(\widehat{\beta})^{1/2}\X_{-j}\widehat{\gamma}\|_2^2/n + \lambda_j \|\widehat{\gamma}_j\|_1.
\end{align*}
We construct $\widehat{\Theta}(\psi)$ by taking  $\widehat{\Theta}(\psi)_j^\top  :=\widehat{\tau}_j^{-2} [-\widehat{\gamma}_{j,1},\dots,1,-\widehat{\gamma}_{j,p}] \in \R^{1\times p}.$

\begin{lemma}[Theorem 3.2 in \citealp{Van_de_Geer2014-if}]\label{lem:4.3} Assume {\bf A\ref{a1'}', {\bf A\ref{a3}}, A\ref{a6}-A\ref{a7}} and $\lambda_j \asymp \sqrt{\log p/n}$ for all $j$. In addition we assume there exists $C_X>0$ such that $\|\x_i\|_\infty \leq C_X$ a.s. for all $i$. Then for any $j \in \{1,\dots,p\}$, we have
\begin{align*}
\|\widehat{\Theta}_j(\psi) - \Theta_j(\psi)\|_1 &= o_p(1/\sqrt{\log p}), \qquad
\|\widehat{\Theta}_j(\psi) - \Theta_j(\psi)\|_2 = o_p(n^{-1/4}).
\end{align*}
\end{lemma}

\begin{proof}
The result follows by checking the conditions of Theorem 3.2 in \citet{Van_de_Geer2014-if}. \end{proof}

% Note: in this sample, the section number is hard-coded in. Following
% proper LaTeX conventions, it should properly be coded as a reference:

\vskip 0.2in
\bibliography{references}

\end{document}